\documentclass[draftclsnofoot,onecolumn]{IEEEtran}
\usepackage{amsthm,amsmath,amssymb,cite}
\usepackage{eucal}
\usepackage{xifthen}
\usepackage{mathtools}
\usepackage{enumerate}
\usepackage{microtype}
\usepackage{xspace}
\usepackage{bm}
\usepackage{fancyhdr}
\usepackage{lastpage}
\usepackage[center]{caption}
\usepackage{xcolor}
\allowdisplaybreaks

\newcommand{\mbs}[1]{\bm{#1}}
\newcommand{\vect}[1]{{\lowercase{\mbs{#1}}}}
\newcommand{\mat}[1]{{\uppercase{\mbs{#1}}}}

\renewcommand{\Bmatrix}[1]{\begin{bmatrix}#1\end{bmatrix}}
\newcommand{\Pmatrix}[1]{\begin{array}{ll}#1\end{array}}

\newcommand{\T}{{\scriptscriptstyle\mathsf{T}}}

\renewcommand{\Re}[1][]{\ifthenelse{\isempty{#1}}{\operatorname{Re}}{\operatorname{Re}\left(#1\right)}}
\renewcommand{\Im}[1][]{\ifthenelse{\isempty{#1}}{\operatorname{Im}}{\operatorname{Im}\left(#1\right)}}

\newcommand{\vv}{\vect{v}}

\newcommand{\yv}{\vect{y}}

\newcommand{\Am}{\mat{a}}
\newcommand{\Bm}{\mat{b}}
\newcommand{\Cm}{\mat{c}}
\newcommand{\Dm}{\mat{d}}

\newcommand{\Hm}{\mat{h}}

\newcommand{\Pm}{\mat{p}}
\newcommand{\Qm}{\mat{q}}

\newcommand{\Vm}{\mat{V}}

\newcommand{\Xm}{\mat{x}}
\newcommand{\Ym}{\mat{y}}
\newcommand{\Zm}{\mat{z}}

\newcommand{\Ac}{{\mathcal A}}
\newcommand{\Bc}{{\mathcal B}}
\newcommand{\Cc}{{\mathcal C}}

\newcommand{\Ec}{{\mathcal E}}

\newcommand{\Gc}{{\mathcal G}}
\newcommand{\Hc}{{\mathcal H}}
\newcommand{\Ic}{{\mathcal I}}
\newcommand{\Jc}{{\mathcal J}}
\newcommand{\Kc}{{\mathcal K}}

\newcommand{\Pc}{{\mathcal P}}

\newcommand{\Rc}{{\mathcal R}}
\newcommand{\Sc}{{\mathcal S}}
\newcommand{\Tc}{{\mathcal T}}
\newcommand{\Uc}{{\mathcal U}}
\newcommand{\Wc}{{\mathcal W}}
\newcommand{\Vc}{{\mathcal V}}
\newcommand{\Xc}{{\mathcal X}}

\newcommand{\CC}{\mathbb{C}}

\newcommand{\Id}{\mat{\mathrm{I}}}

\newcommand{\CN}[1][]{\ifthenelse{\isempty{#1}}{\mathcal{N}_{\mathbb{C}}}{\mathcal{N}_{\mathbb{C}}\left(#1\right)}}

\renewcommand{\P}[1][]{\ifthenelse{\isempty{#1}}{\mathbb{P}}{\mathbb{P}\left(#1\right)}}
\newcommand{\E}[1][]{\ifthenelse{\isempty{#1}}{\mathbb{E}}{\mathbb{E}\left(#1\right)}}
\renewcommand{\det}[1][]{\ifthenelse{\isempty{#1}}{\mathrm{det}}{\mathrm{det}\left(#1\right)}}
\newcommand{\trace}[1][]{\ifthenelse{\isempty{#1}}{\mathrm{tr}}{\mathrm{tr}\left(#1\right)}}
\newcommand{\rank}[1][]{\ifthenelse{\isempty{#1}}{\mathrm{rank}}{\mathrm{rank}\left(#1\right)}}
\newcommand{\diag}[1][]{\ifthenelse{\isempty{#1}}{\mathrm{diag}}{\mathrm{diag}\left(#1\right)}}

\DeclarePairedDelimiter\abs{\lvert}{\rvert}
\DeclarePairedDelimiter\Abs{\lvert}{\rvert^2}

\newcommand{\defeq}{\triangleq}

\newtheorem{remark}{Remark}%[section]
\newtheorem{definition}{Definition}%[section]
\newtheorem{theorem}{Theorem}%[section]
\newtheorem{example}{Example}%[section]
\newtheorem{corollary}{Corollary}%[section]
\DeclareMathAlphabet{\mathcal}{OMS}{cmsy}{m}{n}
\newcommand{\sym}{\mathrm{sym}}

\usepackage{graphicx}
\usepackage[font=small]{caption}
\usepackage{subcaption}

\newcommand{\DS}{\mathrm{desired~signal}}
\newcommand{\AI}{\mathrm{aligned~interferences}}
\newcommand{\idx}{\mathrm{idx}}
\newcommand{\nn}{\nonumber}
\newcommand{\IC}{\mathrm{\bf IC}}

\begin{document}
\title{Topological Interference Management with Transmitter Cooperation}
\author{
\authorblockN{Xinping Yi, \emph{Student Member, IEEE},  David Gesbert, \emph{Fellow, IEEE}}
\thanks{The present work was carried out within the framework of Celtic-Plus SHARING project.}
\thanks{X.~Yi and D. Gesbert are with the Mobile Communications Dept., EURECOM, 06560 Sophia Antipolis, France (email: \{xinping.yi, david.gesbert\}@eurecom.fr).}
\thanks{This work was presented in part at ISIT 2014, Honolulu, HI, USA \cite{Yi:2014ISIT}.}
}

\maketitle

\begin{abstract}
Interference networks with no channel state information at the transmitter (CSIT) except for the knowledge of the connectivity graph have been recently studied under the topological interference management (TIM) framework. In this paper, we consider a similar problem with topological knowledge but in a distributed broadcast channel setting, i.e. a network where transmitter cooperation is enabled. We show that the topological information can also be exploited in this case to strictly improve the degrees of freedom (DoF) as long as the network is not fully connected, which is a reasonable assumption in practice.
Achievability schemes based on selective graph coloring, interference alignment, and hypergraph covering, are proposed. Together with outer bounds built upon generator sequence, the concept of compound channel settings, and the relation to index coding, we characterize the symmetric DoF for so-called regular networks with constant number of interfering links, and identify the sufficient and/or necessary conditions for the arbitrary network topologies to achieve a certain amount of symmetric DoF.
\end{abstract}

\section{Introduction}
The advancing interference management techniques have sharpened our understanding in the fundamental limits (e.g., channel capacity) of wireless networks with interference.
The degrees of freedom (DoF) characterization serves as the first-order capacity approximation for wireless networks, by which the obtained insights can be transferred to practical scenarios. The DoF indicates the system throughput scaling with the signal-to-noise ratio (SNR) in the high SNR regime. Although the DoF as a figure of merit has limitations \cite{Lozano:Coop}, it has proved useful in understanding the fundamental limits of several cooperative communication protocols, such as interference alignment (IA) \cite{Jafar:IA} and network MIMO \cite{Gesbert:2010JSAC} among many others. A common feature behind much of the analysis of cooperation benefits in either interference channels (IC) or broadcast channels (BC) has been the availability of instantaneous channel state information at the transmitters (CSIT), with exceptions dealing with so-called limited feedback schemes. Nevertheless, most efforts on limited \cite{Jindal:2006,Love:2008,Caire:2010}, imperfect \cite{Lapidoth:2006, Caire:2010}, or delayed feedback settings \cite{MAT,Vaze:2012IC,SalmanDelayed, Yang:2013MISO,Gou:2012Mixed,Yi:2014MIMO}, among others \cite{Jafar:BIA, Alternating, Lee:2012, Slock:2012}, rely on the assumption that the transmitters are endowed with an instantaneous form of channel information whose coherence time is similar to that of the actual fading channels, so that a good fraction or the totality of the DoF achieved with perfect CSIT can be obtained.  Such an assumption is hard to realize in many practical scenarios, such as cellular networks \cite{Caire:2010Rethinking}. Conversely, it has been reported in \cite{Jafar:noCSIT,Guo:2012noCSIT,Huang:2012noCSIT,Vaze:noCSIT} that a substantial DoF gain cannot be realized in IC or BC scenario without CSIT. A closer examination of these pessimistic results however reveals that many of the considered networks are fully connected, in that any transmitter interferes with any non-intended receiver in the network.

Owing to the nodes' random placement, the fact that power decays fast with distance, the existence of obstacles, and local shadowing effects, we may argue that certain interference links are unavoidably much weaker than others, suggesting the use of a partially-connected graph to model, at least approximately, the network topology.  An interesting question then arises as to whether the partial connectivity could be leveraged to allow the use of some relaxed form of CSIT while still achieving a substantial DoF performance. In particular the exploitation of topological information, simply indicating which of the interfering links are weak enough to be approximated by zero interference and which links are too strong to do so, is of great practical interest. The evidence that the topological information is beneficial can be traced back to \cite{graph_scheduling}, where some local topological information was exploited to improve network performance by some coloring schemes such as ``coded set scheduling''.

Most recently, this question was intensively addressed in \cite{Jafar:CBIA,Jafar:2013TIM,Avestimehr:2013TIM,Jafar:1D,Gou:2013TIM,Sun:2013TIM,Sezgin:TIM,Geng:TIM}, in the context of the interference channel and X channel with topology information, and focusing on the symmetric DoF.
These different topological interference management (TIM) approaches arrive at a common conclusion that the symmetric DoF can be significantly improved with the sole topology information, {\em provided that the network is partially connected}. In \cite{Jafar:2013TIM}, the TIM problem is bridged with the index coding problem \cite{ISCOD,Index2011,Jafar:2012Index,Sun:2013Index,Kim:2013Index,IndexLP}, stating that the optimal solution to the latter is the outer bound of the former, and the linear solution to the former is automatically transferrable to the latter. The ensuing extension in \cite{Sun:2013Index} that attacks the TIM problem from an index coding perspective, covers a wider class of network topologies, partly settling the problem for the sparse networks with each receiver interfered by at most two interfering links.

Given such promising results, a logical question is whether the TIM framework can somehow be exploited in the context of an interference network where a message exchange mechanism between transmitters pre-exists. For instance, in future LTE-A cellular networks, a backhaul routing mechanism ensures that base stations selected to cooperate under the
coordinated multi-point (CoMP) framework receive a copy of the messages to be transmitted. 
With perfect instantaneous CSIT, the benefit of transmitter cooperation was investigated in fully connected IC \cite{CoMP2012} and partially connected IC \cite{ICCoMP2012}.
Still, the exchange of timely CSI is challenging due to the rapid obsolescence of instantaneous CSI and the latency of backhaul signaling links. In this case, a broadcast channel over distributed transmitters (a.k.a.~network MIMO) ensues, with a lack of instantaneous CSIT. The problem raised by this paper concerns the use of topology information in this setting. We follow the same strategy as \cite{Jafar:2013TIM,Avestimehr:2013TIM} in targeting the symmetric DoF as a simple figure of merit. By resorting to interference avoidance and alignment techniques, we characterize the achievable and/or optimal symmetric DoF of the distributed BC with topology information in several scenarios of interest.

More specifically, our contributions are organized as follows:
\begin{itemize}
\item A graph theoretic perspective will be provided in Section III, in which we propose an interference avoidance approach built upon fractional selective graph coloring over the square of line graph of the original network topology.
In doing so, the optimal symmetric DoF of three-cell networks with all possible topologies is determined, by a new outer bound on the basis of the concept of generator sequence.
\item An interference alignment perspective will be also offered in Section IV by introducing an alignment-feasible graph to show the feasibility of interference alignment between any two messages.
The sufficient conditions for arbitrary network topologies to achieve a certain amount of symmetric DoF are identified with this graph, by which we also identify the achievable symmetric DoF of so-called {\em regular} networks (i.e., network topologies with same number of interfering links at all transmitters/receivers). Further, the optimality for the Wyner-type regular networks (i.e., with only one interfering link) is characterized with the aid of an outer bound based on an application of compound settings.
Lastly, the above alignment feasibility condition is generalized to arbitrary number of messages, leading us to a construction of a hypergraph, by which achievable symmetric DoF of arbitrary network topologies are consequently established via hypergraph covering.
\item In Section IV, we also bridge our problem to index coding problems, letting the outer bounds of the latter serve our problem as well, by which we identify the sufficient and necessary condition when time division is symmetric DoF optimal.
\end{itemize}

\underline{\bf Notation}: 
Throughout this paper, we define $\Kc \defeq \{1,2,\dots,K\}$, and $[n] \defeq \{1,2,\dots,n\}$ for any integer $n$. Let $A$, $\Ac$, and $\Am$ represent a variable, a set, and a matrix/vector, respectively. In addition, $\Ac^c$ is the complementary set of $\Ac$, and $\abs{\Ac}$ is the cardinality of the set $\Ac$. $\Am_{ij}$ or $[\Am]_{ij}$ presents the $ij$-th entry of the matrix $\Am$, and $\Am_{i}$ or $[\Am]_i$ is the $i$-th row of $\Am$. $A_{\Sc} \defeq \{A_i, i \in \Sc\}$, $\Ac_{\Sc} \defeq\cup_{i \in \Sc} \Ac_i$, and $\Am_{\Sc}$ denotes the submatrix of $\Am$ with the rows out of $\Sc$ removed. Define $\Ac \backslash a \defeq \{x| x \in \Ac, x \neq a\}$ and $\Ac_1 \backslash \Ac_2 \defeq \{x | x \in \Ac_1, x \notin \Ac_2\}$. We use $\Id_M$ to denote an $M \times M$ identity matrix where the dimension is omitted whenever the confusion is not probable. $\mathbf{1}(\cdot)$ is the indicator function with values 1 when the parameter is true and 0 otherwise. $O(\cdot)$ follows the standard Landau notation. Logarithms are in base 2.

\section{System Model}

\subsection{Channel Model}
We consider a $K$-cell partially connected cellular network, in which each transmitter (e.g. base station)  is equipped with one antenna and serves one single-antenna receiver (e.g., user).
This cellular network can be modeled by a partially connected interference channel. The received signal for Receiver $j$ at time instant $t$ can be modeled by
\begin{align}
Y_j(t) = \sum_{i \in \Tc_j} h_{ji}(t) X_i(t) + Z_j(t)
\end{align}
where $h_{ji}(t)$ is the channel coefficient between Transmitter $i$ and Receiver $j$ at time instant $t$ and the nonzero channel coefficients drawn from a continuous distribution are independent and identically distributed (i.i.d.), the transmitted signal $X_i(t)$ is subject to the average power constraint, i.e., $\E[\Abs{X_i(t)}] \le P$, with $P$ being the average transmit power, and $Z_j(t)$ is the additive white Gaussian noise with zero-mean and unit-variance and is independent of transmitted signals and channel coefficients.

We denote by $\Tc_k$ the transmit set containing the indices of transmitters that are {\em connected} to Receiver $k$, and by $\Rc_k$ the receive set consisting of the indices of receivers that are {\em connected} to Transmitter $k$, for $k\in \Kc\defeq \{1,2,\dots,K\}$.  In practice, the partial connectivity may be modeled by taking those interference links that are ``weak enough" (due to distance and/or shadowing) to zero. For instance  in \cite{Jafar:2013TIM}, a reasonable model is suggested whereby a link is disconnected if the received signal power falls below the effective noise floor. However, other models maybe envisioned and the study of how robust the derived schemes are with respect to modeling errors is an open problem beyond the scope of this paper.

Conforming with TIM framework, the actual channel realizations are not available at the transmitters, yet the network topology (i.e., $\Tc_k, \Rc_k, \forall k$) is known by all transmitters and receivers. A typical transmitter cooperation is enabled in the form of message sharing, where every transmitter is endowed the messages desired by its connected receivers, i.e., Transmitter $k$ has access to a subset of messages $W_{\Rc_k}$, where $W_j$ $(j \in \Rc_k)$ denotes the message desired by Receiver $j$. We refer hereafter to TIM problem with transmitter cooperation as ``TIM-CoMP'' problem. Each message may originate from multiple transmitters but is intended for one unique receiver. As such, the so-called direct links in TIM settings are not required to be present here. We consider a block fading channel, where the channel coefficients stay constant during a coherence time $\tau_c$ but vary to independent realizations in the next coherence time. 
The coherence time is $\tau_c=1$ by default unless otherwise specified.  For channel coefficients and transmitted signals, the time index $t$ is omitted during the coherence time for the sake of brevity. The network topology is fixed throughout the communication.

While message sharing creates the opportunity of transmitter cooperation, it also imposes some challenges. For the multiple-unicast TIM problem in partially connected IC or X networks \cite{Jafar:CBIA,Jafar:2013TIM,Avestimehr:2013TIM}, each message has a unique source and a unique destination that are determined \emph{a priori} such that the desired and interfering links are known. By contrast, with transmitter cooperation, the message can be sent from \emph{any} source that has access to this message. Consequently, the approaches developed for IC and X networks cannot be directly applied here, as the desired and interference links are not able to be predetermined.

For notational convenience, we define $\Hc \defeq \{h_{ji},\forall~i,j\}$ as the ensemble of channel coefficients, and denote by $\Gc$ the network topology known by all transmitters and receivers.

\subsection{Definitions}
Throughput this paper, we treat partially connected networks as bipartite graphs $\Gc=(\Uc,\Vc,\Ec)$, where the transmitters and receivers are two sets of vertices, denoted by $\Uc$ and $\Vc$, and the connectivities between transmitters and receivers are represented as edges, e.g., $e_{ij} \in \Ec$ where $i \in \Uc$ and $j \in \Vc$. 
\begin{definition}[Topology Matrix]
For a network topology, the topology matrix $\Bm$ is defined as
  \begin{align}
    [\Bm]_{ji} = \left\{ \Pmatrix{1, & e_{ij} \in \Ec\\ 0, & \text{otherwise}} \right..
  \end{align}
\end{definition}
\begin{definition} [Special Network Topologies]
A $(K,d)$-{\rm \bf regular network} refers to the $K$-cell network where each receiver will overhear the signals from the transmitter with the same index as well as the successive $d-1$ ones, i.e., $\Tc_j = \{j,j+1,\dots,j+d-1\}$, and any network whose topology graph is {\em similar} to this one. The network topologies except regular networks are referred to as irregular networks. One typical example of irregular networks is the {\rm \bf triangular network}, which refers to a category of cellular networks with $\Tc_j = \{1,\dots,j\}$ (i.e., topology matrix is lower triangular) or $\Tc_j = \{j,\dots,K\}$ (i.e., topology matrix is an upper triangular), as well as those whose topology graphs are similar to either one.
\end{definition}

A rate tuple $(R_1,\dots,R_K)$ is said to be {\em achievable} to TIM-CoMP problems, if these exists a $(2^{nR_1},\dots,2^{nR_K},n)$ code scheme including the following elements:
\begin{itemize}
\item $K$ message sets $\Wc_k \defeq [1:2^{nR_k}]$, from which the message $W_k$ is uniformly chosen, $\forall~k \in \Kc$;
\item one encoding function for Transmitter~$i$ $(\forall~i \in \Kc)$:
\begin{align}
X_i(t) = f_{i,t} \left( W_{\Rc_i}, \ \Gc \right),
\end{align}
where only a subset of messages $W_{\Rc_i}$ is available at Transmitter $i$ for encoding;
\item one decoding function for Receiver~$j$ $(\forall~j \in \Kc)$:
\begin{align}
\hat{W}_j = g_j \left(Y_j^n, \ \Hc^n, \ \Gc \right),
\end{align}
\end{itemize}
such that the average decoding error probability is vanishing as the code length $n$ tends to infinity. The capacity region $\Cc$ is defined as the set of all achievable rate tuples.

In this work, we follow the strategy of \cite{Jafar:CBIA,Jafar:2013TIM,Avestimehr:2013TIM,Jafar:2012Index,Sun:2013Index,Gou:2013TIM} and set the symmetric DoF (i.e., the DoF which can be achieved by all users simultaneously) as our main figure of merit.
\begin{definition} [Symmetric DoF]
\begin{align}
d_{\sym} = \limsup_{P \to \infty}  \sup_{(R_{\sym},\dots,R_{\sym}) \in \Cc} \frac{R_{\sym}}{\log P}
\end{align}
where $P$ is the average transmit power.
\end{definition}

\section{A Graph Theoretical Perspective}
As a baseline, an interference avoidance approach (also known as orthogonal access \cite{Jafar:1D}) is first presented in Theorem~\ref{theorem:IAvoid} for general topologies with the aid of graph coloring, followed by an outer bound in Theorem~\ref{theorem:Gen} built upon the concept of generator, by which we are able to characterize the optimality for three-cell networks with arbitrary topologies and triangular networks.

\subsection{Interference Avoidance via Selective Graph Coloring}

Before proceeding further, we introduce the following definition generalized from the standard graph coloring. Some basic graph theoretic definitions are recalled in Appendix A.

\begin{definition} [Fractional Selective Graph Coloring]
\label{def:dis-color}
Consider an undirected graph $\Gc=(\Vc,\Ec)$ with a vertex partition $\mathbb{V}=\{\Vc_1,\Vc_2,\dots,\Vc_p\}$ where $\cup_{i=1}^p \Vc_i=\Vc$ and $\Vc_i \cap \Vc_j = \emptyset$, $\forall~i\ne j$. The portion $\Vc_i$ $(i \in [p] \defeq \{1,2,\dots,p\})$ is called a cluster.
A graph with the partition $\mathbb{V}$ is said to be {\bf selectively $n:m$-colorable}, if
\begin{itemize}
\item each cluster $\Vc_i$ $(\forall~i)$ is assigned a set of $m$ colors drawn from a palette of $n$ colors, no matter which vertex in the cluster receives;
\item any two adjacent vertices have no colors in common.
\end{itemize}
Denote by $s\chi_f (\Gc,\mathbb{V})$ the {\bf fractional selective chromatic number} of the above selective coloring over the graph $\Gc$ with the partition $\mathbb{V}$, which is defined as
\begin{align}
s\chi_f(\Gc,\mathbb{V}) = \lim_{m \to \infty} \frac{s\chi_m(\Gc,\mathbb{V})}{m} = \inf_m \frac{s\chi_m(\Gc,\mathbb{V})}{m}
\end{align}
where $s\chi_m(\Gc,\mathbb{V})$ is the minimum $n$ for the selective $n:m$-coloring associated with the partition $\mathbb{V}$.
\end{definition}

\begin{remark}\normalfont
If $m=1$, fractional selective graph coloring boils down to standard selective graph coloring (a.k.a.~partition coloring) \cite{partcoloring,selcoloring}. If $\abs{\Vc_i}=1$ $(\forall~i \in [p])$, then fractional selective graph coloring will be reduced to standard fractional graph coloring.
\end{remark}

\begin{theorem} [Achievable DoF via Graph Coloring]
\label{theorem:IAvoid}
For TIM-CoMP problems with arbitrary topologies, the symmetric DoF
\begin{align}
d_{\sym} = \frac{1}{s\chi_f(\Gc_{e}^2, \mathbb{V}_e)}
\end{align}
can be achieved by interference avoidance (i.e., orthogonal access) built upon fractional selective graph coloring, where
\begin{itemize}
\item $\Gc_e$: the line graph of network topology $\Gc$, where the vertices in $\Gc_e$ are edges of $\Gc$;
\item $\mathbb{V}_e$: a vertex partition of $\Gc_e$, and specifically vertices in $\Gc_e$ whose corresponding edges in $\Gc$ have a common receiver form a cluster;
\item $\Gc_{e}^2$: the square of $\Gc_e$, in which any two vertices in $\Gc_e$ with distance no more than 2 are joint with an edge;
\item $s\chi_f$: fractional selective chromatic number as defined in Definition \ref{def:dis-color}.
\end{itemize}
\end{theorem}

\begin{proof}
See Appendix \ref{proof:IAvoid}.
\end{proof}

By connecting the achievable symmetric DoF of TIM-CoMP problem to fractional selective chromatic number, we are able to calculate the former by computing the latter with rich toolboxes developed in graph theory. The connection will be illustrated by the following example whose network topology was studied in~\cite{Avestimehr:2013TIM} (as shown in Fig.~\ref{fig:5}) with no transmitter cooperation.

\begin{example}
\normalfont
For the network topology shown in Fig.~\ref{fig:5}(a), the optimal symmetric DoF value is pessimistically $\frac{1}{3}$ without message sharing \cite{Jafar:2013TIM,Avestimehr:2013TIM,Sun:2013Index}. In contrast, if transmitter cooperation is allowed, the achievable symmetric DoF can be remarkably improved to $\frac{2}{5}$ even with orthogonal access according to Theorem \ref{theorem:IAvoid}.

Without message sharing, the interference avoidance scheme consists in scheduling transmitters to avoid mutual interference.
For instance, by delivering $W_1$, Transmitter 1 will cause interferences to Receivers 2 and 3, and consequently Transmitters 2 and 3 should be deactivated, because $W_2$ and $W_3$ cannot be delivered to Receivers 2 and 3 free of interference.
In contrast, with message sharing, the desired message $W_1$ can be sent either from Transmitter 1 or 4.
Hence, scheduling can be done across links rather than across transmitters.
For instance, if the link Transmitter 4 $\to$ Receiver 1 (denoted by $e_{41}$) is scheduled, the links adjacent to $e_{41}$ (i.e., $e_{11}$, $e_{42}$, and $e_{44}$) as well as the links adjacent to $e_{11}$, $e_{42}$ and $e_{44}$ (i.e., $e_{12},e_{13}$, $e_{22},e_{32}$, $e_{34}$ and $e_{54}$) should not be scheduled, because activating Transmitter 1 will interfere Receiver 1, and Receivers 2 and 4 will overhear interferences from Transmitter 4 such that any delivery from Transmitter 1 or to Receivers 2 and 4 causes mutual interference.
A possible link scheduling is shown in Table~\ref{tab:link-sch}. It can be found that each message is able to be independently delivered twice during five time slots, and hence symmetric DoF of $\frac{2}{5}$ are achievable.
\begin{table} [htb]
 \caption{Link Scheduling }
  \centering
  \begin{tabular}{@{} c|c|c @{}} \hline  \hline
  Slot & Scheduled Links ($e_{ij}$: TX $i$ $\to$ RX $j$) & Delivered Messages \\ \hline \hline
  A & $e_{41}$, $e_{55}$, $e_{66}$ & $W_1,W_5,W_6$ \\ \hline
  B & $e_{12}$, $e_{54}$, $e_{66}$ & $W_2,W_4,W_6$ \\ \hline
  C & $e_{13}$, $e_{54}$ & $W_3,W_4$ \\ \hline
  D & $e_{41}$, $e_{33}$ & $W_1,W_3$ \\ \hline
  E & $e_{12}$, $e_{55}$ & $W_2,W_5$ \\ \hline
   \hline
  \end{tabular}
  \label{tab:link-sch}
 \end{table}

Although the above link scheduling solution provides an achievable scheme for the topology in Fig.~\ref{fig:5}(a), the generalization is best undertaken by reinterpreting the link scheduling into a graph coloring problem, such that the rich graph theoretic toolboxes can be directly utilized to solve our problem.  In what follows, we reinterpret the link scheduling from a fractional selective graph coloring perspective.
To ease presentation, we translate graph edge-coloring into graph vertex-coloring of its line graph.

As shown in Fig.~\ref{fig:5}, we first transform the topology graph $\Gc$ (left) into its line graph $\Gc_e$ (right) and map the links connected to each receiver in $\Gc$ to the vertices in $\Gc_e$. For instance, the four links to Receiver 2 in $\Gc$ are mapped to Vertices $e_{12},e_{22},e_{32},e_{42}$ in $\Gc_e$. Then, we group relevant vertices in $\Gc_e$ as clusters, e.g., Vertices $e_{12},e_{22},e_{32},e_{42}$ in $\Gc_e$ corresponding to the links to Receiver 2 are grouped as one cluster. By now, a clustered-graph is generated with $\mathbb{V}_e=\{ \{e_{11},e_{41}\},\{e_{12},e_{22},e_{32},e_{42}\},\{e_{13},e_{33}\},\{e_{34},e_{44},e_{54}\},\{e_{35},e_{55}\},\{e_{36},e_{66}\}\}$. The selective graph coloring can be performed as follows. For the sake of brevity, the color assignment is performed over the line graph $\Gc_e$ in which any two vertices with distance less than 2 should receive different colors. This is equivalent to assign colors to square of line graph $\Gc_e^2$ where any two adjacent vertices receive distinct colors. For instance, if Vertex $e_{41}$ in $\Gc_e$ receives a color indicated by `A', then Vertices $e_{55}$ and $e_{66}$ can receive the same color, because the distance between any two of them is no less than 2 in $\Gc_e$ and hence any two of them are nonadjacent in $\Gc_e^2$. Try any possible color assignment until we obtain a proper one, where each cluster receives $m$ distinct colors out of total $n$ ones, such that any two vertices with distance less than 2 receive distinct colors. There may exist many proper color assignments. 

\begin{figure}[htb]
 \centering
\includegraphics[width=0.8\columnwidth]{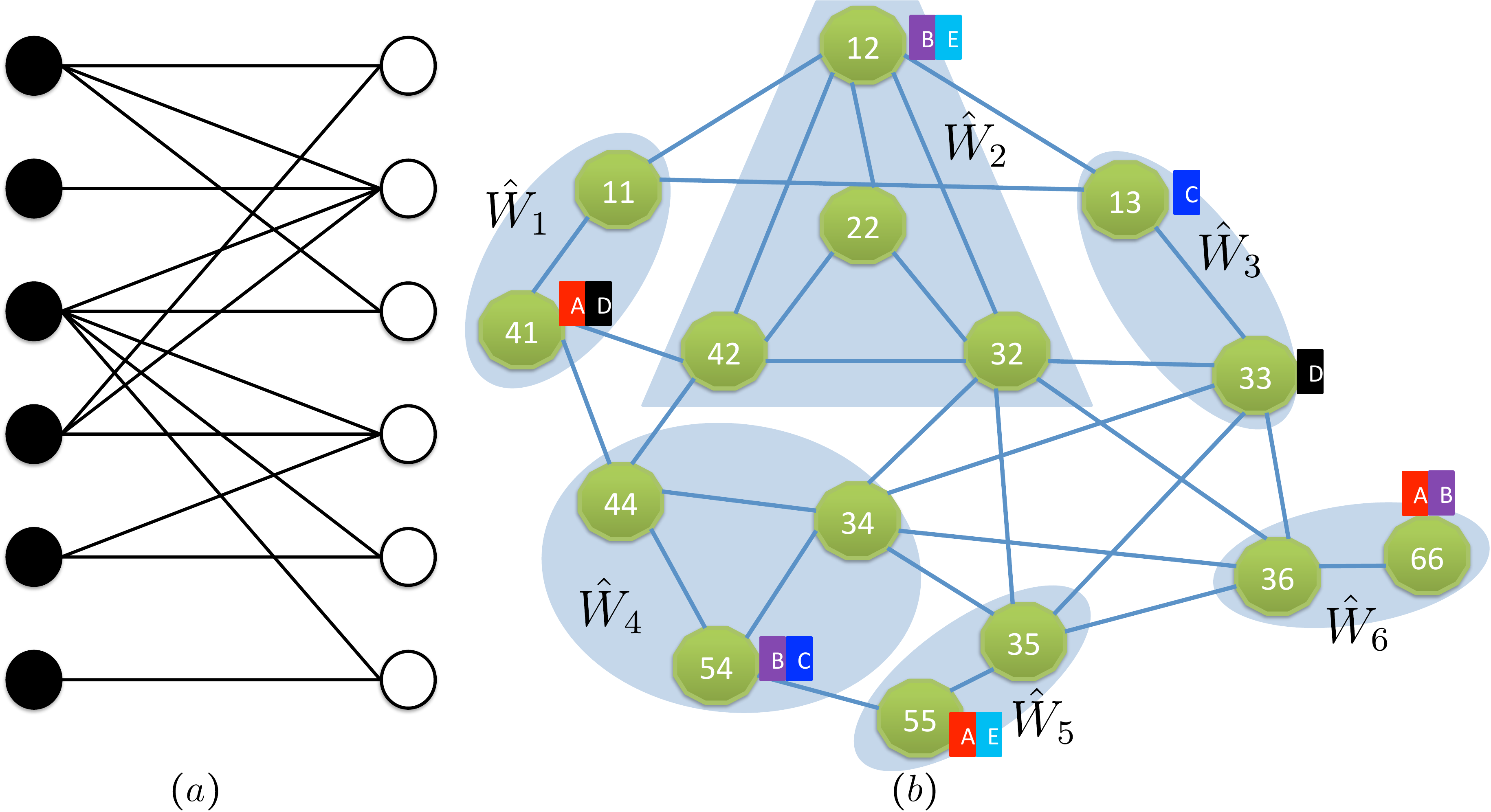}
\caption{An instance of TIM-CoMP problem $(K=6)$. (a) The network topology graph $\Gc$, and (b) its line graph $\Gc_e$. The fractional selective coloring is performed to offer each cluster two out of in total five colors, where any two vertices that receives the same color are set apart with distance no less than 2.}
\label{fig:5}
\end{figure}

The fractional selective chromatic number $s\chi_f(\Gc_{e}^2,\mathbb{V})$ refers to the minimum of $\frac{n}{m}$ among all proper color assignments. In this example, we have $m=2$ and $n=5$. The vertices (i.e., links in $\Gc$) with the same color can be scheduled in the same time slot. Accordingly, each cluster receives two out of five colors means every message is scheduled twice during five time slots, yielding the symmetric DoF of $\frac{2}{5}$. According to this connection between link scheduling and graph coloring, the inverse of the fractional selective chromatic number, i.e., $\frac{1}{s\chi_f(\Gc_{e}^2,\mathbb{V})}$, can serve as the achievable symmetric DoF of TIM-CoMP problems, although its computation is still NP-hard. \hfill$\square$
\end{example}

\subsection{Outer Bound via Generator Sequence}
To see how tight this interference avoidance scheme is, we provide an outer bound based on the concept of generator~\cite{Avestimehr:2013TIM}. 
For simplicity of presentation, we introduce an index function $\mathbf f_{\idx}$, which is defined as $\mathbf f_{\idx}: \Bc \mapsto \{0,1\}^K$, to map the position indicated by $\Bc \subseteq \Kc$ to a $K \times 1$ binary vector with the corresponding position being 1, and 0 otherwise, e.g., $\mathbf f_{\idx}(\{1,3,5\}) = [1 \, 0 \, 1 \, 0 \, 1 \, 0 ]^\T$ with $K=6$. Thus, we have the following definition.

\begin{definition} [Generator Sequence]
\label{def:gen}
Given $\Sc \subseteq \Kc$, a sequence $\{\Ic_0, \Ic_1, \dots, \Ic_S\}$ is called a {\bf generator sequence}, if it is a partition of $\Sc$ (i.e., $\cup_{s=0}^S \Ic_s = \Sc$ and $\Ic_i \cap \Ic_j = \emptyset$, $\forall~i \neq j$), such that
\begin{align}
\Bm_{\Ic_s} \subseteq^{\pm} {rowspan} \left\{ \Bm_{\Ic_0}, \Id_{\Ac_{s}} \right\}, \quad \forall~s=1,\dots,S
\end{align}
where $\Bm_{\Ic}$ is the submatrix of $\Bm$ with rows of indices in $\Ic$ selected, $\Ac_{s} \defeq \{i | [\Bm^\T]_i \cdot \mathbf f_{\idx}(\cup_{r=0}^{s-1} \Ic_r) = \abs{\Rc_i \backslash \Sc^c}\}$ with $[\Bm^\T]_i$ being the $i$-th row of $\Bm^\T$ (i.e., $i$-th column of $\Bm$), and $\Id_{\Ac_{s}}$ denotes a submatrix of $\Id_K$ with the rows in $\Ac_{s}$ selected. $\Am_1 \subseteq^{\pm} rowspan \{\Am_2\}$ is such that two matrices $\Am_1 \in \CC^{m_1 \times n}$ and $\Am_2 \in \CC^{m_2 \times n}$ satisfy $\Am_1=\Cm \Am_2 \Id^{\pm}$, where $\Cm \in \CC^{m_1 \times m_2}$ can be any full rank matrix, $\Id^{\pm}$ is as same as the identity matrix up to the sign of elements. This implies that the row of $\Am_1$ can be represented by the rows of $\Am_2$ with possible difference of signs of elements.
We refer to $\Ic_0$ as the initial generator with regard to $\Sc$, and denote by $\Jc(\Sc)$ all the possible initial generators.
\end{definition}

\begin{theorem} [Outer Bound via Generator Sequence]
\label{theorem:Gen}
The symmetric DoF of the $K$-cell TIM-CoMP problem are upper bounded by
\begin{align}
d_{\sym} \le \min_{\Sc \subseteq \Kc} \min_{\Ic_0 \subseteq \Jc(\Sc)} \frac{|\Ic_0|}{|\Sc|}
\end{align}
where $\Ic_0$ is the initial generator, from which a sequence can be initiated and generated subsequently as defined in Definition~\ref{def:gen}.
\end{theorem}
\begin{proof}
See Appendix \ref{proof:Gen}.
\end{proof}

Roughly speaking, the key of this outer bound is to first properly select a subset of receivers of interest, from which a smaller subset is carefully chosen then as an initial generator, such that statistically equivalent received signals of others can be gradually generated. To obtain a relatively tight bound, it is preferred an initial generator with a small cardinality to generate the rest of sequence with a large cardinality. Intuitively, irregular networks favor this generator sequence outer bound. The more irregular the topology is, the tighter the outer bound is expected to be, because it is likely to start with small initiator and generator a long sequence.
This point will be confirmed by one of the most irregular networks (triangular networks) in Corollary \ref{cor:tri-cell}.
In what follows, we illustrate the identification of a generator sequence for the irregular network studied in Example 1. 
\begin{example}
\normalfont
For the topology in Fig.~\ref{fig:5}(a), we have the transmit sets $\Tc_1=\{1,4\}, \ \Tc_2 = \{1,2,3,4\}, \ \Tc_3=\{1,3\}, \ \Tc_4=\{3,4,5\}, \ \Tc_5=\{3,5\}, \ \Tc_6=\{3,6\}$ and receive sets $\Rc_1=\{1,2,3\}, \ \Rc_2=\{2\}, \ \Rc_3=\{2,3,4,5,6\}, \ \Rc_4=\{1,2,4\}, \ \Rc_5=\{4,5\}, \ \Rc_6=\{6\}$. With the message sharing strategy mentioned earlier, the messages $W_{\Rc_i}$ are accessible at Transmitter $i$.

As symmetric DoF metric is considered, the DoF outer bound regarding any subset of messages serves as one candidate in general.
In what follows, we select a subset of receivers $\Sc=\{1,3,4,5\}$, from which $\{1,4\}$ are chosen as an initial generator, such that statistical equivalent signals at Receivers 3 and 5 can be subsequently generated.
Before proceeding further, we define the following virtual signals
\begin{align}
\tilde{Y}_1^n &\defeq h_1^n X_1^n + h_4^n X_4^n + \tilde{Z}_1^n\\
\tilde{Y}_4^n &\defeq h_3^n X_3^n + h_4^n X_4^n + h_5^n X_5^n + \tilde{Z}_4^n,
\end{align}
where $h_k^n~(k=1,\dots,6)$ is assumed to be independent and identically distributed as $h_{ji}^n$ when there is a strong link between Transmitter $i$ and Receiver $j$, and the noise terms $\{\tilde{Z}_1^n,\tilde{Z}_4^n\}$ are identically distributed  as $Z_j^n$ with zero-mean and unit-variance. Given the fact that the distribution of channel gain is symmetric around zero, it follows that $\{\tilde{Y}_1^n, \tilde{Y}_4^n\}$ are statistically equivalent to $\{Y_1^n, Y_4^n\}$, respectively. From both $\{\tilde{Y}_1^n, \tilde{Y}_4^n\}$ and $\{Y_1^n, Y_4^n\}$, the corresponding messages $\{\hat{W}_1, \hat{W}_4\}$ can be decoded with error probability tends to 0 as $n \to \infty$.

Let us focus on the subset of messages $W_\Sc$, where $W_i$ $(i\in {\Sc}^c=\{2,6\})$ are set to be deterministic. Note that eliminating some messages or setting them to be deterministic does not hurt the maximum achievable rate of remaining messages. Thus, the sum rate associated with the receivers in $\Sc$ can be upper bounded as
\begin{align}
n\sum_{i \in \Sc} R_i &= H(W_{\Sc} | \Hc^n, \Gc)\\
&= I(W_{\Sc}; \tilde{Y}_{1,4}^n | \Hc^n, \Gc) + H(W_{\Sc} | \tilde{Y}_{1,4}^n, \Hc^n, \Gc) \\
&= I(W_{\Sc}; \tilde{Y}_{1,4}^n| \Hc^n, \Gc) + H(W_{1,4} | \tilde{Y}_{1,4}^n, \Hc^n, \Gc) \nn \\ &+ H(W_{\Sc \backslash \{1,4\}} | W_{1,4}, \tilde{Y}_{1,4}^n,  \Hc^n, \Gc) \\
&\le 2n \log P + H(W_{\Sc \backslash \{1,4\}} | W_{1,4}, \tilde{Y}_{1,4}^n, \Hc^n, \Gc) + n \cdot O(1) + n \epsilon_n
\end{align}
where the last inequality is obtained by Fano's inequality, and $n\epsilon_n \defeq 1 + n R P_e^{(n)}$ tends to zero as $n \to \infty$ by the assumption that $\lim_{n \to \infty} P_e^{(n)} = 0$.

Since the transmitted signal $X_i^n$ is encoded from the messages $W_{\Rc_i}$ $(\forall~i)$, it suffices to reproduce $X_4^n$ and $X_5^n$ from $W_1,W_4$ and $W_4,W_5$, respectively, with $W_2,W_6$ switched off (i.e., being set to be deterministic). Thus, we have
\begin{align}
\MoveEqLeft H(W_{\Sc \backslash \{1,4\}} | W_{1,4}, \tilde{Y}_{1,4}^n, \Hc^n, \Gc) \\
&= H(W_{3,5} | W_{1,4}, X_4^n, \tilde{Y}_{1,4}^n, \Hc^n, \Gc) \label{ex-1-1}\\
&=H(W_5 | W_{1,4}, X_4^n, \tilde{Y}_{1,4,5}^n, \Hc^n, \Gc) + H(W_3 | W_{1,4,5}, X_4^n, \tilde{Y}_{1,4}^n, \Hc^n, \Gc) \label{ex-1-2}\\
&\le H(W_5 |  \tilde{Y}_5^n, \Hc^n) + H(W_3 | W_{1,4,5}, X_4^n, X_5^n, \tilde{Y}_{1,4}^n, \Hc^n, \Gc) \label{ex-1-3}\\
&\le n \epsilon_n + H(W_3 | W_{1,4,5}, X_4^n, X_5^n, \tilde{Y}_{1,4}^n, \Hc^n, \Gc) \label{ex-1-4}\\
&= n \epsilon_n + H(W_3 | W_{1,4,5}, X_4^n, X_5^n, \tilde{Y}_{1,3,4}^n, \Hc^n, \Gc) \label{ex-1-5}\\
&\le n \epsilon_n + H(W_3 | \tilde{Y}_3^n, \Hc^n, \Gc) \label{ex-1-6}\\
&\le n \epsilon_n \label{ex-1-7}
\end{align}
where \eqref{ex-1-1} is from the fact that $X_4^n$ is reproducible from $W_{1,4}$, \eqref{ex-1-2} is because of the chain rule of entropy and the fact that $\tilde{Y}_5^n = \tilde{Y}_4^n - h_4^n X_4^n =h_3^n X_3^n + h_5^n X_5^n + \tilde{Z}_4^n$ can be generated from $\tilde{Y}_4^n$ and $X_4^n$, \eqref{ex-1-3} is due to a) removing condition does not reduce entropy, and b) $X_5^n$ can be obtained given the messages $W_{4,5}$, \eqref{ex-1-5} comes from the generator sequence where $\tilde{Y}_3^n = \tilde{Y}_1^n - \tilde{Y}_4^n + h_5^n X_5^n = h_1^n X_1^n - h_3^n X_3^n + \tilde{Z}_1^n - \tilde{Z}_4^n$ can be generated from $\tilde{Y}_{1,4}^n$ and $X_5^n$, \eqref{ex-1-6} is due to removing condition does not decrease entropy, and inequality \eqref{ex-1-4} and the last inequalities are due to Fano's inequality, where $\tilde{Y}_5^n$ and $\tilde{Y}_3^n$ are statistically equivalent to $Y_5^n$ and $Y_3^n$ respectively, with bounded difference of noise variance, such that both $W_5$ and $W_3$ can be decoded respectively with negligible errors.
Hence, we have
\begin{align}
n\sum_{i \in \Sc} R_i  \le 2n \log P + n \cdot O(1) + n \epsilon_n
\end{align}
which leads to one possible outer bound for symmetric DoF
\begin{align}
d_{\sym} \le \frac{1}{2}.
\end{align}

To summarize, we first take $\{1,4\}$ as an initial generator, and generate two statistically equivalent signals $\{\tilde{Y}_1^n,\tilde{Y}_4^n\}$. With the messages $W_1,W_4$, we reconstruct $X_4^n$, and then generate $\tilde{Y}_5^n$ from $\tilde{Y}_4^n$. Finally, $\tilde{Y}_3^n$ can be generated from $\{\tilde{Y}_1^n,\tilde{Y}_4^n\}$ and $X_5^n$ encoded from $W_4,W_5$. As such, the generation sequence is $\{\{1,4\}, \{5\}, \{3\} \}$, initiated from $\Ic_0=\{1,4\}$. With $\Sc=\{1,3,4,5\}$, according to Definition~\ref{def:gen}, we have
\begin{align}
\Bm = \begin{bmatrix} 1 & 1 & 1 & 0 & 0 & 0 \\ 0 & 1 & 0 & 0 & 0 & 0 \\ 0 & 1 & 1 & 1 & 1 & 1\\ 1 & 1 & 0 & 1 & 0 & 0 \\ 0 & 0 & 0 & 1 & 1 & 0 \\ 0 & 0 & 0 & 0 & 0 & 1 \end{bmatrix}^\T,
\quad
\Bm_{\{1,4\}} = \begin{bmatrix} 1 & 0  \\ 0  & 0 \\ 0 & 1 \\1 & 1 \\ 0  & 1 \\ 0  & 0 \end{bmatrix}^\T
\quad
\Bm_5 = \begin{bmatrix} 0  \\ 0 \\ 1 \\ 0 \\ 1 \\ 0 \end{bmatrix}^\T
\quad
\Bm_3 = \begin{bmatrix} 1  \\ 0 \\ 1 \\ 0 \\ 0 \\ 0 \end{bmatrix}^\T
\end{align}
and $\Ac_1=\{4\}$, $\Ac_2=\{4,5\}$, and $\Ac_3=\{1,3,4,5\}$.
It is readily verified that $\Bm_5 \subseteq^{\pm}  rowspan \{\Bm_{\{1,4\}}, \Id_{\Ac_1}\}$ and $\Bm_3 \subseteq^{\pm}  rowspan \{\Bm_{\{1,4\}}, \Id_{\Ac_2}\}$.

One may notice that the above outer bound derivation has common properties as those in \cite{Avestimehr:2013TIM}, the differences however are two-fold: 1) due to transmitter cooperation (i.e., message sharing), the transmitted signal is encoded from multiple messages, instead of the single message in the TIM setting, and 2) when we switch off some messages (e.g., by setting them to be deterministic), we only eliminate them from the message set $\Rc_i$ of $X_i^n$, instead of switching off $X_i^n$ as did in \cite{Avestimehr:2013TIM}.\hfill$\square$
\end{example}

\subsection{The Optimality of Interference Avoidance}
By interference avoidance and the above outer bound, we characterize the optimal symmetric DoF of some special networks below.
\begin{corollary} [Optimal DoF for Three-cell Networks]
\label{cor:3-cell}
The optimal symmetric DoF of the three-cell TIM-CoMP problem can be achieved by interference avoidance (i.e., orthogonal access).
\end{corollary}

\begin{proof}
See Appendix \ref{proof:3-cell}.
\end{proof}

\begin{corollary} [Optimal DoF for Triangular Networks]
\label{cor:tri-cell}
For the $K$-cell triangular networks, the optimal symmetric DoF value of the TIM-CoMP problem is $\frac{1}{K}$.
\end{corollary}

\begin{proof}
See Appendix \ref{proof:tri-cell}.
\end{proof}

\section{An Interference Alignment Perspective}

To gain further improvement, an interference alignment perspective is introduced with the alignment-feasible graph defined in Definition~\ref{def:afg}, by which the sufficient conditions achieving a certain amount of symmetric DoF is identified in Theorem~\ref{theorem:Arb}. Further, by these condition, in Theorem~\ref{theorem:Reg} we identify the achievable symmetric DoF of regular networks. To see the tightness of interference alignment, a new outer bound with the application of compound settings are derived in Theorem~\ref{theorem:Com}, with which the optimal symmetric DoF of Wyner-type networks with only one interfering link are characterized. The interference alignment feasibility condition is further generalized in Definition~\ref{def:pp}, which leads us to the construction of a hypergraph and hence an achievability scheme via hypergraph covering in Theorem~\ref{theorem:Cov}. 

\subsection{Interference Alignment with Alignment-Feasible Graph}
 In what follows, we introduce new notions of alignment-feasible graph and alignment non-conflict matrix, which indicate respectively the feasibility of interference alignment for any two messages, and the non-conflict of alignment feasibility of two messages to a third one, namely whether those two messages are aligned or not has no influence on the third one.
\begin{definition} [Alignment-Feasible Graph]
\label{def:afg}
  The {\bf alignment-feasible graph (AFG)}, denoted by $\Gc_{AFG}$, refers to a graph with vertices representing the messages and with edges between any two messages indicating if they are alignment-feasible. Two messages $W_i$ and $W_j$ are said to be alignment-feasible, denoted by $i \leftrightarrow j$, if
        \begin{align} \label{eq:afgcond}
          \Tc_i \nsubseteq \Tc_j, \quad \text{and} \quad \Tc_j \nsubseteq \Tc_i.
        \end{align}
\end{definition}

\begin{remark} \normalfont
The condition in \eqref{eq:afgcond} implies the alignment feasibility, that is, it is feasible to align these two messages $W_i$ and $W_j$ in the same subspace without causing mutual interference by choosing proper transmitting sources, such that the transmitted signal of one message will not interfere the intended receiver of the other message. A similar insight was also revealed in \cite{Jafar:2012Index} in the context of index coding. 
\end{remark}

 \begin{definition} [Alignment Non-Conflict Matrix]
\label{def:anm}
 Regarding a cycle $i_1 \leftrightarrow i_2 \leftrightarrow \dots \leftrightarrow i_K \leftrightarrow i_1$ in an alignment-feasible graph, we construct a $K \times K$ binary matrix $\Am$, referred to as {\bf alignment non-conflict matrix}, with element $\Am_{kj}=1$ ($j,k \in \Kc$), if
        \begin{align} \label{eq:non-conflict}
          \Tc_{i_j} \cap \Tc_{i_{j+1}}^c \nsubseteq \Tc_{i_k}, \quad \text{and} \quad \Tc_{i_{j+1}} \cap \Tc_{i_j}^c \nsubseteq \Tc_{i_k},
        \end{align}
and with $\Am_{kj}=0$ otherwise. Further, we reset $\Am_{kj}=0$ $(\forall~k)$, if
\begin{align} \label{eq:conflict}
 \Tc_{i_j} \bigcap \Tc_{i_{j+1}}^c \bigcap_{k:\Am_{kj}=1} \Tc_{i_k}^c = \emptyset, \quad \text{or} \quad \Tc_{i_{j+1}} \bigcap \Tc_{i_{j}}^c \bigcap_{k:\Am_{kj}=1} \Tc_{i_k}^c = \emptyset.
\end{align}
\end{definition}
\begin{remark}\normalfont
The elements in $\Tc_{i_j} \cap \Tc_{i_{j+1}}^c$ and $\Tc_{i_{j+1}} \cap \Tc_{i_j}^c$ represent the indices of potential transmitters (without loss of generality, we assume Transmitters $i_j$ and $i_{j+1}$) that carry $W_{i_j}$ and $W_{i_{j+1}}$, respectively. As such, the condition in \eqref{eq:non-conflict} indicates that both Transmitters $i_j$ and $i_{j+1}$ are not connected to Receiver $i_k$, and therefore the subspace occupied by the aligned signals $X_{i_j}(W_{i_j})$ and $X_{i_{j+1}}(W_{i_{j+1}})$ is absent at Receiver $i_k$ such that the total dimensions of required subspace are reduced. In contrast, the condition in \eqref{eq:conflict} indicates a conflict in which there do not exist any common elements in $\Tc_{i_j} \cap \Tc_{i_{j+1}}^c$ and $\Tc_{i_{j+1}} \cap \Tc_{i_j}^c$ satisfying \eqref{eq:non-conflict} for all $i_k$ when $\Am_{kj}=1$. Hence, the number of `1's in each row of $\Am$ indicates the number of dimensions associated with the cycle in alignment-feasible graph that can be absent to Receiver $i_k$. The minimum value among all rows gives the number of reducible dimensions (say $q$) for all receivers. As such, $K-q$ indicates the number of dimensions required by all receivers for a feasible interference alignment.
\end{remark}

Given the above alignment-feasible condition and alignment non-conflict matrix, we are able to identify the sufficient conditions to achieve a certain amount of symmetric DoF as follows.

\begin{theorem} [Achievable DoF with Alignment-Feasible Graph]
\label{theorem:Arb}
  For a $K$-cell TIM-CoMP problem with arbitrary topologies, the following symmetric DoF are achievable:
  \begin{itemize}
    \item $d_\sym = \frac{2}{K}$, if there exists a Hamiltonian cycle or a perfect matching in $\Gc_{AFG}$;
    \item $d_\sym = \frac{2}{K-q}$, if there exists a Hamiltonian cycle in $\Gc_{AFG}$, say $i_1 \leftrightarrow i_2 \leftrightarrow \dots \leftrightarrow i_K \leftrightarrow i_1$, associated with an alignment non-conflict matrix $\Am$, such that
        \begin{align}
          q \defeq \min_{k} \sum_{j} \Am_{kj}
        \end{align}
        when $\tau_c \ge K-q$.
  \end{itemize}
\end{theorem}
\begin{proof}
See Appendix \ref{proof:Arb}.
\end{proof}

Let us consider again the network topology studied in Example 1 to show how Theorem~\ref{theorem:Arb} works with alignment-feasible graph and alignment non-conflict matrix.
\begin{example}
\normalfont
We first detail an interference alignment scheme, followed by the interpretation with alignment-feasible graph and alignment non-conflict matrix.

Recall that we have transmit and receive sets $\Tc_1=\{1,4\}, \ \Tc_2 = \{1,2,3,4\}, \ \Tc_3=\{1,3\}, \ \Tc_4=\{3,4,5\}, \ \Tc_5=\{3,5\}, \ \Tc_6=\{3,6\}$, $\Rc_1=\{1,2,3\}, \ \Rc_2=\{2\}, \ \Rc_3=\{2,3,4,5,6\}, \ \Rc_4=\{1,2,4\}, \ \Rc_5=\{4,5\}, \ \Rc_6=\{6\}$.
For notational convenience, we denote by $a,b,c,d,e,f$ the messages desired by six receivers, with the subscript distinguishing different symbols for the same receiver. We consider a multiple time-slotted protocol, in which a space is spanned such that the symbols will be sent in certain subspaces.
Given six random vectors $\Vm_1, \Vm_2, \Vm_3, \Vm_4, \Vm_5, \Vm_6  \in  \CC^{5 \times 1}$, any five of which are linearly independent, the transmitters send signals with precoding
\begin{gather}
\Xm_1 = \Vm_2 b_1 + \Vm_3 c_2 + \Vm_4 a_1,\quad
\Xm_2 = \Vm_6 b_2\\
\Xm_3 = \Vm_4 d_2 + \Vm_5 c_1,\quad
\Xm_4 = \Vm_5 a_2\\
\Xm_5 = \Vm_1 d_1 + \Vm_3 e_2 + \Vm_6 e_1, \quad
\Xm_6 = \Vm_1 f_1 + \Vm_2 f_2
\end{gather}
within five time slots, where $\Xm_i \in \CC^{5 \times 1}$ is the vector of the concatenated transmit signals from Transmitter $i$, with each element being the transmitted signal at each corresponding time slot.

We assume the coherence time $\tau_c \ge 5$, during which the channel coefficients keep constant. The received signal at Receiver 2 for example within five time slots, with $\Tc_2=\{1,2,3,4\}$, can be written as
\begin{align}
\Ym_2
&= h_{21} \Xm_1 + h_{22} \Xm_2 + h_{23} \Xm_3 + h_{24} \Xm_4 + \Zm_2 \\
&= \underbrace{h_{21} \Vm_2 b_1 + h_{22} \Vm_6 b_2}_{\DS} + \underbrace{\Vm_3 h_{21} c_2 + \Vm_4(h_{21} a_1+h_{23} d_2) + \Vm_5 (h_{23} c_1 + h_{24} a_2)}_{\AI} + \Zm_2.
\end{align}

Recall that $\{\Vm_i, i=1,\dots,6\}$ are $5 \times 1$ linearly independent vectors spanning five-dimensional space, by which it follows that the interferences are aligned in the three-dimensional subspace spanned by $\Vm_3$, $\Vm_4$ and $\Vm_5$, leaving two-dimensional interference-free subspace spanned by $\Vm_2$ and $\Vm_6$ to the desired symbols $b_1,b_2$. Note that the subspace spanned by $\Vm_1$ is absent to Receiver 2. Hence, the desired messages of Receiver 2 can be successfully recovered, almost surely. In doing so, all receivers can decode two messages within five slots, yielding the symmetric DoF of $\frac{2}{5}$, which coincides with those achieved by fractional selective graph coloring. 

\begin{figure}[h]
\centering
\includegraphics[width=0.9\columnwidth]{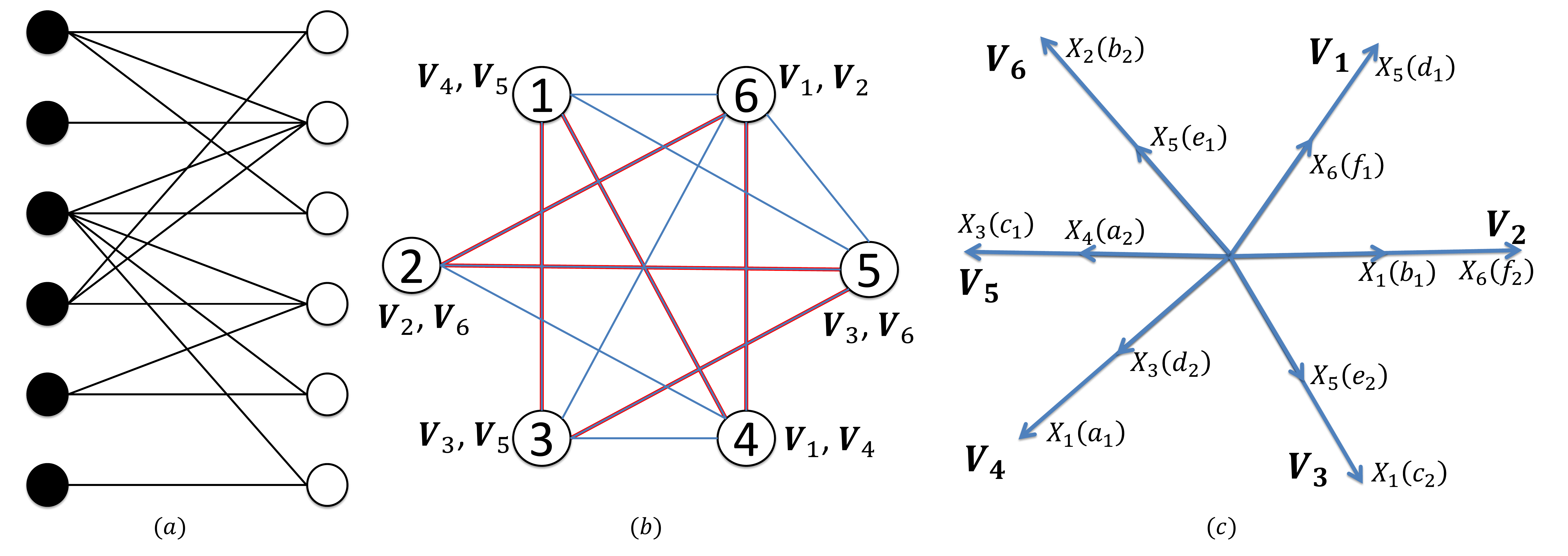}
\caption{ (a) An instance of TIM-CoMP problem $(K=6)$, and (b) the alignment-feasible graph $\Gc_{AFG}$, in which there exists a Hamiltonian cycles with edges in red. (c) An interference alignment scheme, where for example $X_5(d_1)$ denotes a signal sent from Transmitter 5 carrying a symbol $d_1$ desired by Receiver 4. Overall, every message appears twice, and for each receiver there exists at least one absent subspace ($q=1$).}
\label{fig:hamil_q}
\end{figure}

Let us see how Theorem~\ref{theorem:Arb} works.
Based on the transmit sets and the definition of alignment-feasible graph, we construct $\Gc_{AFG}$ as shown in Fig.~\ref{fig:hamil_q}(b). The vertices correspond to messages, and any two messages are joint with an edge if their transmit sets are not the subset of one another.
Notably, there exist a Hamiltonian cycle $1 \leftrightarrow 3 \leftrightarrow 5 \leftrightarrow 2 \leftrightarrow 6 \leftrightarrow 4 \leftrightarrow 1$, and the corresponding alignment non-conflict matrix
\begin{align}
\Am = \begin{bmatrix} 0 & 0 & 1 & 1& 1 & 0 \\ 0 & 0 & 1 & 1 & 1 & 0 \\ 0 & 0 & 0 & 1 & 1 &0 \\ 0 & 0 & 0 & 0 & 1 &0 \\ 0 & 1 & 1 & 0 & 0 & 1 \\ 0 & 0 & 0 & 1 & 0 &0  \end{bmatrix}
\end{align}
where $q=1$. As such, according to Theorem~\ref{theorem:Arb}, we conclude that symmetric DoF of $\frac{2}{5}$ are achievable. It is shown in Fig.~\ref{fig:hamil_q}(c) an interference alignment solution. For each message, two symbols are sent, each of which are along with one direction spanned by a $5 \times 1$ vector $\Vm_i$. Two adjacent messages in the Hamiltonian cycle in Fig.~\ref{fig:hamil_q}(b) are aligned in one direction in Fig.~\ref{fig:hamil_q}(c), e.g., messages $W_1$ and $W_3$ are joint with an edge in $\Gc_{AFG}$, such that two symbols $X_4(a_2)$ and $X_3(c_1)$ are aligned in subspace spanned by $\Vm_5$. Due to $\Tc_2=\{1,2,3,4\}$, Receiver 2 will not hear signals from Transmitters 5 and 6, such that the linear independence of $\Vm_1$ is not necessary. So, five-dimensional subspace is sufficient for Receiver 2. The similar phenomenon  can be observed at all receivers. As such, only five vectors in $\{\Vm_i,i=1,\dots,6\}$ are required to be linearly independent, that is $q =1$. The feasible solution in Fig.~\ref{fig:hamil_q}(c) can be interpreted as vector assignment in Fig.~\ref{fig:hamil_q}(b), where the adjacent vertices in the Hamiltonian cycle are with some vectors shared. \hfill$\square$

\end{example}

There is a very interesting observation. The alignment-feasible condition in \eqref{eq:afgcond} also implies the feasibility of selective graph coloring on $\Gc_e^2$. The fact that two messages satisfy \eqref{eq:afgcond} means there exist two vertices in two clusters $i$ and $j$ of $\Gc_e^2$ are not adjacent and hence can be assigned the same color. It follows that interference alignment is a general form of interference avoidance, in agreement with the observation in \cite{Jafar:CBIA}. 
Thus, interference alignment provides at least the same performance as interference avoidance. Even better, one advantage of interference alignment over interference avoidance is that, the number of dimensions of the subspace to make interference alignment feasible could be less than the total number of colors (i.e., the total number of time slots to schedule links), as some subspaces may be absent at some receivers (according to the alignment non-conflict matrix) so as to decrease the number of required dimensions. 

The advantage of interference alignment over interference avoidance becomes more evident when it comes to regular networks. Specifically, by the above interference alignment approach, we could identify the achievable symmetric DoF of regular networks as follows.

\begin{theorem} [Achievable DoF for Regular Networks]
\label{theorem:Reg}
For a $(K,d)$-regular network, the symmetric DoF
\begin{align}
d_{\sym}(K,d) = \left\{\Pmatrix{\frac{2}{d+1}, & d \le K-1 \\ \frac{1}{K}, & d = K} \right.
\end{align}
are achievable, when channel coherence time satisfies $\tau_c \ge d+1$.
\end{theorem}
\begin{proof}
See Appendix \ref{proof:Reg}.
\end{proof}
\begin{remark}
For a regular network, the alignment-feasible graph $\Gc_{AFG}$ is a complete graph, and there always exists an alignment non-conflict matrix with $q=K-d-1$ for any Hamiltonian cycle in $\Gc_{AFG}$.
\end{remark}

In what follows, we present a detailed transmission scheme with interference alignment, followed by an interpretation with the concepts of alignment-feasible graph and alignment non-conflict matrix. Remarkably, we also offer a transmission scheme for fast fading channel ($\tau_c=1$) achieving the same symmetric DoF by using retransmission (or so-called repetition coding).
\begin{example}
\normalfont
Let us consider a $(5,3)$-regular network as shown in Fig.~\ref{fig:15}(a). By enabling transmitter cooperation, the achievable symmetric DoF are improved from $\frac{2}{5}$ (as reported in~\cite{Jafar:CBIA}) to $\frac{1}{2}$ according to Theorem \ref{theorem:Reg}. In what follows, we will show two interference alignment schemes to achieve this, with channel coherence time $\tau_c \ge 4$ and $\tau_c=1$, respectively.

According to the network topology, we have transmit and receive sets $\Tc_1=\Rc_1=\{1,3,4\},  \Tc_2 = \Rc_2=\{2,4,5\},  \Tc_3=\Rc_3=\{1,3,5\},   \Tc_4=\Rc_4=\{1,2,4\},  \Tc_5=\Rc_5=\{2,3,5\}$.
Similarly, $a,b,c,d,e$ are symbols desired by five receivers. We consider a four time-slotted protocol, in which the symbols are sent as
\begin{gather}
\Xm_1 = \Vm_1 c_1 + \Vm_3 d_1,\quad
\Xm_2 = \Vm_2 d_2 + \Vm_4 e_1\\
\Xm_3 = \Vm_5 a_1 + \Vm_3 e_2,\quad
\Xm_4 = \Vm_4 a_2 + \Vm_1 b_2\\
\Xm_5 = \Vm_5 b_1 + \Vm_2 c_2
\end{gather}
where $\Vm_1, \Vm_2, \Vm_3, \Vm_4, \Vm_5  \in  \CC^{4 \times 1}$ and any four of thme are linearly independent, and $\Xm_i \in \CC^{4 \times 1}$.

To illustrate the interference alignment, we describe the transmitted signals geometrically as shown in Fig.~\ref{fig:15}. In this figure, we depict the subspace spanned by $\{\Vm_i, i=1,\dots,5\}$ as a four-dimensional space, where any four of them suffice to represent this space. We also denote by $X_i(W_j)$ the message $W_j$ sent from Transmitter $i$. Let us still take Receiver 1 for example. Because of $\Tc_1=\{1,3,4\}$, the transmitted signals from the transmitters that do not belong to $\Tc_1$ will not reach Receiver 1, and hence the vector $\Vm_2$ is absent to Receiver 1. In addition, we have the interference-free signals in the directions of $\Vm_4$ and $\Vm_5$, and the aligned interferences carrying messages other than $a_1,a_2$ in the subspace spanned by $\Vm_1$ and $\Vm_3$. Recall that vectors $\{\Vm_1,\Vm_3,\Vm_4,\Vm_5\}$ are linearly independent, almost surely, so that the interference alignment is feasible at Receiver 1, and it can also be checked to be feasible at other receivers.

As shown in Fig.~\ref{fig:15}(b), the corresponding alignment-feasible graph is a complete graph. Given for example a Hamiltonian cycle $1 \leftrightarrow 2 \leftrightarrow 3 \leftrightarrow 4 \leftrightarrow 5 \leftrightarrow 1$, the associated alignment non-conflict matrix is
\begin{align}
\Am = \begin{bmatrix} 0 & 1 & 0 & 0 & 0 \\ 0 & 0 & 1 & 0 & 0 \\ 0 & 0 & 0 & 1 & 0 \\0 & 0 & 0 & 0 & 1 \\1 & 0 & 0 & 0 & 0  \end{bmatrix}
\end{align}
which gives $q=1$ and thus $d_\sym = \frac{2}{K-q}=\frac{1}{2}$.

\begin{figure}[htb]
 \centering
\includegraphics[width=1\columnwidth]{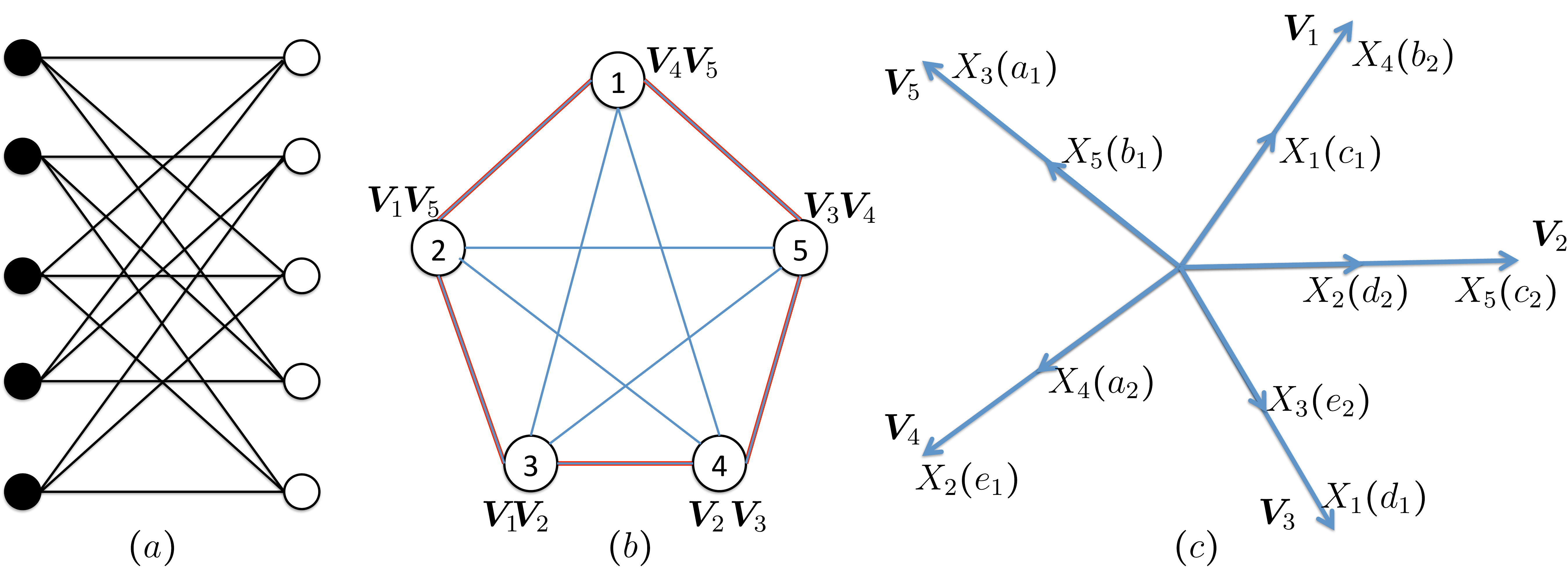}
\caption{(a) topology graph of a $(5,3)$-regular cellular network, (b) alignment-feasible graph as a complete graph with a Hamiltonian cycle in red, and (c) an interference alignment scheme with $\Vm_i$ being a four-dimensional vector, where channel coherence time $\tau_c\ge 4$.}
\label{fig:15}
\end{figure}

When it comes to the case with channel coherence time $\tau_c=1$, the above interference alignment scheme does not work. The symmetric DoF by interference avoidance are $\frac{2}{5}$ according the Theorem~\ref{theorem:IAvoid}. However, it can be improved to $\frac{1}{2}$ as well by an new scheme combining interference alignment and repetition coding as below.

Differently from the above transmission protocol with four time slots, here we use ten time slots to send
\begin{gather}
\Xm_1 = \Vm_{1} c_1 + \Vm_{2} c_3 + \Vm_{5} d_1 + \Vm_{6} d_3 + \Vm_{7}a_5 + \Vm_{9} d_5, \\
\Xm_2 = \Vm_{3} d_2 + \Vm_{4} d_4 + \Vm_{7} e_1 + \Vm_{8} e_3 + \Vm_{2} e_5 + \Vm_{10} b_5,\\
\Xm_3 = \Vm_{5} e_2 + \Vm_{6} e_4 + \Vm_{9} a_1+ \Vm_{10} a_3 + \Vm_{1} c_5 + \Vm_{4} a_5,\\
\Xm_4 =  \Vm_{1} b_2 + \Vm_{2} b_4 + \Vm_{7} a_2  + \Vm_{8} a_4  + \Vm_{3} d_5  + \Vm_{5} b_5,\\
\Xm_5 = \Vm_{3} c_2  + \Vm_{4} c_4  + \Vm_{9} b_1  + \Vm_{10} b_3  + \Vm_{6} e_5  + \Vm_{8} c_5
\end{gather}
where $\Vm_j$ can be chosen as the $j$-th column of identity matrix $\Id_{10}$. Note that the symbols $\{a_5,b_5,c_5,d_5,e_5\}$ are repeatedly sent twice.
Let us look at the decoding at Receiver 1 for example, and the similar procedure holds for other receivers as well. By the above transmission protocol, the signal at Receiver 1 becomes
\begin{align}
\yv_1 &= \Hm_{11} \Xm_1 + \Hm_{13} \Xm_3 + \Hm_{14} \Xm_4 + \Zm_1 \\
&= (c_1 \Hm_{11} + c_5 \Hm_{13} + b_2 \Hm_{14}) \Vm_1 + (c_3 \Hm_{11} + b_4 \Hm_{14}) \Vm_2 \\
&\quad + d_5 \Hm_{14} \Vm_3 + a_5 \Hm_{13} \Vm_4 + (d_1 \Hm_{11} + e_2 \Hm_{13} + b_5 \Hm_{14}) \Vm_5 \\
&\quad+ (d_3 \Hm_{11} + e_4 \Hm_{13}) \Vm_6 + (a_5 \Hm_{11} + a_2 \Hm_{14}) \Vm_7\\
&\quad+ a_4 \Hm_{14} \Vm_8 + (d_5 \Hm_{11} + a_1\Hm_{13}) \Vm_9 + a_3 \Hm_{13} \Vm_{10} + \Zm_1
\end{align}
where $\Hm_{ij} = \diag \{h_{ij}(1), \dots, h_{ij}(10)\}$ is a diagonal matrix. By setting $\Vm_j$ as the $j$-th column of $\Id_{10}$, we have
\begin{gather}
y_1(1) = c_1 h_{11}(1) + c_5 h_{13}(1) + b_2 h_{14}(1), \quad y_1(2) = c_3 h_{11}(2) + b_4 h_{14}(2),\\
y_1(3) = d_5 h_{14}(3), \quad y_1(4) = a_5 h_{13}(4), \quad y_1(5) = d_1 h_{11}(5) + e_2 h_{13}(5) + b_5 h_{14}(5),\\
y_1(6) = d_3 h_{11}(6) + e_4 h_{13}(6), \quad y_1(7) = a_5 h_{11}(7) + a_2 h_{14}(7),\\
y_1(8) = a_4 h_{14}(8), \quad y_1(9) = d_5 h_{11}(9) + a_1 h_{13}(9), \quad y_1(10) = a_3 h_{13}(10)
\end{gather}
with noise terms omitted.

Clearly, the interested symbols $\{a_1,a_2,\dots,a_5\}$ can be recovered from $\{y_1(3),y_1(4),y_1(7),y_1(8),y_1(9),y_1(10)\}$.
Based on the similar analysis and network symmetry, we conclude that 5 symbols per user are delivered within 10 time slots, which gives symmetric DoF $\frac{1}{2}$. 

\begin{figure}[htb]
 \centering
\includegraphics[width=0.78\columnwidth]{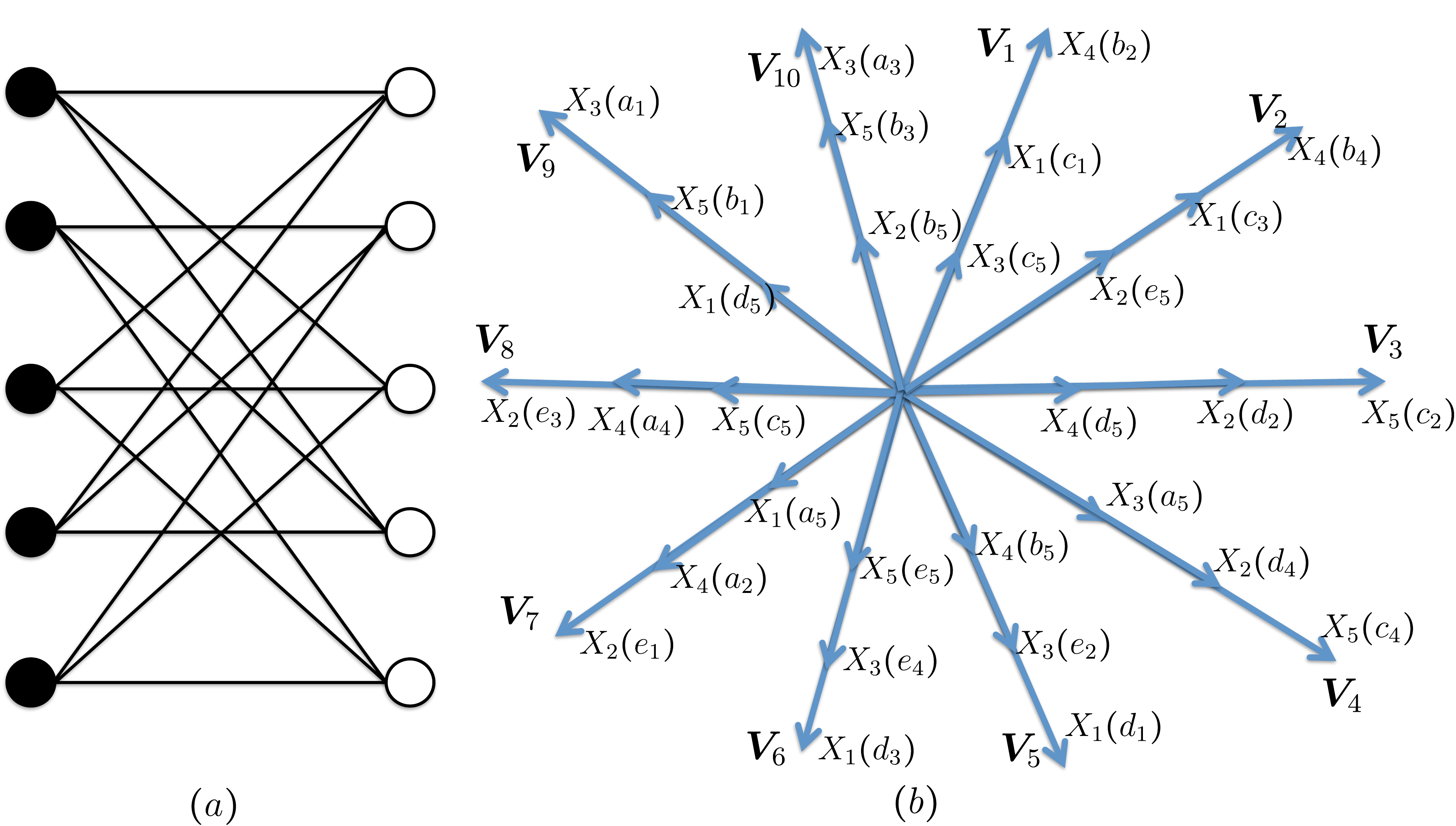}
\caption{(a) topology graph of a $(5,3)$-regular cellular network, (b) an interference alignment scheme with $\Vm_i$ being $i$-th column of $\Id_{10}$, with channel coherence time $\tau_c=1$.}
\label{fig:15-coh}
\end{figure}

It is convenient to look at the transmission/decoding from an interference alignment perspective, as shown in Fig.~\ref{fig:15-coh}(b), although interference alignment here is reduced to interference avoidance. By symbol extension with ten time slots, the transmitter signals span a ten-dimensional subspace. For Receiver 1, the transmitted signals $X_4(d_5),X_3(a_5),X_4(a_4)$, and $X_3(a_3)$ lie in the subspaces spanned by $\Vm_3,\Vm_4,\Vm_8$, and $\Vm_{10}$, respectively, and are free of interference, such that the symbols $\{d_5,a_5,a_4,a_3\}$ can be recovered almost surely. There are two subspaces spanned by $\Vm_7$ and $\Vm_9$ respectively, where the desired signals $X_2(a_2)$ and $X_3(a_1)$ are contaminated respectively by interfering signals $X_1(a_5)$ and $X_1(d_5)$. With the already recovered symbols $a_5$ and $d_5$, the interferences are reconstructed and subtracted at the receiver, so that the desired symbols $\{a_2,a_1\}$ can be recovered almost surely. As such, all desired symbols $\{a_1,a_2,a_3,a_4,a_5\}$ can be recovered within ten time slots, yielding $\frac{1}{2}$ DoF. This applies to all other receivers and symmetric DoF of $\frac{1}{2}$ is achievable even in a fast fading channel.  

This demonstrates that interference alignment together with repetition coding can be beneficial over interference avoidance even in fast fading channel $(\tau_c=1)$. This scheme is inspired by the interference alignment approach in \cite{Jafar:CBIA}, and the repetition coding approach in \cite{Avestimehr:2013TIM}.\hfill$\square$ 
\end{example}

\subsection{Outer Bound via Compound Settings}
For the regular networks, the outer bound via generator sequence becomes loose. This urges us to find another bounding techniques. By generalizing and extending the idea in \cite{Gou:2012Mixed,Jafar:CBIA}, we obtain in what follows a new outer bound with the aid of compound settings.

\begin{theorem} [Outer Bound via Compound Settings]
\label{theorem:Com}
The symmetric DoF of $K$-cell TIM-CoMP problems are upper bounded by the solution of the following optimization problem:
\begin{align}
 \min_{\Sc \subseteq \Kc} & \quad \frac{K-\abs{\Sc'}}{2K-\abs{\Sc'} - \abs{\Sc}}\\
 s.t.& \quad \Sc' = \{i | \Rc_i \subseteq \Sc \}\\
 &\quad \cup_{j \in \Sc} \Tc_j =\Kc
\end{align}
where $[\Bm^\T]_i$ is the $i$-th row of $\Bm^\T$ (i.e., $i$-th column of $\Bm$).
\end{theorem}
\begin{proof}
See Appendix~\ref{proof:Com}.
\end{proof}

In general, with transmitter cooperation, the interference channel form a virtual broadcast channel, such that it enables us to obtain a not-too-loose outer bound by mimicking the compound channel setting with quite limited knowledge of channel uncertainty \cite{Gou:Compound,Gou:2012Mixed}.
For each receiver, we introduce a number of compound receivers, each of which is statistically equivalent to the original one and requires the same message. So, with TIM setting, it looks as if the transmitter in this virtual BC has only knowledge of linearly independent channel realizations, which put us in a finite-state compound BC setting \cite{Gou:2012Mixed}. The corresponding outer bound can therefore serve as a outer bound of our problem, because above procedure does not reduce capacity. Nevertheless, the particularity of our problem calls for some specific treatments.
Due to partial connectivity, to enable linear independence of channel realizations of compound receivers (i.e., states), it needs at most $\abs{\Tc_j}-1$ compound receivers for Receiver $j$. The message mapping relation, which reflect the network topology, further reduces the required states, because the presence of a certain set of messages makes some transmitters transparent in compound BC settings, such that $\abs{\Tc_j}$ can be further reduced. Intuitively, regular or semi-regular (i.e., nearly regular) networks would prefer this compound setting outer bound, because it makes the numbers of required states with linear independence more balanced across receivers.

In what follows, we derive an outer bound with compound settings for a regular topology. A more general version will be presented in Appendix~\ref{proof:Com}.

\begin{example} \label{ex:compound}
\normalfont
We take the (5,3)-regular cellular network studied in Example 4 into account. By Fano's inequality, we have
\begin{align}
n (R_1 - \epsilon_n) &\le I(W_1, Y_1^n | \Hc^n, \Gc) \\
&= h(Y_1^n | \Hc^n, \Gc) - h(Y_1^n | W_1, \Hc^n, \Gc)\\
&\le n \log P - h(Y_1^n | W_1, \Hc^n, \Gc) + n \cdot O(1).
\end{align}

Assuming there are two compound receivers demanding the same message $W_1$, we have two compound signals $Y_1',Y_1''$, which are also the linear combinations of $X_1,X_3,X_4$ as $Y_1$, yet with independent channel coefficients. Thus, these three received signals are linearly independent with regard to $X_1,X_3,X_4$, almost surely, and are statistically equivalent, which results in the same achievable rate $R_1$. Similarly, we have
\begin{align}
n (R_1 -\epsilon_n)  &\le n \log P - h(Y_1'^n | W_1, \Hc^n, \Gc) + n \cdot O(1) \\
n (R_1 -\epsilon_n)  &\le n \log P - h(Y_1''^n | W_1, \Hc^n, \Gc) + n \cdot O(1).
\end{align}
For Receiver 2, we consider the statistically equivalent received signals $Y_2$ by itself and $Y_2'$ by a compound receiver, and have
\begin{align}
n (R_2 -\epsilon_n)  &\le n \log P - h(Y_2^n | W_2, \Hc^n, \Gc) + n \cdot O(1)\\
n (R_2 -\epsilon_n)  &\le n \log P - h(Y_2'^n | W_2, \Hc^n, \Gc) +n \cdot O(1).
\end{align}

Combining all above inequalities, we have
\begin{align}
\MoveEqLeft n(3R_1 + 2R_2 - \epsilon_n) \\&\le 5n \log P - h(Y_1^n, Y_1'^n, Y_1''^n, Y_2^n, Y_2'^n | W_1, W_2, \Hc^n, \Gc) + n \cdot O(1) \\
&= 5n \log P - h(\{X_i^n+\bar{Z}_i^n, i=1,\dots,5\} | W_1, W_2, \Hc^n, \Gc) + n \cdot O(1) \\
&= 5n \log P -n (R_3+R_4+R_5) + n \cdot O(1)
\end{align}
where $Y_1, Y_1', Y_1'', Y_2, Y_2'$ are linearly independent with regard to $\{X_i, i=1,2,3,4,5\}$, by which the noisy versions of $\{X_i, i=1,2,3,4,5\}$, i.e., $X_i^n+\bar{Z}_i^n$ with $\bar{Z}_i$ being bounded noise term, can be recovered, almost surely; the last equality due to
\begin{align}
n (R_3 + R_4 + R_5) &=H(W_3,W_4,W_5)\\
&=H(W_3,W_4,W_5) - H(W_3,W_4,W_5 | \{X_i^n,i=1,\dots,5\},W_1,W_2, \Hc^n, \Gc) \label{eq:codebooks}\\
&= I(W_3,W_4,W_5; \{X_i^n, i=1,\dots,5\} | W_1, W_2, \Hc^n, \Gc)\\
&= I(W_3,W_4,W_5; \{X_i^n+\bar{Z}_i^n, i=1,\dots,5\} | W_1, W_2, \Hc^n, \Gc) + n \cdot O(1) \\
&= h(\{X_i^n+\bar{Z}_i^n, i=1,\dots,5\} | W_1, W_2, \Hc^n, \Gc) + n \cdot O(1) \label{eq:decodable}.
\end{align}
where the second term in \eqref{eq:codebooks} is zero because $\{X_i^n,i=1,\dots,5\}$ are encoded from $W_{1:5}$ and the encoding process (or mapping) is invertible, such that the knowledge/uncertainty of $\{X_i^n,i=1,\dots,5\}$ is equivalent to the knowledge/uncertainty of $W_{1:5}$.
By now, according to the definition of symmetric DoF, it follows that
\begin{align}
d_\sym \le \frac{5}{8}.
\end{align}

In contrast, by generator bound, the best possible outer bound is $d_\sym \le \frac{4}{5}$, which is looser. On the other hand, if this compound setting bound applies to the irregular network in Example 1, then the best possible outer bound will be $d_\sym \le \frac{4}{7}$, which is looser than that by generator bound. This confirms that compound setting bound is more suitable to regular networks, while generator sequence bound is more preferable to irregular networks. \hfill$\square$ 
\end{example}

\subsection{The optimality of Interference Alignment}
By the above outer bound, we are able to characterize the optimal symmetric  DoF of a subset of regular networks.

\begin{corollary} [Optimal DoF of Cyclic Wyner-type Networks]
\label{cor:Wyner}
For a $(K,2)$-regular network, e.g., a cyclic Wyner-type network, the optimal symmetric DoF are
\begin{align}
  d_{\sym}(K,2) = \left\{\Pmatrix{ \frac{1}{2}, & K = 2 \\ \frac{2}{3}, & K \ge 3} \right.
\end{align}
if the coherence time $\tau_c \ge 3$ when $K \ge 3$.
\end{corollary}
\begin{proof}
See Appendix \ref{proof:Wyner}.
\end{proof}

\subsection{Interference Alignment with Proper Partition and Hypergraph Covering}

The alignment feasibility condition in Defenitions~\ref{def:afg} and \ref{def:anm} can also be generalized to more than two messages, as shown in the following definitions.

\begin{definition} [Proper Partition]
\label{def:pp}
A partition $\Kc=\{\Pc_1, \Pc_2,\dots,\Pc_\kappa \}$ with size $\kappa$, where $\cup_{i=1}^\kappa \Pc_i = \Kc$ and $\Pc_i \cap \Pc_j=\emptyset$ $\forall~i\neq j$, is called a {\bf proper partition}, if for every portion $\Pc_i = \{i_1,i_2,\dots,i_{p_i}\}$ with $p_i \defeq \abs{\Pc_i}$ $(i \in [\kappa])$, we have
        \begin{align} \label{eq:ppcond}
          \Tc_{i_k} \bigcap \left(\bigcup_{i_j \in \Pc_i \backslash i_k} \Tc_{i_j}\right)^c \neq \emptyset, \quad \forall~i_k \in \Pc_i.
        \end{align}
 \end{definition}

 \begin{definition} [Alignment Non-Conflict Matrix]
 \label{def:pp-anm}
For a proper partition  $\{\Pc_{1},\dots,\Pc_{\kappa}\}$, we construct a $K \times \kappa$ binary matrix $\Am$, with $\Am_{ij}=1$ ($j \in [\kappa], i\in \Kc)$, if
        \begin{align}
          \Tc_{j_{t}} \bigcap \left(\bigcup_{j_{s} \in \Pc_{j} \backslash j_{t}} \Tc_{j_{s}}\right)^c \nsubseteq \Tc_i, \quad \forall~j_t \in \Pc_j
        \end{align}
and with $\Am_{ij}=0$ otherwise. Further, we reset $\Am_{ij}=0$, if there exist $j_t \in \Pc_j$ and $i \in \Kc$, such that
\begin{align}
 \Tc_{j_{t}} \bigcap \left(\bigcup_{j_{s} \in \Pc_{j} \backslash j_{t}} \Tc_{j_{s}}\right)^c \bigcap_{i:\Am_{ij}=1} \Tc_{i}^c = \emptyset.
\end{align}
\end{definition}

The elements in each portion of proper partition imply that the corresponding messages are able to align in the same subspace, whereas the alignment non-conflict matrix identifies if this subspace is absent to some receivers. As such, relying on these definitions, the sufficient conditions to achieve a certain amount of symmetric DoF are presented as follows.

\begin{theorem} [Achievable DoF with Proper Partition]
\label{theorem:Arb-more}
  For a $K$-cell cellular network with arbitrary topologies, the following symmetric DoF are achievable:
  \begin{itemize}
    \item $d_\sym = \frac{1}{\kappa}$, if there exists a proper partition with size $\kappa$;
    \item $d_\sym = \frac{1}{\kappa-q}$ with $\tau_c \ge \kappa-q$,  if there exists a proper partition with size $\kappa$, say $\{\Pc_{1},\dots,\Pc_{\kappa}\}$, associated with an $K \times \kappa$ alignment non-conflict matrix $\Am$, such that
    \begin{align}
          q \defeq \min_{i} \sum_{j} \Am_{ij}.
        \end{align}
  \end{itemize}
\end{theorem}
\begin{proof}
See Appendix~\ref{proof:Arb-more}.
\end{proof}

The same observation of alignment-feasible graphs can be obtained here. A proper portion in \eqref{eq:ppcond} implies the feasibility of a proper selective graph coloring in $\Gc_e^2$. Any two (or more) vertices in clusters $j_\Sc$ $(\Sc \subseteq \Kc)$ in $\Gc_e$ (corresponding to edges in $\Gc$ connecting Transmitter $i_s$ to Receiver $j_s$ $(\forall~s \in \Sc)$) that receive the same color are scheduled in a single time slot without causing interference, implying that the transmitted signals in the form of $\{X_{i_s}(W_{j_s}),s \in \Sc\}$ are alignment-feasible in the same subspace. Due to the fact that the required number of subspace can be less (according to alignment non-conflict matrix in Definition \ref{def:pp-anm}), interference alignment  based on proper partition performs no worse than interference avoidance.

\begin{example}
\normalfont
An example regarding proper partition is shown in Fig.~\ref{fig:partite}. Given the transmit sets $\Tc_1=\{1,4\}$, $\Tc_2=\{2,3\}$, $\Tc_3=\{2,3\}$, $\Tc_4=\{1,2,4\}$, $\Tc_5=\{3,5,6\}$, and $\Tc_6=\{4,5,6\}$, we have a proper partition $\{\{1,3,5\},\{2,4,6\}\}$ with $\kappa=2$, such that $\{X_1(d),X_3(b),X_5(f)\}$ and $\{X_2(c),X_4(a),X_6(e)\}$ are aligned in a subspace respectively. As shown in Fig.~\ref{fig:partite}(b), an interference alignment can be constructed to deliver one symbol per user within two time slots. Thus, symmetric DoF $\frac{1}{2}$ is achievable. In this example, $q=0$. \hfill$\square$ 

\begin{figure}[htb]
\centering
\includegraphics[width=0.5\columnwidth]{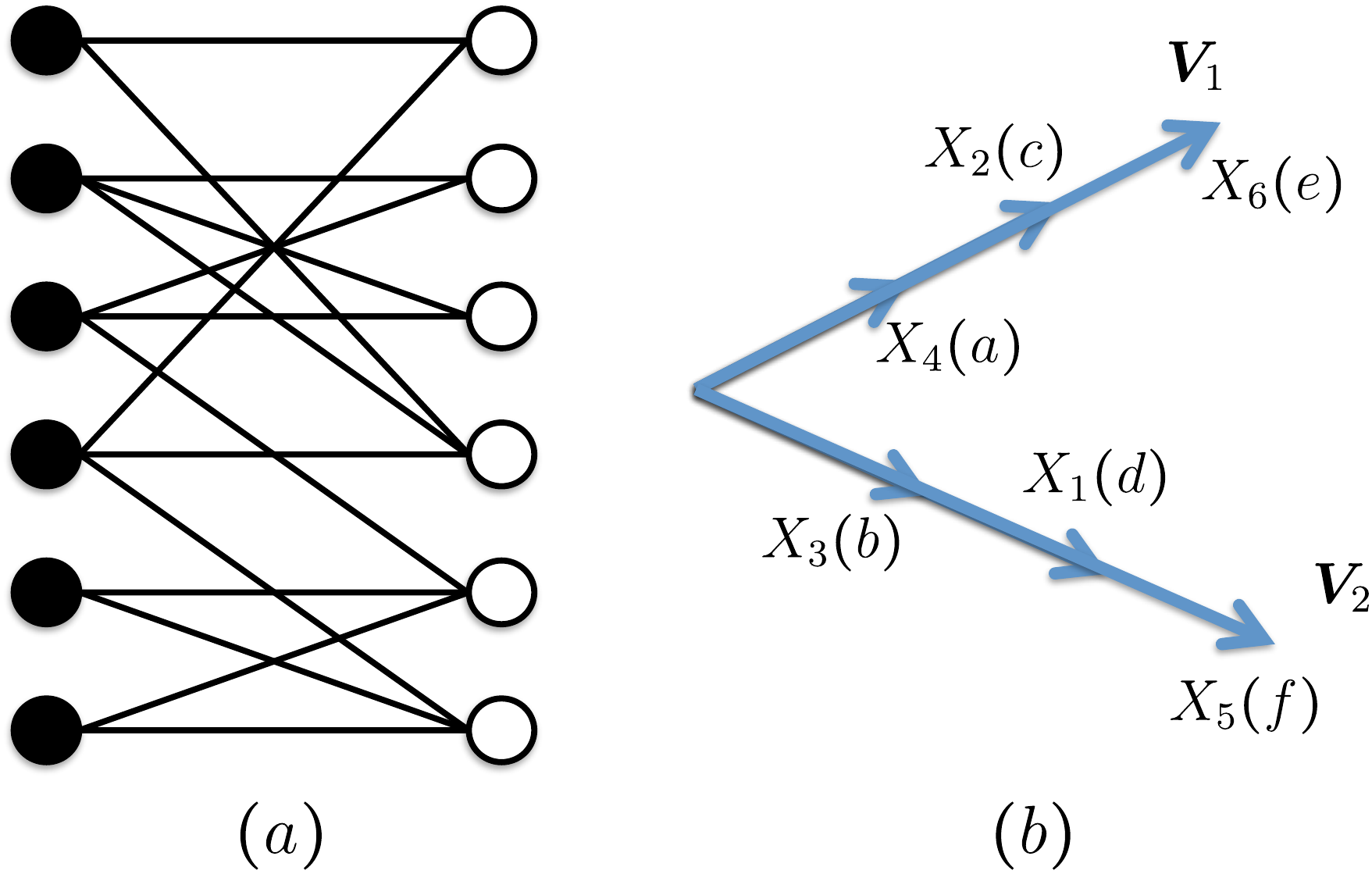}
\caption{ (a) An instance of TIM-CoMP problem $(K=6)$ with a proper partition $\{\{1,3,5\},\{2,4,6\}\}$. (b) An interference alignment scheme, where the messages whose transmitted signals are aligned in the same subspace belong to one portion.}
\label{fig:partite}
\end{figure}

\end{example}

From the previous theorems, we observe that the messages connected by an edge in $\Gc_{AFG}$ or belonged to the same portion of a proper partition are able to be scheduled at the same time slot or be aligned at the same direction. Inspired by this observation, we construct a hypergraph and translate our problem into a covering problem of this hypergraph. 

\begin{theorem} [Achievable DoF via Hypergraph Covering]
\label{theorem:Cov}
For the TIM-CoMP problem with arbitrary topologies, the symmetric DoF
\begin{align}
  d_\sym = \frac{1}{\tau_f(\Hc_\Gc)}
\end{align}
are achievable,
where $\tau_f(\Hc_\Gc)$ is the fractional covering number of the hypergraph $\Hc_\Gc=(\Kc,\Xc)$ with the vertex set $\Kc$ representing messages and the hyperedge set $\Xc$ including all satisfactory subsets $\Xc_i \defeq \{{i_1},{i_2},\dots,{i_{\abs{\Xc_i}}}\} \subseteq \Kc$ such that
\begin{align}
  \Tc_{{i_k}} \bigcap \left(\bigcup_{{i_j} \in \Xc_i \backslash {i_k}} \Tc_{{i_j}}\right)^c \neq \emptyset, \quad \forall~{i_k} \in \Xc_i.
\end{align}
\end{theorem}

\begin{proof}
See Appendix~\ref{proof:Cov}.
\end{proof}

Note that the relation of the vertices of a hyperedge is similar to that of the portion of a proper partition as in \eqref{eq:ppcond}, indicating that the messages that belong to any hyperedge are alignment feasible.
The characterization of the fractional hypergraph covering number $\tau_f(\Hc_\Gc)$ can also be performed by the following integer linear programming relaxation
\begin{align}
\tau_f(\Hc_\Gc) = \min & \quad \sum_{i \in \Kc} \rho_i\\
s.t.& \quad \sum_{i \in \Kc : j \in \Xc_i} \rho_i \ge 1, \quad \forall~j \in \Kc \\
& \quad \rho_i \in [0,1], \quad \forall~i\in \Kc
\end{align}
where $\rho_i$ is an indicator variable associated with the hyperedge $\Xc_i \in \Xc$ with value between 0 and 1 indicating the weight assigned to $\Xc_i$ accounts for the total weight, the first constraint ensures that every vertex in $\Kc$ is covered at least once, and the last constraint specifies a fractional $\rho_i$, which is the relaxation of integers $\{0,1\}$. Although the optimization of this linear program is NP-hard, the connection of our problem and hypergraph covering bridges the TIM-CoMP problem and the hypergraph covering problem, such that the progress on one problem is automatically transferrable to the other one. 

Essentially, the above hypergraph covering aided approach relies on the one-to-one alignment.
As known in TIM problems, subspace alignment is a generalized version of one-to-one alignment and the former usually performs better than the latter. In what follows, we show that, with message sharing, subspace alignment boils down to one-to-one alignment with proper message and subspace splitting.
\begin{example}
\normalfont
Consider a network topology shown in Fig.~\ref{fig:hycover}(a). Without message sharing, the optimal symmetric DoF value is $\frac{1}{3}$, which is achieved by a subspace alignment scheme. Every transmitter sends message in a one-dimensional subspace out of in total three-dimensional space. At receiver 1, the interference from Transmitter 4 lies in the subspace spanned by the interference caused by Transmitters 2 and 3. As such, the desired message of Receiver 1 can be recovered almost surely. At Receivers 2, 3, and 4, the interference occupies one-dimensional subspace, leaving two-dimensional interference-free subspace to desired messages. Thus, the symmetric DoF of $\frac{1}{3}$ are achievable.

In contrast, with message sharing and proper message splitting, a one-to-one alignment scheme can achieve symmetric DoF of $\frac{2}{5}$. Intuitively, every transmitter sends two messages occupying a two-dimensional subspace in a five-dimensional space. Denote by $\Vm_i$ the subspace occupied by Transmitter $i$, where $\dim(\Vm_i)=2$ and $\dim(\cup_{i=1}^4 \Vm_i)=5$. 
At Receiver 1, the interfering subspaces associated with Transmitters 2 and 3 are overlapped with one-dimensional subspace, i.e., $\dim(\Vm_2 \cup \Vm_3)=3$ and $\dim(\Vm_2 \cap \Vm_3)=1$. In addition, the interfering symbols from Transmitter 4 lie in the subspace spanned by the interference from Transmitters 2 and 3, i.e., $\Vm_4 \in span \{\Vm_2, \Vm_3\}$. It would seem subspace alignment is required. In fact, it can be done by a one-to-one alignment scheme by splitting subspace into, e.g., 
\begin{align}
\Vm_1 = [\vv_1 \ \vv_2], \quad \Vm_2 = [\vv_3 \ \vv_4], \quad \Vm_3 = [\vv_3 \ \vv_5], \quad \Vm_4 = [\vv_4 \ \vv_5]
\end{align}
where $\{\vv_i, i=1,\dots,5\}$ are $5 \times 1$ linearly independent vectors,
and by splitting messages and sending
\begin{align} \label{eq:hyperedge}
\Xm_1 = \Vm_1 \Bmatrix{a_1 \\ a_2} , \quad \Xm_2 = \Vm_2 \Bmatrix{b_2 \\ c_1}, \quad \Xm_3=\Vm_3 \Bmatrix{d_1 \\ c_2}, \quad \Xm_4=\Vm_4 \Bmatrix{d_2 \\ b_1} 
\end{align}
from four transmitters within five time slots, respectively. The concept of interference alignment is illustrated in Fig.~\ref{fig:hycover}(b).

\begin{figure}[htb]
\centering
\includegraphics[width=0.7\columnwidth]{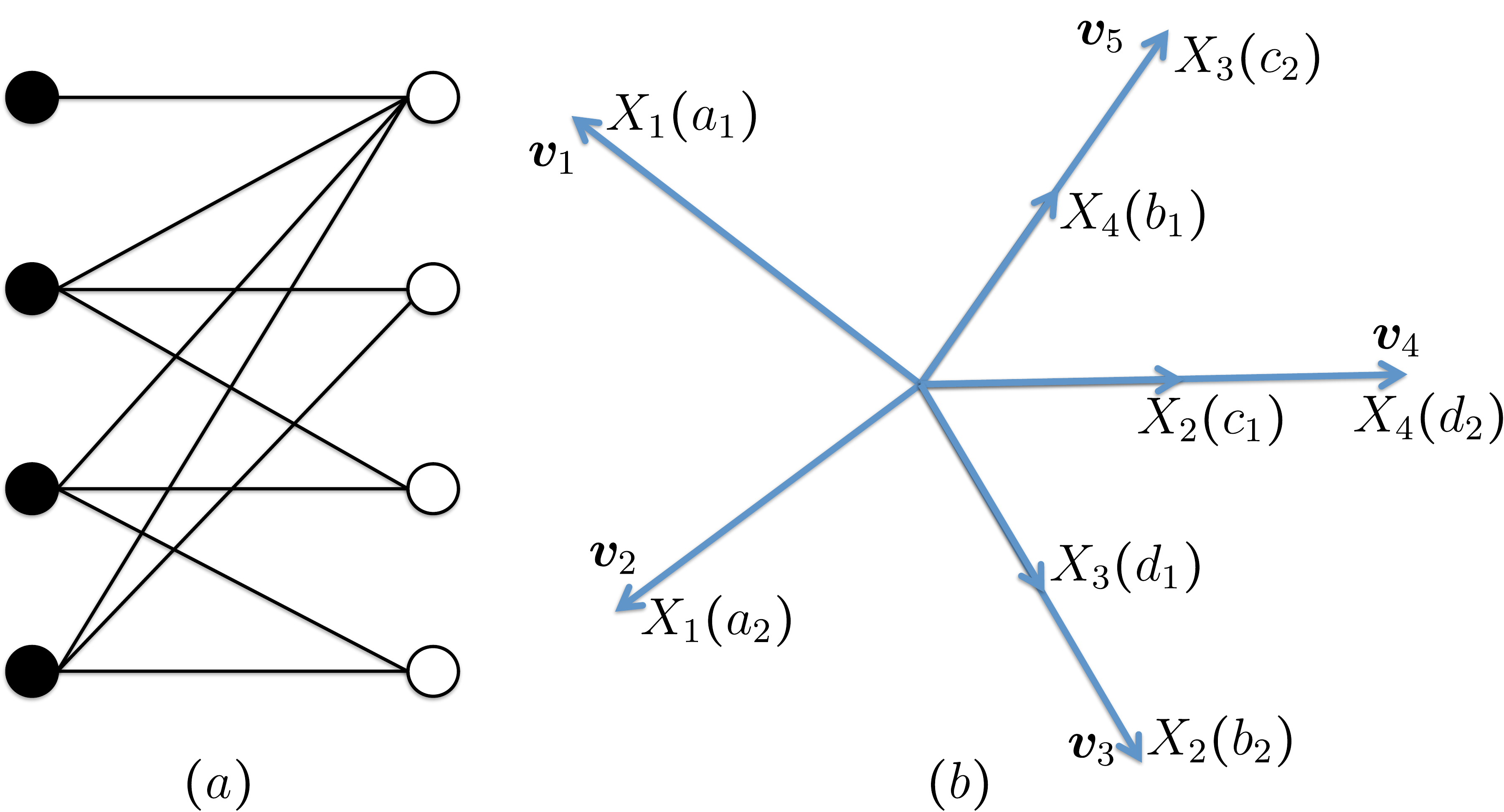}
\caption{ (a) An instance of TIM-CoMP problem $(K=4)$. (b) An one-to-one interference alignment scheme.}
\label{fig:hycover}
\end{figure}

For TIM problem where both the source and destination of one message are determined {\em a priori}, subspace alignment is necessary to align the interference from Transmitter 4 to the subspace spanned by interferences from Transmitters 2 and 3. In contrast, for TIM-CoMP problems, the source of one message can be any transmitter that it is connected, such that by proper message splitting and subspace splitting, it is possible to replace subspace alignment by one-to-one alignment.

Let us look at the above subspace and one-to-one alignment schemes from a hypergraph covering perspective. According to the condition of hyperedges in \eqref{eq:hyperedge}, we have following hyperedges
\begin{align}
\{1\}, \{2\},\{3\},\{4\},\{5\}, \{2,3\}, \{3,4\},\{4,5\}
\end{align}
A proper fractional hypergraph covering is to choose the following hyperedges
\begin{align}
\{1\}, \{1\}, \{2,3\}, \{3,4\},\{4,5\}
\end{align}
which gives fractional hypergraph covering number of $\frac{5}{2}$ and thus yields the symmetric DoF of $\frac{2}{5}$. \hfill$\square$ 
\end{example}

\section{Relation to Index Coding Problems}

Knowing that the TIM problem was nicely bridged to the index coding problem~\cite{Jafar:2013TIM}, one may wonder if there exist relations between our problem and index coding. Indeed, our problem can also be related to the index coding problem. Before presenting this relation, we first define the index coding problem and its demand graph similarly to those in \cite{Jafar:2013TIM,DynamicIndex}.
\begin{definition} [Index Coding]
  A multiple unicast {\bf index coding} problem, denoted as $\IC(k | \Sc_k)$, is comprised of a transmitter who wants to send $K$ messages $W_k,k \in \Kc$ to their respective receivers over a noiseless link, and $K$ receivers, each of which has prior knowledge of $W_{\Sc_k}$ with $\Sc_k \subseteq \Kc\backslash k$. Its {\bf demand graph} is a directed bipartite graph $\Gc_d=(\Wc,\Kc,\Ec)$ with vertices of Message $W_k \in \Wc$ and Receiver $k$ $(k\in \Kc)$, and there exists a directed forward edge $i \to j$ from Message $W_i$ to Receiver $j$ if $W_i$ is demanded by Receiver $j$ and a backward edge $k \gets j$ from Receiver $j$ to Message $W_k$ if Receiver $j$ has the knowledge of $W_k$ as side information.
\end{definition}

\begin{theorem} [Outer Bound via Index Coding]
\label{theorem:Index}
  For the TIM-CoMP problem, given the topological information $\{\Tc_k,\Rc_k,\forall~k \in \Kc\}$, the DoF region is outer bounded by the capacity region of a multiple unicast index coding problem $\IC(k | \Sc_k)$, where
  \begin{align}
    \Sc_k \defeq \bigcup_{j \in \Tc_k^c} \Rc_j.
  \end{align}
\end{theorem}
\begin{proof}
See Appendix~\ref{proof:Index}.
\end{proof}

The above theorem implies that the outer bounds of the multiple unicast index coding problem in literature are still applicable to our problem, but with the modified side information sets. While the DoF region of TIM problem is outer bounded by the capacity region of the index coding problem $\IC(k | \Tc_k^c)$, our problem with transmitter cooperation is outer bounded by $\IC(k | \cup_{j \in \Tc_k^c} \Rc_j)$. In general, this bound is loose, because the side information might be over-endowed to the receivers. Nevertheless, we obtain in the following corollary that this outer bound is tight to identify the necessary and sufficient condition of the optimality of TDMA.

\begin{corollary}
\label{cor:1-K-optimal}
For the $K$-cell TIM-CoMP problem, the symmetric DoF value $d_\sym = \frac{1}{K}$ is optimal, if and only if the demand graph of the index coding problem $\IC(k | \bigcup_{j \in \Tc_k^c} \Rc_j)$ is acyclic, and more specifically, if and only if $\Gc_{AFG}$ is an empty graph.
\end{corollary}
\begin{proof}
  See Appendix~\ref{proof:1-K-optimal}.
\end{proof}
\begin{remark}
For the triangular network, the alignment-feasible graph is empty and thus the symmetric DoF value is $\frac{1}{K}$, which coincides with Corollary~\ref{cor:tri-cell}. Note that this triangular network is the minimum graph with empty alignment-feasible graph.
\end{remark}
 
In what follows, an example is presented to illustrate this corollary.

\begin{example}
\normalfont
We consider in Fig.~\ref{fig:Index_all}(a) a four-cell network with transmit sets $\Tc_1=\Tc_2=\{1,2\},\Tc_3=\Tc_4=\{1,2,3,4\}$ and receive sets $\Rc_1=\Rc_2=\{1,2,3,4\},\Rc_3=\Rc_4=\{3,4\}$. By providing Receivers 1 and 2 with $W_{3,4}$, we connect the missing links as shown in Fig.~\ref{fig:Index_all}(b) without reducing the capacity region. Allowing full CSIT, the problem now is equivalent to the index coding problem (as in Fig.~\ref{fig:Index_all}(c)) where messages $W_{1,2,3,4}$ are sent from one transmitter to Receiver $j$ $(j=1,2,3,4)$ who demands $W_j$, and both Receivers 1 and 2 have the side information $W_{3,4}$. This index coding problem has no cycles in its demand graph as shown in Fig.~\ref{fig:Index_all}(d), such that the optimal symmetric DoF value is $\frac{1}{K}$. It is also readily verified that the alignment-feasible graph is also empty, because $\Tc_i \subseteq \Tc_j$ or $\Tc_j \subseteq \Tc_i$ for any $i \ne j \in \{1,2,3,4\}$. \hfill$\square$ 

\begin{figure}[htb]
\centering
\includegraphics[width=\columnwidth]{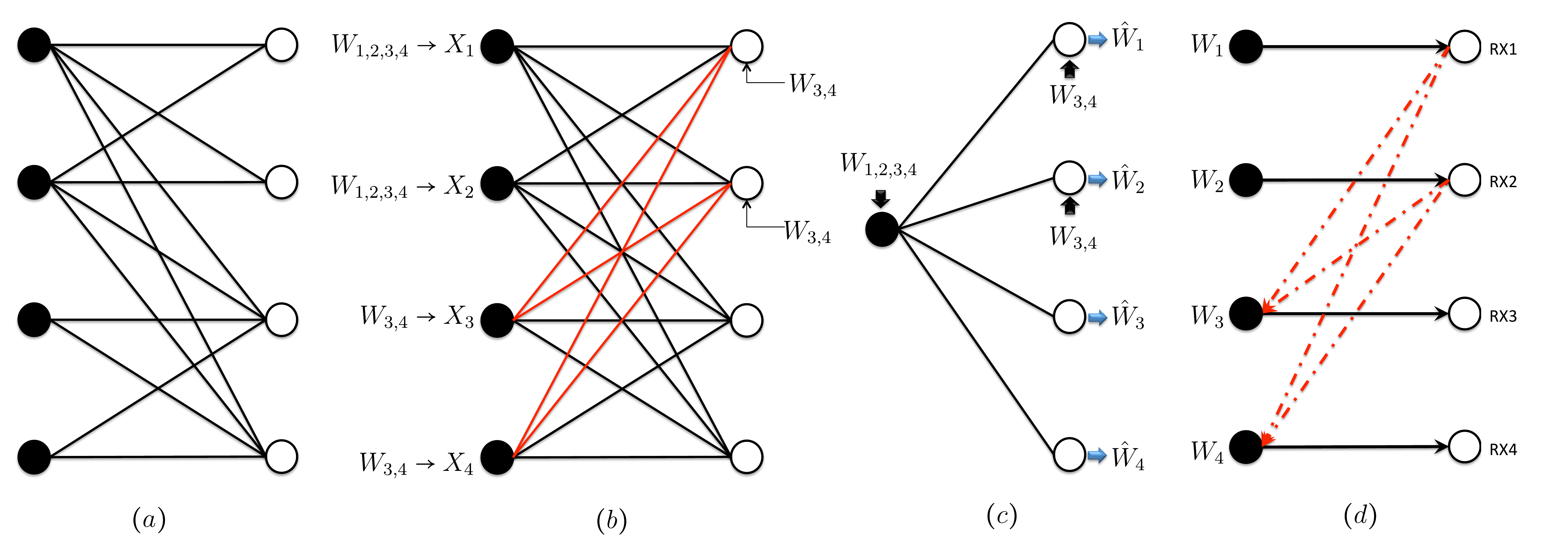}
\caption{ (a) An instance of TIM-CoMP problem $(K=4)$. By providing the side information $W_{3,4}$ to Receivers 1 and 2, the network becomes fully connected as shown in (b). Thus, the DoF region is outer bounded by the capacity region of an index coding problem with side information as in (c), whose corresponding directed demand graph is shown in (d). There exist no directed cycles in this directed graph in (d).}
\label{fig:Index_all}
\end{figure}
\end{example}

\section{Discussion}
The topological interference management problem with transmitter cooperation (i.e., TIM-CoMP problem), where a subset of messages is routed to transmitters before transmission and the transmitters only know the network topology, has been considered in this paper. 
This is the first time in our knowledge that this problem is studied and a number of preliminary results have been obtained which lay down groundwork and illustrate the potential. 
Particularly, interference management techniques under this TIM-CoMP setting are unveiled from graph theoretic and interference alignment perspectives, which exploit the benefits of both topological knowledge and transmitter cooperation. The achievable symmetric DoF are identified for a class of network topologies. The outer bounds build upon the concepts of generator sequence and compound settings to show the optimality of symmetric DoF for some special networks. The relation to index coding problem has been also investigated, with which the necessary and sufficient condition of the optimality of TDMA is also identified. 

Yet, fundamental limits of transmitter cooperation in TIM-CoMP settings are not fully understood. The optimality was only proven for some special topologies, while it demands more innovative achievability and outer bounding techniques to identify the optimality for a wider class of networks. 
 As a low-complexity achievable scheme, orthogonal access has been shown optimal for some special cases, and its optimality for general topologies is an interesting open problem. The complexity of fractional selective graph coloring prohibits the enumeration of all non-isomorphic topologies even for four-cell case, such that a potential indirect solution might be identifying the sufficient condition when orthogonal access is optimal on the network topology.
 Although there is no evidence so far showing subspace alignment outperforms one-to-one alignment, whether these two alignment strategies are equivalent or not is also an interesting problem.

Additionally, the benefit of full message sharing, where the desired message of one receiver is present at all transmitters even if some of them are disconnected to this receiver, is still unclear, although none of the findings shows gains in this regard. Further, the TIM-CoMP problems are similar to TIM problems in X networks, in which each receiver demands a message from the transmitters to which it is connected such that every message at any transmitter is useful. Nevertheless, in TIM-CoMP settings, to achieve a certain symmetric DoF, some messages are never transmitted even if they are present at the transmitters. A natural question then arises as to how much message sharing is really necessary. This question is also of practical interest, as the buffering and offloading of users' data at base stations could be significantly reduced. 

Last but not the least, the current relation to index coding problems is a bit loose in general, as the side information is overly endowed at receivers. 
A tighter relation between TIM-CoMP and index coding problems is still unclear, interesting and challenging. In addition, for TIM-CoMP problem, the necessity of nonlinear schemes is still an open problem due to the lack of tight outer bounds.

\appendix
\label{sec:proofs}

\subsection{Definitions in Graph Theory}
Throughput this paper, the graphs are simple and finite. Unless otherwise specified, the graphs are undirected. A few basic definitions pertaining to graph theory \cite{GT,Graph,Fractional2011} are now recalled.

The {\em distance} between two vertices in a graph is the minimum number of edges connecting them.
A {\em line graph} of $\Gc=(\Vc,\Ec)$ is another graph, denoted by $\Gc_e=(\Vc_e,\Ec_e)$, that represents the adjacencies of the edges in $\Gc$. In particular, each vertex $v_{ei} \in \Vc_e$ corresponds to the edge $e_i \in \Ec$ in $\Gc$, and two vertices $v_{ei}, v_{ej} \in \Ec_e$ are adjacent if and only if two edges $e_i,e_j \in \Ec$ are shared with a common endpoint in $\Gc$. A subgraph of $\Gc=(\Vc,\Ec)$ containing a subset of vertices $\Sc$ $(\Sc \subseteq \Vc)$ is said to be an {\em induced subgraph}, denoted by $\Gc[\Sc]$, if for any pair of vertices $u$ and $v$ in $\Sc$, $uv$ is an edge of $\Gc[\Sc]$ if and only if $uv$ is an edge of $\Gc$.

A $(K,d)$-{\em regular bipartite graph} $\Gc=(\Uc,\Vc,\Ec)$ is such that $\abs{\Uc}=\abs{\Vc}=K$ and $\abs{\Tc_k}=\abs{\Rc_k}=d,~\forall~k$.
A {\em Hamiltonian cycle} for a graph is a cycle that visits all vertices exactly once.
A {\em matching} of the graph is a set of edges with no common vertices between any two edges. A {\em perfect matching} is a matching contains all vertices.
The complete graph is a graph that any two vertices are joint with an edge.

A graph $\Gc$ is said to be {\em $n:m$-colorable} if each vertex in $\Gc$ can be assigned a set of of $m$ colors in which the colors are drawn from a palette of $n$ colors, such that any adjacent vertices have no colors in common. When $m=1$, $n:m$-colorable is also called $n$-colorable.
Denote by $\chi_m(\Gc)$ the minimum required number of $n$, such that the {\em fractional chromatic number} $\chi_f(\Gc)$ can be
defined as
\begin{align}
\chi_f(\Gc) = \lim_{m \to \infty} \frac{\chi_m(\Gc)}{m} = \inf_m \frac{\chi_m(\Gc)}{m}.
\end{align}

Given a graph $\Gc=(\Vc,\Ec)$ with a partition of vertices $\mathbb{V}=\{\Vc_1, \Vc_2, \cdots,\Vc_p\}$ where $\Vc_i \cap \Vc_j =\emptyset$ and $\cup_{i=1}^p \Vc_i = \Vc$, a selection of vertices $\Vc' \subseteq \Vc$ is such that $\abs{\Vc' \cap \Vc_i}=1$, $\forall~i \in \{1,2,\dots,p\}$. For an integer $k \ge 1$, $\Gc$ is {\em selectively $k$-colorable} if the induced subgraph by $\Vc'$, i.e., $\Gc[\Vc']$, is $k$-colorable.

As a {\em reference graph}, the regular bipartite graph $\Gc_r=(\Uc_r,\Vc_r,\Ec_r)$ with topology matrix $\Bm_r$ is characterized by
  \begin{align}
    [\Bm_r]_{ji} = \left\{ \Pmatrix{1, & 0 \le i-j \le d-1 \\ 0, & \text{otherwise}} \right.,
  \end{align}
  which implies $\Tc_j = \{j,j+1,\dots,j+d-1\}$.
Two bipartite graphs are said to be {\em similar}, denoted as $\Gc \simeq \Gc_r$, if their topology matrices $\Bm$ and $\Bm_r$ satisfy
$
\Bm = \Pm^T \Bm_r \Qm
$,
where $\Pm$ and $\Qm$ are permutation matrices. Accordingly, it implies that $\Uc$ and $\Vc$ in $\Gc$ can be obtained by reordering the vertices of $\Uc_r$ and $\Vc_r$ in $\Gc_r$ with $\Uc=\Uc_r$ and $\Vc=\Vc_r$.

A hypergraph $\Hc_\Gc=(\Sc,\Xc)$ associated with $\Gc$ is composed of the vertex set $\Sc \subseteq \Kc$ being a finite set, and the hyperedge set $\Xc$ being a family of subsets of $\Sc$, where $\Xc_i \defeq \{x_{i_1},x_{i_2},\dots,x_{i_{\abs{\Xc_i}}}\} \subseteq \Sc$ is called a hyperedge, i.e., $\Xc_i \in \Xc$.
A covering of a hypergraph $\Hc_\Gc$ is a collection of hyperedges $\Xc_1,\Xc_2,\dots,\Xc_\tau$ such that $\Sc \subseteq \cup_{j=1}^\tau \Xc_j$, and the least number of $\tau$ is called {\em hypergraph covering number}, denoted by $\tau(\Hc_\Gc)$. A $t$-fold covering is a multiset $\{\Xc_1,\dots,\Xc_\tau\}$ such that each $s \in \Sc$ is in at least $t$ of the $\Xc_i$'s, and correspondingly $\tau_t(\Hc_\Gc)$ is referred to as the $t$-fold covering number. Accordingly, the {\em hypergraph fractional covering number} is defined to be
\begin{align}
\tau_f(\Hc_\Gc) \defeq \lim_{t \to \infty} \frac{\tau_t(\Hc_\Gc)}{t} = \inf_{t} \frac{\tau_t(\Hc_\Gc)}{t}.
\end{align}

\subsection{Proof of Theorem~\ref{theorem:IAvoid}}
\label{proof:IAvoid}
To prove this achievability, we first build a connection between interference avoidance of TIM-CoMP problems and link scheduling problems, and then solve the link scheduling problems through graph coloring.

With transmitter cooperation enabled, it requires to schedule links rather than transmitters to avoid mutual interference.
Without transmitter cooperation, the message $W_j$ can only be sent from Transmitter $j$ for all $j$, whose activation will cause interferences to Receiver $k$ $(k \in \Rc_j)$, and consequently inactivate Transmitter $k$ $(k\in \Rc_j)$, because $W_k$ cannot be delivered from Transmitter $k$ to Receiver $k$ free of interference. The interference avoidance in this case is a matter of activating or inactivating transmitters.
In contrast, with transmitter cooperation (i.e., message sharing), the message $W_j$ can be sent from any Transmitter $i$ with $i \in \Tc_j$, and thus, it is not sufficient to schedule transmitters only. In fact, the link - rather than the transmitter - scheduling is of interest, because both the scheduling of the transmitters and the receivers does matter.\footnote{In fact, transmitter scheduling can also be regarded as link scheduling, yet only the direct links (i.e., the links from Transmitter $j$ to Receiver $j$) are candidates of link scheduling.}
For instance, if the link $e_{ij}$ (i.e., from Transmitter $i$ to Receiver $j$) is scheduled, the links adjacent to $e_{ij}$ (i.e., $e_{ik_1}$ and $e_{k_2j}$ with $k_1 \in \Rc_i \backslash j$ and $k_2 \in \Tc_j \backslash i$) as well as the links adjacent to $e_{ik_1}$ and $e_{k_2j}$ should not be scheduled, because activating Transmitter $k_2$ will interfere Receiver $j$ and Receiver $k_1$ will overhear interferences from Transmitter $i$, such that any delivery from Transmitter $k_2$ or to Receiver $k_1$ causes mutual interferences.

Such a link scheduling problem is usually solved through graph edge-coloring, while the nature of our problem calls for a more specific graph coloring solution. Let us represent the cellular network as a bipartite graph $\Gc=(\Uc,\Vc,\Ec)$, where the sets $\Uc$ and $\Vc$ denote transmitters and receivers, respectively.
The links are assigned with distinct colors if they should be scheduled at different time slots. Suppose the edge $e_{ij} \in \Ec$ receives a color. Analogously, the edges $e_{ik_1}$ and $e_{k_2j}$ with $k_1 \in \Rc_i \backslash j$ and $k_2 \in \Tc_j \backslash i$ should not be assigned the same color. Moreover, the edges adjacent to $e_{ik_1}$ and $e_{k_2j}$ should not receive the same color either. In a word, the edges within two-hop should be assigned with distinct colors. In addition, as we aim at symmetric DoF, the total number of scheduled times of the links connecting a common receiver is of interest. Thus, the number of colors received by one message should be counted by the cluster of edges that have a common vertex in $\Vc$.

As such, our problem calls for a distance-2 fractional clustered-graph edge-coloring scheme, which consists of the following ingredients:
\begin{itemize}
  \item {\em Distance-2 fractional coloring}: Both the adjacent links and the adjacency of the adjacent links (resp.~edges less than two hops) should be scheduled in difference time slots (resp.~assigned with different colors).
  \item {\em Clustered-graph coloring}: Only the total number of messages delivered via links with the common receiver (resp.~colors assigned to the edges with the same vertex) matters. Thus, the number of assigned colors should be counted by the clusters of edges.
\end{itemize}

Further, we translate the above edge-coloring of network topology $\Gc$ into vertex-coloring of its line graph $\Gc_{e}$. Accordingly, we group the vertices in $\Gc_{e}$ for which the corresponding edges in $\Gc$ have a vertex $v_j \in \Vc$ in common as a cluster, such that the number of colors is counted by clusters in $\Gc_e$. The above two-hop condition is therefore translated to a distance-2 constraint, where two vertices in $\Gc_e$ with distance less than 2 should receive different colors, and equivalently two adjacent vertices in the square of its line graph, i.e., $\Gc_e^2$, should be assigned distinct colors. Thus, the above link scheduling problem is transferable to a distance-2 selective vertex coloring problem on its line graph $\Gc_e$, and thus to a selective vertex coloring problem over $\Gc_e^2$, in which the vertices are clustered into $\mathbb{V}_e = \{\Vc_1,\dots,\Vc_K\}$ with $\Vc_k = \{e_{jk},~j \in \Tc_k\}$. Specifically, a proper selective coloring of $\Gc_e^2$ over $\mathbb{V}_e$ is a proper color assignment such that each cluster $\Vc_i$ receives $m$ colors out of in total $n$ colors and any two adjacent vertices in $\Gc_e^2$ receive distinct colors. As such, $\Gc_e^2$ is selectively $n:m$ colorable over $\mathbb{V}_e$, indicating that the links in each cluster can be scheduled $m$ times within overall $n$ time slots without causing mutual interference. Consequently, according to Definition \ref{def:dis-color}, the achievable symmetric DoF can be given by
\begin{align}
d_{\sym} = \sup_{m} \frac{m}{s\chi_m(\Gc_e^2,\mathbb{V}_e)} = \frac{1}{s\chi_f(\Gc_{e}^2,\mathbb{V}_e)}
\end{align}
where $s\chi_f$ is the fractional selective chromatic number as in Definition \ref{def:dis-color}.

\subsection{Proof of Theorem~\ref{theorem:Gen}}
\label{proof:Gen}
According to the definition of symmetric DoF, the outer bound of symmetric DoF obtained for any subset of receivers should serve as the outer bound in general. In other words, the general outer bound is the minimum value of all possible outer bounds for any subset of receivers.

Let us take a subset of receivers $\Sc \subseteq \Kc$ with received signals $Y_\Sc$ into account. For those receivers who are not considered, we switch off their desired messages from the transmitted signal, i.e., the constituent messages in transmitted signal $X_i^n$ is now comprised of message $W_j$ where $j \in \Rc_i \backslash \Sc^c$.
Define $\tilde{\Xm}^\T \defeq \big[ h_1 X_1 \ \dots \ h_K X_K \big]$, where $h_i$ $(i \in \Kc)$ is independent and identically distributed as the nonzero $h_{ji}$, and a set of virtual signals in the compact form
\begin{align}
\tilde{\Ym}_{\Ic} &\defeq \Bm_{\Ic} \tilde{\Xm}  + \tilde{\Zm}_{\Ic}\\
\bar{\Ym}_{\Ic} &\defeq \Bm_{\Ic} \Id^{\pm} \tilde{\Xm}  + \tilde{\Zm}_{\Ic}
\end{align}
for a set of receivers in $\Ic$, where $\Bm_{\Ic}$ is the submatrix of $\Bm$ with the rows out of $\Ic$ removed, $\Id^{\pm}$ is the same as the identity matrix up to the sign of elements, and $\tilde{\Ym}_{\Ic}$, $\bar{\Ym}_{\Ic}$, $\tilde{\Zm}_{\Ic}$ are vectors compacted by $\tilde{Y}_{\Ic}$, $\bar{Y}_{\Ic}$, and $\tilde{Z}_{\Ic}$, respectively. Note that $\tilde{Y}_{\Ic}$ and $\bar{Y}_{\Ic}$ are statistically equivalent to $Y_{\Ic}$, because the distribution of channel gain is symmetric around zero.
We assume there exists a generator sequence $\{\Ic_0, \Ic_1, \dots, \Ic_S\}$ with $\cup_{s=0}^S \Ic_s = \Sc$ and $\Ic_i \cap \Ic_j = \emptyset$ $\forall~i\neq j$, such that
\begin{align}
\Bm_{\Ic_s} \subseteq^{\pm} {rowspan} \left\{ \Bm_{\Ic_0}, \Id_{\Ac_{s}} \right\}, \quad \forall~s=1,\dots,S.
\end{align}
This implies that there exist $\Cm_s \in \CC^{\abs{\Ic_s} \times \abs{\Ic_0}}$ and $\Dm_s \in \CC^{\abs{\Ic_s} \times \abs{\Ac_s}}$, such that
\begin{align}
\Bm_{\Ic_s} = (\Cm_s \Bm_{\Ic_0}  + \Dm_s \Id_{\Ac_s}) \Id^{\pm} .
\end{align}
Multiplying $\Id^{\pm} \tilde{\Xm}$ at both sides yields
\begin{align}
\Bm_{\Ic_s} \Id^{\pm} \tilde{\Xm} &=  \Cm_s \Bm_{\Ic_0} \tilde{\Xm} +  \Dm_s \Id_{\Ac_s} \tilde{\Xm} \\
\Rightarrow \bar{\Ym}_{\Ic_s} &= \Cm_s \tilde{\Ym}_{\Ic_0} +  \Dm_s \Id_{\Ac_s} \tilde{\Xm} + \tilde{\Zm}_{\Ic_s} - \Cm_s \tilde{\Zm}_{\Ic_0} \\
&=  \Cm_s \tilde{\Ym}_{\Ic_0} +  \Dm_s  \tilde{\Xm}_{\Ac_s} + \tilde{\Zm}_{\Ic_s} - \Cm_s \tilde{\Zm}_{\Ic_0}\\
&= \Cm_s \tilde{\Ym}_{\Ic_0} +  \Dm_s  \tilde{\Xm}_{\Ac_s} - \bar{\Zm}_s \label{eq:gen0}
\end{align}
with $\bar{\Zm}_s \defeq \Cm_s \Zm_{\Ic_0} - \Zm_{\Ic_s} $ being the entropy-bounded noise term \cite{Avestimehr:2013TIM}.
Thus, according to the mapping $\mathbf f_{\idx}: \Bc \mapsto \{0,1\}^K$ and the definition of $\Ac_s$, we have
\begin{align}
H(W_{\Ic_s} | \tilde{\Ym}_{\Ic_0}^n, \cup_{r=0}^{s-1} W_{\Ic_r}, \Hc^n, \Gc ) &= H(W_{\Ic_s} | \tilde{\Ym}_{\Ic_0}^n, \cup_{r=0}^{s-1} W_{\Ic_r}, X_{\Ac_s}, \Hc^n, \Gc ) \label{eq:gen1} \\
&= H(W_{\Ic_s} | \tilde{\Ym}_{\Ic_0}^n, \bar{\Ym}_{\Ic_s}^n+\bar{\Zm}_s^n, \cup_{r=0}^{s-1} W_{\Ic_r}, X_{\Ac_s}, \Hc^n, \Gc ) \label{eq:gen2}\\
&\le H(W_{\Ic_s} | \bar{\Ym}_{\Ic_s}^n + \bar{\Zm}_s^n, \Hc^n, \Gc ) \label{eq:gen3}\\
&= H(W_{\Ic_s} | \tilde{\Ym}_{\Ic_s}^n + \bar{\Zm}_s^n, \Hc^n, \Gc ) \label{eq:gen4} \\
&\le n \epsilon_n + n \cdot O(1)
\end{align}
where \eqref{eq:gen1} is due to the fact that $X_i^n$ is encoded only from $W_{\Rc_i \backslash \Sc^c}$, \eqref{eq:gen2} comes from \eqref{eq:gen0} where $\tilde{\Xm}_{\Ac_s}$ can be constructed from $X_{\Ac_s}$, \eqref{eq:gen3} is because removing conditioning does not reduce entropy, \eqref{eq:gen4} is due to the argument that $\tilde{Y}$ and $\bar{Y}$ are statistically equivalent, and the last inequality is obtained by following the fact that $H(W|Y^n+\bar{Z}^n) \le n \epsilon_n + n \cdot O(1)$, if $H(W|Y^n+Z^n) \le n \epsilon_n$ \cite{Avestimehr:2013TIM}, since $\bar{\Zm}_s$ is bounded noise term.
Further, we have
\begin{align}
n\sum_{i \in \Sc} R_i &= H(W_{\Sc} | \Hc^n, \Gc)\\
&= I(W_{\Sc}; \tilde{\Ym}_{\Ic_0}^n | \Hc^n, \Gc) + H(W_{\Sc} | \tilde{\Ym}_{\Ic_0}^n, \Hc^n, \Gc)\\
&= I(W_{\Sc}; \tilde{\Ym}_{\Ic_0}^n | \Hc^n, \Gc) + H(W_{\Ic_0} | \tilde{\Ym}_{\Ic_0}^n, \Hc^n, \Gc) + H(W_{\Sc \backslash \Ic_0} | \tilde{\Ym}_{\Ic_0}^n, W_{\Ic_0}, \Hc^n, \Gc)  \\
&\le n |\Ic_0| \log P + n \cdot O(1) +n \epsilon_n + \sum_{s=1}^S H(W_{\Ic_s} | \tilde{\Ym}_{\Ic_0}^n, \cup_{r=0}^{s-1} W_{\Ic_r}, \Hc^n, \Gc) \\
&\le n |\Ic_0| \log P + n \cdot O(1) + n \epsilon_n.
\end{align}
By the definition of symmetric DoF, we have
\begin{align}
d_\sym \le  \lim_{P \to \infty} \frac{R_\sym}{\log P} = \frac{\abs{\Ic_0}}{\abs{\Sc}}.
\end{align}
Among all possible subsets of $\Sc$ and initial generator $\Ic_0$, the symmetric DoF should be outer-bounded by the minimum of them. Thus, we have
\begin{align}
d_\sym \le \min_{\Sc \subseteq \Kc} \min_{\Ic_0 \subseteq \Jc(\Sc)} \frac{\abs{\Ic_0}}{\abs{\Sc}}.
\end{align}

\subsection{Proof of Corollary~\ref{cor:3-cell}}
\label{proof:3-cell}
Enumerating all the possible topologies of three-cell networks, we verify the optimality of symmetric DoF by comparing the achievability in Theorem \ref{theorem:IAvoid} and the outer bound in Theorem \ref{theorem:Gen}. It is readily verified that all but two topologies have enhanced symmetric DoF, compared to the case without transmitter cooperation \cite{graph_scheduling,Jafar:2013TIM,Avestimehr:2013TIM}. As shown in Fig.~\ref{fig:3-cell}, message sharing improves the symmetric DoF from $\frac{1}{2}$ to $\frac{2}{3}$ for the topology $(i)$ and from $\frac{1}{3}$ to $\frac{1}{2}$ for the topology $(m)$.
\begin{figure}[htb]
\centering
\includegraphics[width=1\columnwidth]{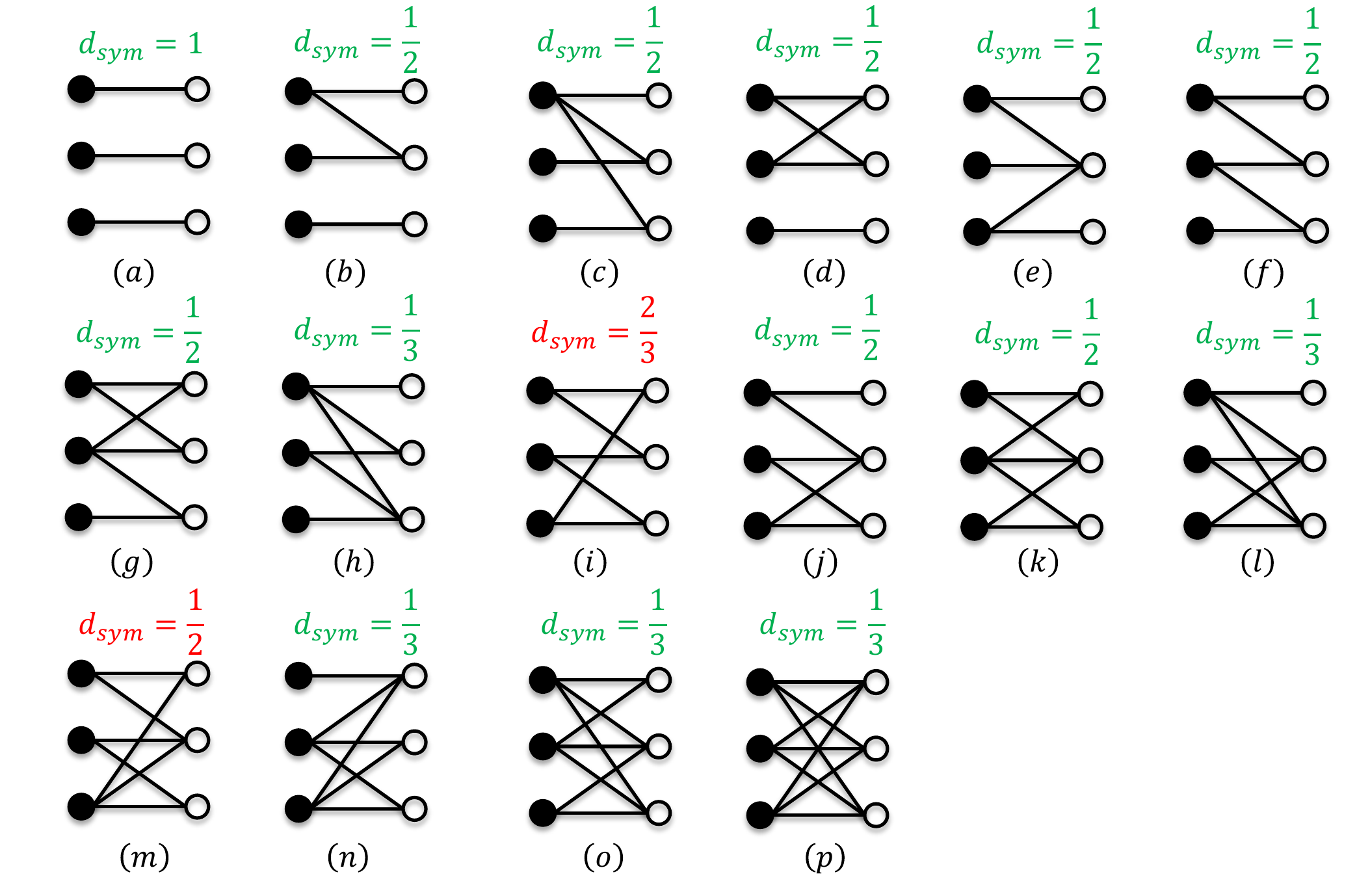}
\caption{The three-cell TIM-CoMP problem, where all non-isomorphic topologies are enumerated. The symmetric DoF improvement over the noncooperation case is due to topologies $(i)$ and $(m)$.}
\label{fig:3-cell}
\end{figure}

For the achievability, two graph coloring realizations are illustrated in Fig. \ref{fig:3-cell-exs} concerning the topologies of $(i)$ and $(m)$. Specifically, every cluster receives two out of three colors in total in $(i)$, and one out of two colors in $(m)$, where the conditions of distance-2 fractional selective graph coloring are satisfied, yielding achievable symmetric DoF $d_\sym = \frac{2}{3}$ and $d_\sym = \frac{1}{2}$, respectively. For other topologies, the achievability can be similarly obtained.
\begin{figure}[htb]
\centering
\includegraphics[width=0.6\columnwidth]{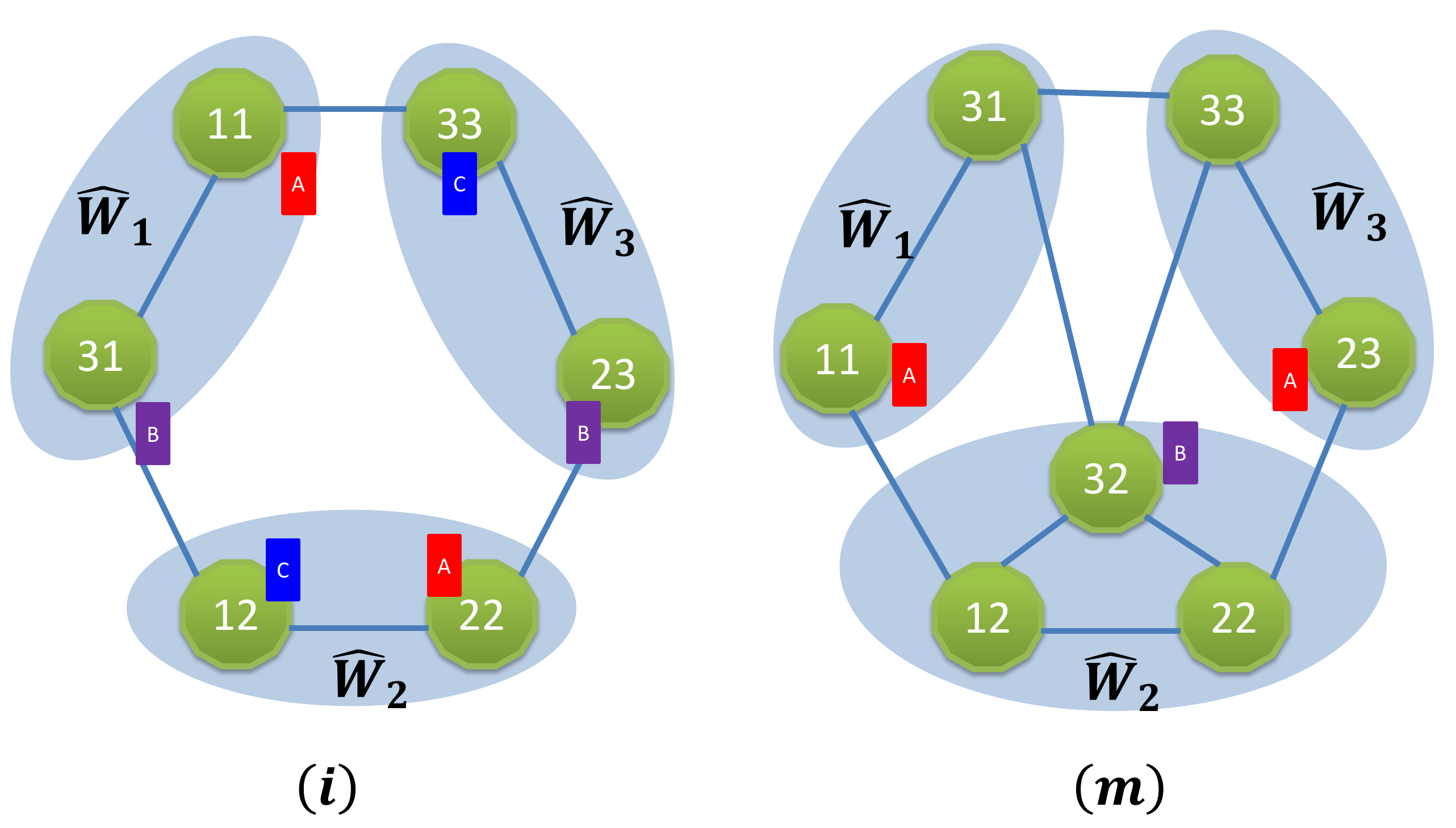}
\caption{Fractional selective graph coloring of the topologies $(i)$ and $(m)$. It requires three colors to ensure every cluster receive two in $(i)$, and two colors are sufficient to offer every cluster one color in $(m)$.}
\label{fig:3-cell-exs}
\end{figure}

Regarding the converse, we apply the outer bound via generator sequence here. Again, we take those two topologies for example. For topology-$(i)$, we have a generator sequence $\{\{1,2\},\{3\}\}$ with $\Ic_0=\{1,2\}$ and $\Ic_1=\{3\}$. By generating the virtual signals $\tilde{Y}_1^n = h_1^n X_1^n + h_3^n X_3^n + \tilde{Z}_1^n$ and $\tilde{Y}_2^n = h_1^n X_1^n + h_2^n X_2^n + \tilde{Z}_2^n$, which are statistically equivalent to $Y_1^n$ and $Y_2^n$ respectively, we obtain $\tilde{Y}_3^n = \tilde{Y}_1^n - \tilde{Y}_2^n =  h_3^n X_3^n -h_2^n X_2^n + \tilde{Z}_1^n-\tilde{Z}_2^n $ that is statistically equivalent to $Y_3^n$ with a bounded noise difference \cite{Avestimehr:2013TIM}. Thus, according to Theorem \ref{theorem:Gen}, we have $d_\sym \le \frac{\abs{\Ic_0}}{\abs{\Sc}} = \frac{2}{3}$. Similarly for topology-$(m)$, we have a generator sequence $\{\{2\},\{1\}\}$ with $\Ic_0=\{2\}$ and $\Ic_1 = \{1\}$. Note that we ignore the received signal at Receiver 3, and therefore eliminate the message $W_3$ from the message sets of the respective transmitted signals. Thus, the message sets of Transmitters 1, 2, and 3 become
$\{W_1,W_2\}, \{W_2\}$, and $\{W_1,W_2\}$, 
respectively.
Following the generator sequence approach, we initiate the generator sequence by a virtual signal $\tilde{Y}_2^n = h_1^n X_1^n + h_2^n X_2^n + h_3^n X_3^n + \tilde{Z}_2^n$, and successively generate $\tilde{Y}_1^n = \tilde{Y}_2^n - h_2^n X_2^n$, where $X_2$ can be encoded from the message $W_2$.
Hence, the symmetric DoF outer bound is $d_\sym \le \frac{\abs{\Ic_0}}{\abs{\Sc}} = \frac{1}{2}$.

Being aware of the coincidence of the achievability and the outer bounds, we conclude that the interference avoidance achieves the optimal symmetric DoF. The optimality verification of other topologies can be similarly done. 

\subsection{Proof of Corollary~\ref{cor:tri-cell}}
\label{proof:tri-cell}

For the converse proof, since the lower and upper triangular matrices are similar, it suffices to consider the lower triangular matrix $\Bm$ without loss of generality, where $\Tc_j = \{1,\dots,j\}$ for all $j \in \Kc$. Thus, the message sets to $X_j$ with transmitter cooperation are comprised of $W_{\{j,\cdots,K\}}$. It is readily verified that $\{\{K\},\{K-1\},\dots,\{1\}\}$ forms a generator sequence with $\Ic_0=\{K\}$ and $\Sc=\Kc$. Thus, we have the outer bound $d_\sym \le \frac{\abs{\Ic_0}}{\abs{\Sc}} = \frac{1}{K}$, which is achievable by time division. This completes the proof.

\subsection{Proof of Theorem~\ref{theorem:Arb}}
\label{proof:Arb}
\subsubsection{$d_\sym = \frac{2}{K}$ is achievable}

First, we consider the case when there exists a Hamiltonian cycle, say without loss of generality $1 \leftrightarrow 2 \leftrightarrow \dots \leftrightarrow K \leftrightarrow 1$, in the alignment-feasible graph $(\Gc_{AFG})$.
According to the definition of $\Gc_{AFG}$,  it follows that, there exist $z_{j}^1$ and $z_{j+1}^2$, such that
\begin{align}
z_{j}^1 \in \Tc_{j} \cap \Tc_{{j+1}}^c, \quad \text{and} \quad z_{j+1}^2 \in \Tc_{{j+1}} \cap \Tc_{j}^c
\end{align}
with $z_{j}^1, z_{j}^2 \in \Tc_{j}$ and $z_{j-1}^1, z_{j+1}^2 \notin \Tc_{j}$, for $j \in \Kc$.
Thus, we send along the direction $\Vm_j \in \CC^{K \times 1}$ two signals $X_{z_j^1}(W_{j}^1)$ and $X_{z_{j+1}^2}(W_{{j+1}}^2)$ from Transmitter $z_j^1$ and Transmitter $z_{j+1}^2$, respectively, for $j \in \Kc$. 

On one hand, if the channel coefficients are constant during the communication, the received signals at Receiver $j$ during $K$ time slots can be given as a compact form by
{\small
\begin{align}
\Ym_{j} &= \sum_{s=1}^K \Vm_s \left(h_{j,z_s^1} X_{z_s^1} (W_{s}^1) \mathbf{1}(z_s^1 \in \Tc_{j}) + h_{j,z_{s+1}^2} X_{z_{s+1}^2} (W_{{s+1}}^2)  \mathbf{1}(z_{s+1}^2 \in \Tc_{j}) \right) \\
&= \underbrace{\Vm_j h_{j,z_j^1} X_{z_j^1} (W_{j}^1) + \Vm_{j-1} h_{j,z_{j}^2} X_{z_{j}^2} (W_{j}^2)}_{\DS} \nn \\ & + \underbrace{ \sum_{s=1,s \neq j-1,j}^K \Vm_s \left(h_{j,z_s^1} X_{z_s^1} (W_{s}^1) \mathbf{1}(z_s^1 \in \Tc_{j}) + h_{j,z_{s+1}^2} X_{z_{s+1}^2} (W_{{s+1}}^2)  \mathbf{1}(z_{s+1}^2 \in \Tc_{j}) \right)}_{\AI}
\end{align}
}
where $ \mathbf{1}(\cdot)$ is the indicator function with value 1 if the parameter is true and 0 otherwise. It is readily verified that two symbols $W_{j}^1$ and $W_{j}^2$ can be retrieved almost surely, yielding symmetric DoF of $\frac{2}{K}$. On the other hand, if the channel is time-varying, we can simply choose $\Vm_j$ as the $j$-th column of $\Id_K$, and the same symmetric DoF can be achieved. In this case, interference alignment boils down to interference avoidance.

Second, we consider a perfect matching in $\Gc_{AFG}$ where $K$ is even, say $1 \leftrightarrow 2, \dots, {K-1} \leftrightarrow K$. Similarly, there exist $z_{j}$ and $z_{j+1}$, such that
\begin{align}
z_{j} \in \Tc_{j} \cap \Tc_{{j+1}}^c, \quad \text{and} \quad z_{j+1} \in \Tc_{{j+1}} \cap \Tc_{j}^c, \quad j=1,3,\dots,K-1
\end{align}
with $z_{j} \in \Tc_{j}$ and $ z_{j+1} \notin \Tc_{j}$.
Thus, during in total $\frac{K}{2}$ time slots, we send two signals $X_{z_j}(W_{j})$ and $X_{z_{j+1}}(W_{{j+1}})$ from Transmitter $z_j$ and Transmitter $z_{j+1}$, respectively, with the same precoder $\Vm_j \in \CC^{\frac{K}{2} \times 1}$. The received signals at Receiver $j$ during $\frac{K}{2}$ time slots can be similarly written as
{\small
\begin{align}
\Ym_{j} &= \sum_{s=1}^{\frac{K}{2}} \Vm_s \left(h_{j,z_s} X_{z_s} (W_{s}) \mathbf{1}(z_s \in \Tc_{j}) + h_{j,z_{s+1}} X_{z_{s+1}} (W_{{s+1}})  \mathbf{1}(z_{s+1} \in \Tc_{j}) \right) \\
&= \underbrace{\Vm_j h_{j,z_j} X_{z_j} (W_{j}) }_{\DS}  \nn \\ &+ \underbrace{ \sum_{s=1,s \neq j}^{\frac{K}{2}} \Vm_s \left(h_{j,z_s} X_{z_s} (W_{s}) \mathbf{1}(z_s \in \Tc_{j}) + h_{j,z_{s+1}} X_{z_{s+1}} (W_{{s+1}})  \mathbf{1}(z_{s+1} \in \Tc_{j}) \right)}_{\AI}
\end{align}
}
with which the message $W_{j}$ is recovered, yielding $\frac{2}{K}$ DoF per user. This completes the proof.

\subsubsection{$d_\sym = \frac{2}{K-q}$ is achievable}
The achievability is similar to the previous case, but the duration of transmission is shortened. Without loss of generality, we assume the Hamiltonian cycle $1 \leftrightarrow 2  \leftrightarrow \dots \leftrightarrow K \leftrightarrow 1$ for the brevity of presentation. According to the definition of alignment-feasible graph, there exist $z_{s}^1$ and $z_{s+1}^2$, such that
\begin{align}
z_{s}^1 \in \Tc_{s} \cap \Tc_{s + 1}^c, \quad &\text{and} \quad z_{s+1}^2 \in \Tc_{s+1} \cap \Tc_{s}^c
\end{align}
with $z_{s}^1 \in \Tc_{s}$ and $z_{s+1}^2 \notin \Tc_{s}$, for $s \in \Kc$.
Assuming
\begin{align}
k_0 \in \arg \min_{k} \sum_j \Am_{kj},
\end{align}
we have $ \sum_{j} \Am_{k_0j} = q$ and thus
\begin{align}
\mathbf{f}_{\idx}^{-1}(\Am_{k_0}^\T) = \{j_1,\dots,j_q\} \defeq \Jc_q
\end{align}
where $\mathbf{f}_{\idx}^{-1}: \{0,1\}^K \mapsto \Bc$ is the inverse function of $\mathbf{f}_{\idx}$.

According to the definition of alignment non-conflict matrix,  if $\Am_{k_0j}=1$, then
\begin{align}
 \Tc_{i_j} \bigcap \Tc_{i_{j+1}}^c \bigcap_{k:\Am_{kj}=1} \Tc_{k_0}^c \neq \emptyset, \quad \text{and} \quad \Tc_{i_{j+1}} \bigcap \Tc_{i_{j}}^c \bigcap_{k:\Am_{kj}=1} \Tc_{k_0}^c \neq \emptyset,
\end{align}
meaning that there is non-conflict to make $W_{i_j}$ and $W_{i_{j+1}}$ aligned with the occupied subspace absent to Receiver $k_0$.
It follows that, there exist $z_{j_t}^1$ and $z_{j_t+1}^2$ $(j_t \in \Jc_q)$, such that
\begin{align}
z_{j_t}^1 \in \Tc_{j_t} \cap \Tc_{j_t + 1}^c \cap \Tc_{k_0}^c , \quad &\text{and} \quad z_{j_t+1}^2 \in \Tc_{j_t+1} \cap \Tc_{j_t}^c \cap \Tc_{k_0}^c
\end{align}
with $z_{j_t}^1,z_{j_t+1}^2 \notin \Tc_{k_0}$.
We send $X_{z_s^1}(W_{s}^1)$ and $X_{z_{s+1}^2}(W_{s+1}^2)$ at Transmitter $z_s^1$ and Transmitter $z_{s+1}^2$, respectively, along with the subspace spanned by $\Vm_s \in \CC^{(K-q) \times 1}$. Given channel coherence time $\tau_c \ge K-q$, the received signal at Receiver $k_0$ can  be written as

\vspace{-15pt}
{\small
\begin{align}
\Ym_{k_0} &= \sum_{s=1}^{K} \Vm_s \left(h_{k_0,z_s^1} X_{z_s^1} (W_{s}^1) \mathbf{1}(z_s^1 \in \Tc_{k_0}) + h_{k_0,z_{s+1}^2} X_{z_{s+1}^2} (W_{s+1}^2)  \mathbf{1}(z_{s+1}^2 \in \Tc_{k_0}) \right) \nn \\
&= \sum_{s=1,s \notin \Jc_q}^{K} \Vm_s \left(h_{k_0,z_s^1} X_{z_s^1} (W_{s}^1) \mathbf{1}(z_s^1 \in \Tc_{k_0}) + h_{k_0,z_{s+1}^2} X_{z_{s+1}^2} (W_{s+1}^2)  \mathbf{1}(z_{s+1}^2 \in \Tc_{k_0}) \right)\\
&= \Vm_{k_0} h_{k_0,z_{k_0}^1} X_{z_{k_0}^1} (W_{k_0}^1) + \Vm_{k_0-1} h_{k_0,z_{k_0}^2} X_{z_{k_0}^2} (W_{k_0}^2) \nn \\ &+ \sum_{\substack{s=1,s \notin \Jc_q,\\ s \neq k_0-1, k_0}}^{K} \Vm_s \left(h_{k_0,z_s^1} X_{z_s^1} (W_{s}^1) \mathbf{1}(z_s^1 \in \Tc_{k_0}) + h_{k_0,z_{s+1}^2} X_{z_{s+1}^2} (W_{s+1}^2)  \mathbf{1}(z_{s+1}^2 \in \Tc_{k_0}) \right)
\end{align}
}
It follows that the desired messages by Receiver $k_0$ can be recovered in a $K-q$ dimensional space with two interference-free subspace and $K-q-2$ dimensional subspace with interferences aligned. According to the definition of $q$, we conclude that the overall $K-q$ dimensional space is sufficient to support other receivers with $\sum_j \Am_{kj} \ge q$. As such, the symmetric DoF $\frac{2}{K-q}$ are achievable.

\subsection{Proof of Theorem~\ref{theorem:Reg}}
\label{proof:Reg}
According to the definition of $(K,d)$-regular networks, we have $\abs{\Tc_j}=d$, $\forall~j \in \Kc$. As we know, when $d=K$, the network is fully connected and therefore the optimal symmetric DoF value is $\frac{1}{K}$ by time division. So, in what follows, we will consider the general achievability proof when $d \le K-1$.

Since the cellular network graph is assumed to be similar to the reference one by reordering the transmitters and/or receivers, we directly consider the referred network topology, because they are equivalent in terms of symmetric DoF with transmitter cooperation. For the referred network topology, the transmit set of Receiver $j$ is given by
\begin{align}
  \Tc_j = \{j,j+1,\dots,j+d-1\},
\end{align}
where all the receiver indices are modulo $K$, e.g., $j-K=j$ and $0=K$.
Thus, at Transmitter $i$ we send symbols with careful design
\begin{align*}
  \Xm_i = \Vm_{i+1} X_{i}(W_{i}^1) + \Vm_{i+2} X_{i}(W_{i-d+1}^2),  \forall~i=1,\dots,K
\end{align*}
where $\{\Vm_i,~i=1,\dots,K\}$ are $(d+1) \times 1$ random vectors, and linearly independent among any $(d+1)$ vectors, almost surely, $X_i(W_j)$ is the signal transmitted from Transmitter $i$ carrying on message $W_j$, and $W_j^1$, $W_j^2$ are two realizations (symbols) of message $W_j$.
The signals at Receiver $j$ during $d+1$ time slots, with coherence time $\tau_c \ge d+1$, can be compacted as
\begin{align*}
  \Ym_j &= \sum_{i \in \Tc_j} h_{ji} \Xm_i + \Zm_j\\
  &= \sum_{i =j}^{j+d-1} h_{ji} (\Vm_{i+1} X_{i}(W_{i}^1) + \Vm_{i+2} X_{i}(W_{i-d+1}^2)) + \Zm_j\\
  &= h_{j,j} \Vm_{j+1} X_{j}(W_{j}^1) + h_{j,j+d-1} \Vm_{j+d+1} X_{j+d-1}(W_{j}^2) \nn \\ & \qquad \qquad \qquad+ \sum_{i =j+1}^{j+d-1} h_{ji} \Vm_{i+1} X_{i}(W_{i}^1) + \sum_{i =j}^{j+d-2} h_{ji} \Vm_{i+2} X_{i}(W_{i-d+1}^2) + \Zm_j\\
  &= \underbrace{h_{j,j} \Vm_{j+1} X_{j}(W_{j}^1) + h_{j,j+d-1} \Vm_{j+d+1} X_{j+d-1}(W_{j}^2)}_{\DS}  \nn \\ & \qquad \qquad \qquad+ \underbrace{\sum_{i =j}^{j+d-2} \Vm_{i+2} (h_{j,i+1}  X_{i+1}(W_{i+1}^1) +  h_{j,i} X_{i}(W_{i-d+1}^2))}_{\AI} + \Zm_j.
\end{align*}
It is readily shown that the interferences occupy $d-1$ dimensional subspace out of the total $d+1$ dimensional space, leaving 2-dimensional interference-free subspace spanned by $\{\Vm_{j+1}, \Vm_{j+d+1}\}$ to the desired signals, such that the desired messages for Receiver $j$, $W_j^1$ and $W_j^2$, can be successfully recovered. This philosophy applies to all other receivers. During $d+1$ time slots, every receiver can decode two messages, yielding symmetric DoF of $\frac{2}{d+1}$.

\begin{figure}[htb]
\centering
\includegraphics[width=0.55\columnwidth]{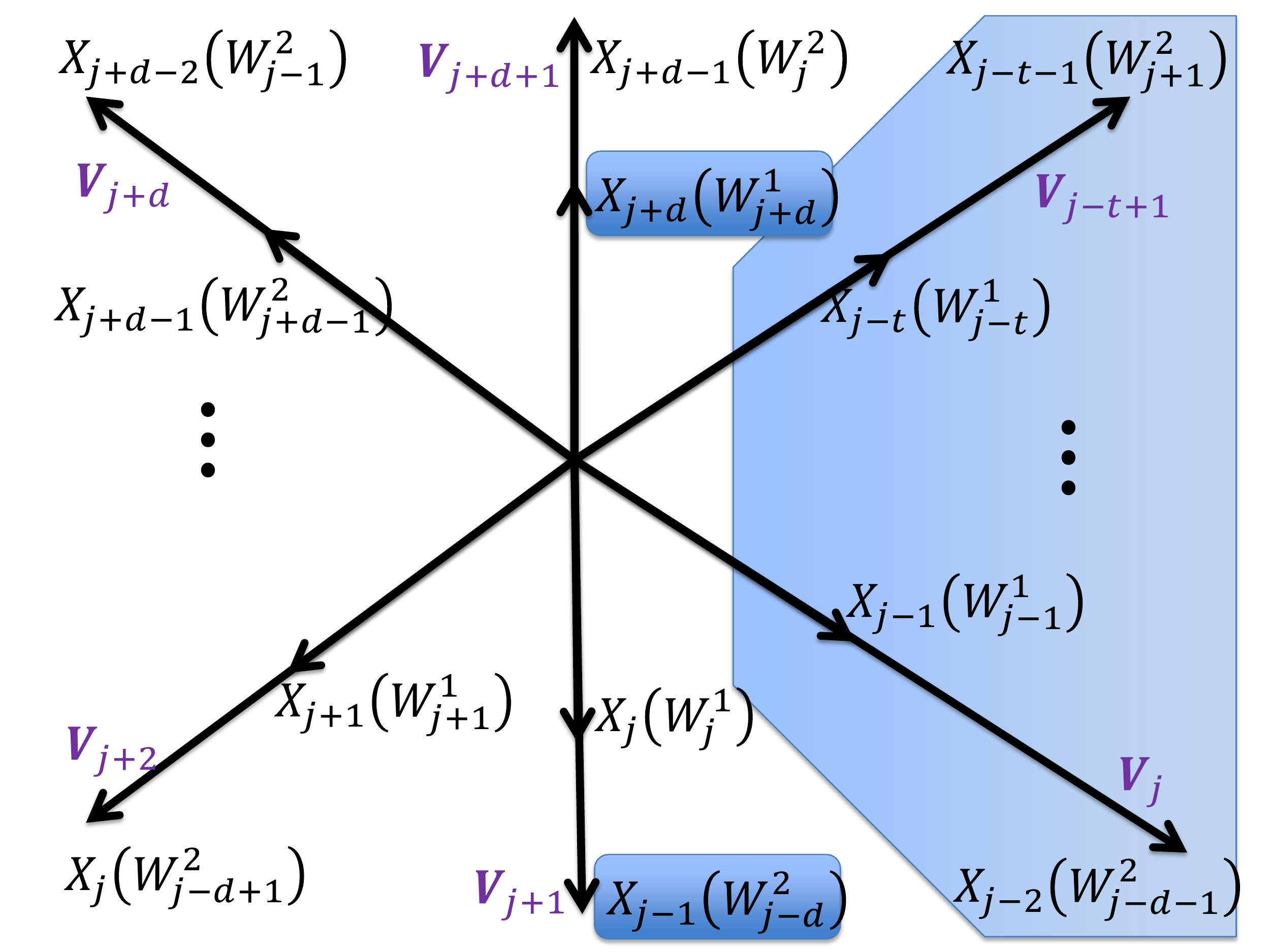}
\caption{Interference alignment for the general $(K,d)$ regular cellular networks.}
\label{fig:general}
\end{figure}

Geometrically, the interference alignment can be shown in Fig.~\ref{fig:general}, and also interpreted as follows. Transmitted signals $X_{j-1}(W_{j-1}^1)$ and $X_{j-2}(W_{j-d-1}^2)$ are aligned in the same subspace spanned by vector $\Vm_j$, which is absent to Receiver $k$ $(k \in \{j, \dots, j+K-3\})$. Note that $t \defeq K-d-1$ and $j-t=j+d+1$ modulo $K$. By deduction, the subspace spanned by $\{\Vm_{j+d+2}, \dots, \Vm_j\}$ are absent to Receiver $j$ (i.e., the shadow in Fig.~\ref{fig:general}), leaving $d+1$ linearly independent vectors $\{\Vm_{j+1},\dots,\Vm_{j+d+1}\}$ to span the space.  As such, the signal carrying $X_j(W_j^1)$ is aligned with $X_{j-1}(W_{j-d}^2)$ in the subspace spanned by $\Vm_{j+1}$, and $X_{j+d-1}(W_j^2)$ is aligned with $X_{j+d}(W_{j+d}^1)$ in the subspace spanned by $\Vm_{j+d+1}$. Note that the signals from Transmitter $j-1$ and $j+d$ cannot be heard by Receiver $j$ according to the network topology, such that $W_j^1$ and $W_j^2$ are free of interference, and retrievable from overall $d+1$ dimensional subspace.

It is worth noting that, although the message $W_j$ is shared among the transmitters $i~(\forall~i \in \Tc_j)$, its two realizations $W_j^1$ and $W_j^2$ are only utilized in this scheme by Transmitter~$j$ and Transmitter~$(j-d+1)$, respectively.

\subsection{Proof of Theorem~\ref{theorem:Com}}
\label{proof:Com}
In what follows, we present an outer bound with the aid of compound settings. As illustrated in Example \ref{ex:compound}, it is necessary to determine the least required compound receivers such that the noisy versions of $X_i$ can be recovered. Thus, we first look into this problem, given that a subset of messages is known {\em a priori}.

Consider a set of receivers $\Sc \subseteq \Kc$ satisfying $\cup_{j \in \Sc} \Tc_j = \Kc$. The received signals $Y_j$ $(j \in \Sc)$ at Receiver $j$ is a linear combination of $\{X_i, i \in \Tc_j\}$ polluted by noise. To recover the noisy versions of $\{X_i, i \in \Kc\}$, it requires at most $K-\abs{\Sc}$ extra linearly independent equations, which can be provided by compound receivers that are assume to be possessing the same topology as the original receivers and demanding the same messages. In the rest of the proof, we do not distinguish the original from the compound receivers explicitly.

In fact, in the present of a set of messages $W_\Sc$, the required number of compound receivers can be further reduced. According to transmitter cooperation, the transmitted signal $X_i^n$ is encoded with the messages $W_{\Rc_i}$. Being aware of $W_{\Rc_i}$, we are able to reconstruct the transmitted signals $X_{\Sc'}^n$, where
\begin{align}
\Sc' = \{i | {\Rc_i} \subseteq \Sc \}.
\end{align}
In other words, the knowledge of $W_{\Sc}$ is equivalent to that of $X_{\Sc'}^n$. With $X_{\Sc'}^n$, we can remove their contributions from the received signals. Denote by $Y_{j,i}$ and $\tilde{Y}_{j,i}$ the received signals of the $i$-th compound receiver of Receiver $j$ before and after removing the contribution of $X_{\Sc'}$, respectively, i.e.,
\begin{align}
Y_{j,i} &= \sum_{k \in \Tc_i} h_{j,i,k} X_k + Z_{j,i}\\
\tilde{Y}_{j,i} &= \sum_{k \in \Tc_i \backslash \Sc'} h_{j,i,k} X_k + Z_{j,i}.
\end{align}
where $h_{j,i,k}$ is drawn from the same distribution and independent of $h_{j,k}(t)$.
Let $\Tc_j'$ be the set of the least required compound receivers associated with Receiver $j$.
Thus, we collect all the compound signals and compact them as
\begin{align}
\tilde{\Ym}_{\Tc'_\Sc} = \Hm_{\Tc'_\Sc} \Xm_{\Kc\backslash\Sc'} + \Zm_{\Tc'_\Sc}
\end{align}
where $\Hm_{\Tc'_\Sc} \in \CC^{\sum_{j \in \Sc}\abs{\Tc_j'}  \times (K-\abs{\Sc'})}$ is the reduced channel matrix with the columns indexed by $\Sc'$ removed. It suffices to recover $X_{\Kc \backslash \Sc'}$ from $\tilde{\Ym}_{\Tc'_\Sc}$  as long as $\sum_{j \in \Sc} \abs{\Tc_j'} \ge K-\abs{\Sc'} $. We conclude that the required number of compound receivers can be reduced to $K-\abs{\Sc'}-\abs{\Sc}$, given the knowledge of $W_\Sc$.

Secondly, we proceed to present the outer bound of achievable rates of compound receivers. For the $i$-th compound receiver of Receiver $j$, by Fano's inequality, we have
\begin{align}
n (R_{j,i} - \epsilon_n) &\le I(W_j, Y_{j,i}^n | \Hc^n, \Gc) \\
&= h(Y_{j,i}^n | \Hc^n, \Gc) - h(Y_{j,i}^n | W_j, \Hc^n, \Gc)\\
&\le n \log P - h(Y_{j,i}^n | W_j, \Hc^n, \Gc) + n \cdot O(1)
\end{align}
where $R_{j,i}$ denotes the achievable rate of the $i$-th compound receiver, and is the same as $R_j$. Let $\sum_{j \in \Sc} \abs{\Tc_j'} = K-\abs{\Sc'} $. By adding all achievable rates of all compound receivers, we have
\begin{align}
\MoveEqLeft n \left(\sum_{j \in \Sc} \sum_{i \in \Tc_j'} R_{j,i} - \epsilon_n \right) \\
&\le n \sum_{j \in \Sc} \abs{\Tc_j'} \log P - h(\{Y_{j,i}^n, j\in \Sc, i\in \Tc_j'\} | W_\Sc, \Hc^n, \Gc) + n \cdot O(1) \label{eq:proof_com1}\\
&= n \sum_{j \in \Sc} \abs{\Tc_j'} \log P - h(\{Y_{j,i}^n, j\in \Sc, i\in \Tc_j'\} | W_\Sc, X_{\Sc'}^n, \Hc^n, \Gc) + n \cdot O(1) \label{eq:proof_com2}\\
&= n \sum_{j \in \Sc} \abs{\Tc_j'} \log P - h(\tilde{\Ym}_{\Tc'_\Sc}^n | W_\Sc, X_{\Sc'}^n, \Hc^n, \Gc) + n \cdot O(1) \label{eq:proof_com3}\\
&= n \sum_{j \in \Sc} \abs{\Tc_j'} \log P - h(\Xm_{\Kc \backslash \Sc'}^n + \Hm_{\Tc'_\Sc}^{-1} \Zm_{\Tc'_\Sc}^n | W_\Sc, X_{\Sc'}^n, \Hc^n, \Gc) + n \cdot O(1) \label{eq:proof_com4}\\
&= n \sum_{j \in \Sc} \abs{\Tc_j'} \log P - h(X_{\Kc\backslash\Sc'}^n + \bar{Z}_{\Kc\backslash\Sc'}^n | W_\Sc, X_{\Sc'}^n, \Hc^n, \Gc) + n \cdot O(1) \label{eq:proof_com5}\\
&= n \sum_{j \in \Sc} \abs{\Tc_j'} \log P - n \sum_{j \in  \Sc^c} R_j + n \cdot O(1)
\end{align}
where \eqref{eq:proof_com2} is due to the fact that the knowledge of $W_\Sc$ is equivalent to the knowledge of $X_{\Sc'}$ given topological information, \eqref{eq:proof_com3} is because translation does not change the differential entropy, \eqref{eq:proof_com4} is obtained because $\Hm_{\Tc'_\Sc}$ is a square matrix and has full rank almost surely, in \eqref{eq:proof_com5}, $\bar{Z}_{\Kc\backslash\Sc'}^n$ is the bounded noise terms, and the last inequality is from the decodable condition similar to that in \eqref{eq:decodable}.
By the definition of the symmetric DoF, it follows that
\begin{align}
d_\sym &\le \frac{\sum_{j \in \Sc} \abs{\Tc_j'}}{\sum_{j \in \Sc} \abs{\Tc_j'} + \abs{\Sc^c}}\\
&= \frac{K-\abs{\Sc'}}{2K-\abs{\Sc'} - \abs{\Sc}}.
\end{align}
Among all the possible $\Sc$, we have the outer bound of symmetric DoF
\begin{align}
d_\sym \le \min_{\Sc \subseteq \Kc} \frac{K-\abs{\Sc'}}{2K-\abs{\Sc'} - \abs{\Sc}}
\end{align}
where $\Sc$ and $\Sc'$ are subject to two constraints: $\cup_{j \in \Sc} \Tc_j = \Kc$ and $\Sc' = \{i |  \Rc_i \subseteq \Sc \}$.

\subsection{Proof of Corollary~\ref{cor:Wyner}}
\label{proof:Wyner}
When $K=2$, the network is fully connected and $d_\sym=\frac{1}{2}$ is optimal. So, in the rest of the proof, we focus on $K \ge 3$. From the graph theoretic perspective, any two $(K,2)$-regular networks are similar, because they are in fact the same cycle with rearranged vertices. Hence, it suffices to consider one typical topology of the $(K,2)$-regular networks, e.g., a $K$-cell cyclic Wyner network, for the convenience of presentation. The received signal at Receiver $j$ of the $K$-cell cyclic Wyner model can be given as
\begin{align}
  Y_j = h_{j,j-1} X_{j-1} + h_{j,j} X_j + Z_j
\end{align}
where the indices are modulo $K$, and $W_i, W_{i+1}$ are the only accessible messages to Transmitter $i$. In what follows, we will present first the converse, followed by the achievability proof.

\subsubsection{Converse}
We consider two cases when $K$ is even or odd.
\begin{itemize}
  \item $K$ is even: Let $\Sc=\{1,3,\dots,K-1\}$ and $\Sc'=\emptyset$. Consider the received signals $Y_\Sc$ and the signals of their respective compound receivers $\tilde{Y}_\Sc$. Following the proof of the general case, we have
      \begin{align}
        2n \sum_{j \in \Sc} R_j 
       &\le n K \log P - h(Y_\Sc^n,\tilde{Y}_\Sc^n | W_\Sc, \Hc^n, \Gc)\\
        &= n K \log P - n (R_2 + R_4 + \dots + R_K) + n \cdot O(1)
      \end{align}
      where the noisy version $\{X_i^n+\bar{Z}_i^n, i \in \Kc\}$ can be recovered from $K$ linearly independent equations.
      Thus, with $\abs{\Sc}=\frac{K}{2}$ and $\abs{\Sc'}=0$, it follows that
      \begin{align}
        d_\sym \le \frac{K}{K+K/2} = \frac{2}{3}.
      \end{align}
  \item $K$ is odd: Let $\Sc=\{1,3,\dots,K-2,K\}$ and $\Sc'=\{K\}$. Consider here the received signals $Y_\Sc$ and the signals of their respective compound receivers $\tilde{Y}_{\Sc \backslash \{K-2,K\}}$. Similarly, we have
      \begin{align}
        \MoveEqLeft 2n \sum_{j \in {\Sc \backslash \{K-2,K\}}} R_j + n R_{K-2} + n R_K \\
        &\le (K-1) \log P - h(Y_\Sc^n, \tilde{Y}_{\Sc \backslash \{K-2,K\}}^n | W_\Sc, \Hc^n, \Gc)\\
        &= n(K-1) \log P - h(Y_\Sc^n, \tilde{Y}_{\Sc \backslash \{K-2,K\}}^n | W_\Sc, X_K^n, \Hc^n, \Gc)\\
        &= n(K-1) \log P - n(R_2+R_4+\dots+R_{K-1})
      \end{align}
      where $X_K^n$ is reproducible with $W_1$ and $W_k$, and the noisy version $\{X_i^n+\bar{Z}_i^n, i \in \Kc \backslash K\}$ can be recovered from $K-1$ linearly independent equations.
      Thus, with $\abs{\Sc}=\frac{K+1}{2}$ and $\abs{\Sc'}=1$, it follows that
      \begin{align}
        d_\sym \le \frac{K-1}{K-1 + \frac{K-1}{2}} = \frac{2}{3}.
      \end{align}
\end{itemize}
To sum up, we have $d_\sym \le \frac{2}{3}$ whenever $K$ is even or odd.

\subsubsection{Achievability}
Although the general achievability proof has been presented with general $d$, we make it concrete here for $d=2$.
During three time slots, we send at Transmitter~$i$
\begin{align}
\Xm_i = \Vm_{i-1} X_i(W_{i+1}^1) + \Vm_{i} X_i(W_i^2)
\end{align}
where $\{\Vm_i, i=1,\dots,n\}$ are $3 \times 1$ vectors satisfy that any three of them are linearly independent, almost surely. At Receiver~$j$, we have
\begin{align}
  \Ym_{j} &= h_{j,j-1} \Xm_{j-1} + h_{j,j} \Xm_{j} + \Zm_{j}\\
  &= h_{j,j-1} (\Vm_{j-2} X_{j-1}(W_{j}^1) + \Vm_{j-1} X_{j-1}(W_{j-1}^2)) \nn \\ & \qquad + h_{j,j} (\Vm_{j-1} X_{j}(W_{j+1}^1) + \Vm_{j} X_{j}(W_{j}^2)) + \Zm_{j}\\
  &= \underbrace{h_{j,j-1} \Vm_{j-2} X_{j-1}(W_{j}^1) + h_{j,j} \Vm_{j} X_{j}(W_{j}^2)}_{\DS}  \nn \\ &\qquad + \underbrace{\Vm_{j-1}(h_{j,j-1} X_{j-1}(W_{j-1}^2) + h_{j,j} X_{j}(W_{j+1}^1))}_{\AI} + \Zm_{j}.
\end{align}
The interferences carrying messages $W_{j-1}$ and $W_{j+1}$ are aligned together in the direction of $\Vm_{j-1}$, leaving two-dimensional interference-free subspace for desired signals carrying on message realizations $W_j^1$ and $W_j^2$. Therefore, two messages are delivered during three time slots, yielding $\frac{2}{3}$ DoF per user, which coincides with the outer bound. This completes the proof of optimality.

\subsection{Proof of Theorem~\ref{theorem:Arb-more}}
\label{proof:Arb-more}
\subsubsection{$d_\sym = \frac{1}{\kappa}$ is achievable}
\label{proof:pp}
According to the definition of proper partition, for a portion $\Pc_i=\{i_1,\dots,i_{p_i}\}$, we assume with $k=1,\dots,p_i$ that
\begin{align}
       z_{i_k} \in   \Tc_{i_k} \bigcap \left(\bigcup_{i_j \in \Pc_i \backslash i_k} \Tc_{i_j}\right)^c , \quad \forall~i_k \in \Pc_i.
\end{align}
with $z_{i_k} \in \Tc_{i_k}$ and $z_{i_k} \notin \Tc_{i_j}, ~\forall~j \ne k$.
Thus we send $\{X_{z_{i_k}}(W_{i_k}), k=1,\dots,p_i\}$ at Transmitter $z_{i_k}$ via the same precoder $\Vm_i \in \CC^{\kappa \times 1}$, yielding the receiver signal at Receiver $i_k$ in a block fading channel (e.g., $\tau_c \ge \kappa$) as
\begin{align}
\Ym_{i_k} &= \sum_{j=1}^\kappa \Vm_j \left( \sum_{s=1}^{p_j} h_{i_k,z_{j_s}} X_{z_{j_s}}(W_{j_s}) \mathbf{1}(z_{j_s} \in \Tc_{i_k}) \right) \\
&= \Vm_i  h_{i_k,z_{i_k}} X_{z_{i_k}}(W_{i_k}) + \Vm_i \left( \sum_{s=1,s \neq k}^{p_i} h_{i_k,z_{i_s}} X_{z_{i_s}}(W_{i_s}) \mathbf{1}(z_{i_s} \in \Tc_{i_k}) \right) \nn \\ & \qquad+ \sum_{j=1,j \neq i}^\kappa \Vm_j \left( \sum_{s=1}^{p_j} h_{i_k,z_{j_s}} X_{z_{j_s}}(W_{j_s}) \mathbf{1}(z_{j_s} \in \Tc_{i_k}) \right) \\
&= \underbrace{ \Vm_i  h_{i_k,z_{i_k}} X_{z_{i_k}}(W_{i_k})}_{\DS} + \underbrace{\sum_{j=1,j \neq i}^\kappa \Vm_j \left( \sum_{s=1}^{p_j} h_{i_k,z_{j_s}} X_{z_{j_s}}(W_{j_s}) \mathbf{1}(z_{j_s} \in \Tc_{i_k}) \right)}_{\AI}
\end{align}
with which the desired signal can be retrieved with high probability during $\kappa$ time slots. This applies to all messages and offers $\frac{1}{\kappa}$ DoF per user. For the time-varying channel (i.e., $\tau_c=1$), by setting $\Vm_i$ to be the $i$-th column of $\Id_{\kappa}$, the same symmetric DoF are still achievable. This confirms our argument that interference alignment is a general form of interference avoidance.

\subsubsection{$d_\sym = \frac{1}{\kappa-q}$ is achievable}

The achievability is similar to the previous case, but the required number of subspace dimension is reduced. By assuming similarly
\begin{align}
m \in \arg \min_{i} \sum_j \Am_{ij},
\end{align}
we have $\sum_{j} \Am_{mj} = q$
and
$
\mathbf{f}_{\idx}^{-1}(\Am_{m}^\T) = \Jc_q.
$

According to the definition of proper partition, there exists $z_{i_k}$ with $i \in \{1,\dots,\kappa\}$ such that
\begin{align}
z_{i_k} &\in \Tc_{i_k} \bigcap \left(\bigcup_{i_j \in \Pc_i \backslash i_k} \Tc_{i_j}\right)^c, \quad \forall~i_k \in \Pc_i
\end{align}
with $z_{i_k} \in \Tc_{i_k}, \forall~i$,
and according to the alignment non-conflict matrix, if $\Am_{mj}=1$, then
\begin{align}
 \Tc_{j_{t}} \bigcap \left(\bigcup_{j_{s} \in \Pc_{j} \backslash j_{t}} \Tc_{j_{s}}\right)^c \bigcap_{i:\Am_{mj}=1} \Tc_{m_k}^c \neq \emptyset, \quad \forall~j_t \in \Pc_j, \forall~m_k \in \Pc_m,
\end{align}
meaning that this is non-conflicting to make the messages in portion $\Pc_j$ aligned with the spanned subspace absent to all the receivers in $\Pc_m$.
It follows that,  there exists $z_{j_t}$ with $j \in \Jc_q$, such that
\begin{align}
z_{j_t} &\in \Tc_{j_{t}} \bigcap \left(\bigcup_{j_{s} \in \Pc_{j} \backslash j_{t}} \Tc_{j_{s}}\right)^c \bigcap \Tc_{m_{k}}^c, \quad \forall~m_{k} \in \Pc_{m}, j_t \in \Pc_j
\end{align}
with $z_{j_t} \notin \Tc_{m_k}$, $\forall~m_k \in \Pc_{m} , j_t \in \Pc_j$. With channel coherence time $\tau_c \ge \kappa-q$, the channel coefficients keep constant throughout the communication. As such, the received signal at Transmitter $m_k$ with $m_k \in \Pc_{m}$ can be given as
\begin{align}
\Ym_{m_k} &= \sum_{l=1}^\kappa \Vm_l \left( \sum_{s=1}^{p_l} h_{m_k,z_{l_s}} X_{z_{l_s}}(W_{l_s}) \mathbf{1}(z_{l_s} \in \Tc_{m_k}) \right) \\
&= \sum_{l=1,l \notin \Jc_q}^\kappa \Vm_l \left( \sum_{s=1}^{p_l} h_{m_k,z_{l_s}} X_{z_{l_s}}(W_{l_s}) \mathbf{1}(z_{l_s} \in \Tc_{m_k}) \right) \\
&= \Vm_{m} h_{m_k,z_{m_k}} X_{z_{m_k}}(W_{m_k}) \mathbf{1}(z_{m_k} \in \Tc_{m_k}) \nn \\& \qquad \qquad  + \Vm_{m} \left( \sum_{s=1,s \neq k}^{p_{m}} h_{m_k,z_{m_s}} X_{z_{m_s}}(W_{m_s}) \mathbf{1}(z_{m_s} \in \Tc_{m_k}) \right) \nn \\ & \qquad \qquad \qquad \qquad+ \sum_{l=1, l \neq m, l \notin \Jc_q}^\kappa \Vm_l \left( \sum_{s=1}^{p_l} h_{m_k,z_{l_s}} X_{z_{l_s}}(W_{l_s}) \mathbf{1}(z_{l_s} \in \Tc_{m_k}) \right) \\
&= \Vm_{m} h_{m_k,z_{m_k}} X_{z_{m_k}}(W_{m_k}) \nn \\ &\qquad \qquad + \sum_{l=1, l \neq m, l \notin \Jc_q}^\kappa \Vm_l \left( \sum_{s=1}^{p_l} h_{m_k,z_{l_s}} X_{z_{l_s}}(W_{l_s}) \mathbf{1}(z_{l_s} \in \Tc_{m_k}) \right)
\end{align}
where $\Vm_l \in \CC^{(\kappa-q) \times 1}$ is sufficient to recover desired message $W_{m_k}$, yielding $\frac{1}{\kappa-q}$ DoF. According to the definition of $q$, this $\kappa-q$ dimensional space suffices to support all other receivers. Thus, symmetric DoF of $\frac{1}{\kappa-q}$ are achievable, almost surely.  

\subsection{Proof of Theorem~\ref{theorem:Cov}}
\label{proof:Cov}
In this theorem, we represent the achievable symmetric DoF of the TIM-CoMP problem by a graph-theoretic parameter, i.e., fractional covering number. To this end, we will bridge our problem to the hypergraph fractional covering problem, which is in general a set covering problem.

First of all, we construct such a hypergraph $\Hc_\Gc$ according to the network topology.
From the definition of proper partition, it follows that if a set $\Xc_i \defeq \{{i_1},{i_2},\dots,{i_{\abs{\Xc_i}}}\} \subseteq \Kc$ satisfies
\begin{align}
  \Tc_{{i_k}} \bigcap \left(\bigcup_{{i_j} \in \Xc_i \backslash {i_k}} \Tc_{{i_j}}\right)^c \neq \emptyset, \quad \forall~{i_k} \in \Xc_i, \label{eq:cover_con}
\end{align}
then any two messages in $W_{\Xc_i}$ are mutually alignment-feasible. The messages $\{W_{{i_k}},{i_k} \in \Xc_i\}$ can be sent from the transmitters $\{z_{i_k},{i_k} \in \Xc_i\}$, respectively, in the form of $X_{z_{i_k}}(W_{{i_k}})$ with the same precoding vector $\Vm_{i}$ (Alternatively, the links from Transmitter $z_{i_k}$ to Receiver ${i_k}$ $(k=1,\dots,\abs{\Xc_i})$ can be scheduled at the same time slot), where
\begin{align}
  z_{i_k} \in \Tc_{{i_k}} \bigcap \left(\bigcup_{{i_j} \in \Xc_i \backslash {i_k}} \Tc_{{i_j}}\right)^c.
\end{align}
As such, only one transmitted signal $X_{z_{i_k}}(W_{{i_k}})$ is active in subspace spanned by $\Vm_i$ at Receiver ${{i_k}}$, and hence $W_{{i_k}}$ is recoverable from this subspace. The presence of the subspace spanned by $\Vm_{i}$ carrying on messages with indices in $\Xc_i$ guarantees the successful delivery of the messages in $W_{\Xc_i}$. Thus, the set $\Xc_i$ can serve as a hyperedge of $\Hc_\Gc$.
Any set of elements in $\Kc$ that satisfies the condition \eqref{eq:cover_con} serves as a hyperedge. As a result, the hypergraph $\Hc_\Gc$ is constructed with vertex set $\Kc$ and hyperedge set being enumeration of all possible sets of elements that satisfy condition \eqref{eq:cover_con}.

Our problem to find the symmetric DoF is equivalent to the covering problem of this hypergraph to find the minimum number of hyperedges $\{\Xc_i, i=1,2,\dots,\tau_t\}$ such that each $j \in \Kc$ appears at least $t$ of the $\Xc_i$'s.
According to the definition of hypergraph covering in Appendix A, the minimum number of hyperedges that meets the covering problem can be represented by the $t$-fold covering number. With hyperedge cover of $\tau_t(\Hc_\Gc)$ times, each vertex in $\Kc$ is covered at least $t$ times, meaning that within a $\tau_t(\Hc_\Gc)$ dimensional subspace spanned by $\Vm_i \in \CC^{\tau_t \times 1}, i=1,\dots,\tau_t$, each $W_j, j\in \Kc$ can be delivered $t$ times free of interference. As a consequence, the achievable symmetric DoF can be represented by
\begin{align}
  d_\sym = \sup_t \frac{t}{\tau_t(\Hc_\Gc)} = \frac{1}{\tau_f(\Hc_\Gc)},
\end{align}
where $\tau_f(\Hc_\Gc)$ is the hypergraph fractional covering number as defined in Appendix A.

\subsection{Proof of Theorem~\ref{theorem:Index}}
\label{proof:Index}
The proof follows the channel enhancement approach in \cite{Jafar:2013TIM} with slight modification by taking transmitter cooperation (i.e., message sharing) into account.
We brief the steps of the channel enhancement as follows.
\begin{itemize}
\item
Denote by $\Cc_1$ the capacity region of the TIM-CoMP problem, where Transmitter $i$ is endowed with the messages desired by its connected receivers, i.e., $W_{\Rc_i}$, for all $i \in \Kc$.
  \item $\forall~k,j \in \Kc$, if $j \in \Tc_k$, we specify
  \begin{align}
    h_{kj} = \sqrt{\frac{SNR}{P_j}}
  \end{align}
  which will not impact on the reliability of the capacity-achieving coding scheme.
  \item $\forall~k,j \in \Kc$, if $j \notin \Tc_k$, we provide $W_{\Rc_j}$ to Receiver $k$ as side information, and connect the missing link by setting the channel coefficient as a non-zero value
  \begin{align}
    h_{kj} = \sqrt{\frac{SNR}{P_j}},
  \end{align}
  where the newly enabled interferences from Transmitter $j$ can be eliminated given the side information $W_{\Rc_j}$.
  \item Allowing full transmitter cooperation and full CSIT, the channel turns out an MISO channel to each receiver, where all received signals are statistically equivalent. Denote by $\Cc_2$ the capacity region of current channel. The capacity region is not diminished, i.e., $\Cc_1 \subseteq \Cc_2$.
  \item With the network equivalence theorem \cite{Koetter}, the MISO channel can be replaced by a noise-free link with finite capacity, as the bottleneck link of index coding problem with capacity region $\Cc_3$.
\end{itemize}
It is noticed that all the above steps do not reduce the capacity region, i.e., $\Cc_1 \subseteq \Cc_2 \subseteq \Cc_3$, such that the capacity region of the index coding problem with side information $\cup_{j \in \Tc_k^c} W_{\Rc_j}$ can serve as an outer bound of our problem.

\subsection{Proof of Corollary~\ref{cor:1-K-optimal}}
\label{proof:1-K-optimal}

First, we prove the sufficient condition that, if the demand graph of index coding problem $\IC(k | \Sc_k)$ with $\Sc_k=\cup_{j \in \Tc_k^c} \Rc_j$ is acyclic or $G_{AFG}$ is an empty graph, then the optimal symmetric DoF value is $\frac{1}{K}$. To this end, we only need to prove the following chain:
\begin{align*}
\text{$\Gc_{AFG}=\emptyset$ $\Rightarrow$ acyclic demand graph $\IC(k | \Sc_k)$ $\Rightarrow$ $\frac{1}{K}$ is optimal}.
\end{align*}

Being aware of the fact that the symmetric DoF $\frac{1}{K}$ can be trivially achieved by time division, we only have to prove $\frac{1}{K}$ is also an outer bound. From Corollary 7 in \cite{Jafar:2013TIM}, the necessary and sufficient condition to achieve symmetric capacity of $\frac{1}{K}$ per message is that the message demand graph is acyclic. Thus, if the demand graph is acyclic, then the TIM-CoMP problem is upper bounded by $\frac{1}{K}$. That is, the second part of the chain is true. By this, the sufficiency can be proved if the first part of the chain is also true.
We construct the proof of this statement by contraposition, i.e.,
\begin{align*}
&\text{if the demand graph of the index coding problem $\IC(k|\Sc_k)$ is not acyclic,} \\  &\text{then $\Gc_{AFG} \neq \emptyset$, i.e., $\exists~i,~j$, such that $\Tc_i \nsubseteq \Tc_j$ and $\Tc_j \nsubseteq \Tc_i$.}
\end{align*}

To prove this contraposition, we first consider there exists a cycle involving only two messages, e.g., $W_m$ and $W_n$, in the demand graph. Thus, we have $m \in \Sc_n = \cup_{j \in \Tc_{n}^c} \Rc_j$ and $n \in \Sc_m = \cup_{j \in \Tc_{m}^c} \Rc_j$, while $m \notin \Sc_m = \cup_{j \in \Tc_{m}^c} \Rc_j$ and $n \notin \Sc_n = \cup_{j \in \Tc_{n}^c} \Rc_j$, such that there exist $j_1 \in \Tc_{n}^c$ and $j_2 \in \Tc_{m}^c$ where $m \in \Rc_{j_1}$ and $n \in \Rc_{j_2}$ whereas $m \notin \Rc_{j_2}$ and $n \notin \Rc_{j_1}$. This leads to $\Rc_{j_1} \nsubseteq \Rc_{j_2}$ and $\Rc_{j_2} \nsubseteq \Rc_{j_1}$. Equivalently, there exist $t_1 \in \Rc_{j_1}$ and $t_2 \in \Rc_{j_2}$, such that $\Tc_{t_1} \nsubseteq \Tc_{t_2}$ and $\Tc_{t_2} \nsubseteq \Tc_{t_1}$, because both conditions imply the same alignment feasibility, where $X_{j_1}(W_{t_1})$ and $X_{j_2}(W_{t_2})$ can be aligned in the same subspace. Consequently, two messages $W_{t_1}$ and $W_{t_2}$ are alignment-feasible, and therefore connected in $\Gc_{AFG}$. Thus, $\Gc_{AFG} \neq \emptyset$ is proven.

Furthermore, we consider the smallest cycle involving more than two messages, i.e., $i_1 \to i_2 \to \dots \to i_s \to i_1$ with directed edge from Message $i_m$ to Receiver $i_{m}$ then via Message $i_{m+1}$ to Receiver $i_{m+1}$ and so on, for $m=1,2,\dots,s$, with modulo applied to the indices. Given the smallest cycle in the directed demand graph, we have 
\begin{align}
i_{m+1} &\in \Sc_{i_m} = \cup_{j \in \Tc_{i_m}^c} \Rc_j, \label{eq:indx-in-set} \\
i_{m+1} & \notin \Sc_{i_n} = \cup_{j \in \Tc_{i_n}^c} \Rc_j, \quad \forall~n \in \{1,2,\dots,s\}, ~~n \ne m. \label{eq:indx-notin-set}
\end{align}
From \eqref{eq:indx-in-set}, it is readily verified that there must exist $j_m \in \Tc_{i_m}^c$, such that $i_{m+1} \in \Rc_{j_m}$. By set $n=m-1$ and $n=m+1$ respectively in \eqref{eq:indx-notin-set}, we have
\begin{align}
\forall&~j_{m-1} \in \Tc_{i_{m-1}}^c, \quad i_{m+1} \notin \Rc_{j_{m-1}},\\
\forall&~j_{m+1} \in \Tc_{i_{m+1}}^c, \quad i_{m+1} \notin \Rc_{j_{m+1}}.
\end{align}
It follows that $\Tc_{i_m}^c \nsubseteq \Tc_{i_{m-1}}^c$ and $\Tc_{i_m}^c \nsubseteq \Tc_{i_{m+1}}^c$  for all $m$, and in turn
\begin{align}
\Tc_{i_{m-1}} \nsubseteq \Tc_{i_{m}}, \quad \text{and} \quad \Tc_{i_{m+1}} \nsubseteq \Tc_{i_{m}}, \quad \forall~m.
\end{align}
Otherwise, it results in contradictions with $i_{m+1} \in \Rc_{j_m}$. Recalling that $i_1 \to i_2 \to \dots \to i_s \to i_1$ forms a cycle, we conclude that $\Tc_{i_m} \nsubseteq \Tc_{i_{m+1}}$ and $\Tc_{i_{m+1}} \nsubseteq \Tc_{i_{m}}$ for all $m \in \{1,2,\dots,s\}$, and therefore messages $W_{i_m}$ and $W_{i_{m+1}}$ are joint with an edge. Thus, we conclude that, if there exist a cycle in demand graph, then $\Gc_{AFG} \neq \emptyset$. 

Consequently, its contraposition is equivalently proven: if $\Gc_{AFG} = \emptyset$, then the corresponding demand graph is acyclic. Thus, the first part of the chain is true, and in turn the necessity is proven.

Second, we prove the necessary condition that, if the optimal symmetric DoF value is $\frac{1}{K}$, then the demand graph is acyclic and $\Gc_{AFG}$ is an empty graph. We achieve this goal by constructing a proof by contraposition, i.e., if $\Gc_{AFG}$ is not empty, or the demand graph of index coding problem $\IC(k | \Sc_k)$ is cyclic, then the symmetric DoF value $\frac{1}{K}$ is sub-optimal. To this end, we only need to prove the following chain
\begin{align*}
\text{ Cyclic demand graph $\IC(k | \Sc_k)$ $\Rightarrow$ $\Gc_{AFG} \neq \emptyset$ $\Rightarrow$ $\frac{1}{K}$ is suboptimal.}
\end{align*}
As proved above, if $\Gc_{AFG} = \emptyset$, then the demand graph of $\IC(k|\Sc_k)$ is acyclic. By contraposition, if the demand graph is cyclic, then $\Gc_{AFG} \neq \emptyset$. That is, the first part of the chain is true. Let us focus on the second part of the chain.
Assume there exists an edge $e_{ij}$ in $\Gc_{AFG}$, which implies that $W_i$ and $W_j$ are alignment feasible, i.e., $\Tc_i \nsubseteq \Tc_j$ and $\Tc_i \nsubseteq \Tc_j$. 
According to the definition of proper partition, we have a partition with size $K-1$ where $W_i$ and $W_j$ belong to one portion and the rest $K-2$ messages form $K-2$ portions, respectively, such that $d_\sym=\frac{1}{K-1}$ is achievable. Thus, the statement of the second part of the chain is automatically implied. Had proven the chain, the necessity is obtained. 

Given the necessity and sufficiency, the proof is completed.

\section*{acknowledgement}
Insightful discussion with Prof. Syed A.~Jafar and Hua Sun at The University of California, Irvine is thankfully acknowledged. The authors would like to gratefully acknowledge many insightful discussion with the Associate Editor and timely, insightful and constructive feedback received from the anonymous reviewers to help improve both the technical content and the presentation of this paper.

%\bibliography{/Users/yix/Documents/Research/References/References}

\begin{thebibliography}{10}
\providecommand{\url}[1]{#1}
\csname url@samestyle\endcsname
\providecommand{\newblock}{\relax}
\providecommand{\bibinfo}[2]{#2}
\providecommand{\BIBentrySTDinterwordspacing}{\spaceskip=0pt\relax}
\providecommand{\BIBentryALTinterwordstretchfactor}{4}
\providecommand{\BIBentryALTinterwordspacing}{\spaceskip=\fontdimen2\font plus
\BIBentryALTinterwordstretchfactor\fontdimen3\font minus
  \fontdimen4\font\relax}
\providecommand{\BIBforeignlanguage}[2]{{%
\expandafter\ifx\csname l@#1\endcsname\relax
\typeout{** WARNING: IEEEtran.bst: No hyphenation pattern has been}%
\typeout{** loaded for the language `#1'. Using the pattern for}%
\typeout{** the default language instead.}%
\else
\language=\csname l@#1\endcsname
\fi
#2}}
\providecommand{\BIBdecl}{\relax}
\BIBdecl

\bibitem{Yi:2014ISIT}
X.~Yi and D.~Gesbert, ``Topological interference management with transmitter
  cooperation,'' in \emph{IEEE International Symposium on Information Theory
  Proceedings (ISIT2014)}, Honolulu, HI, Jun. 2014.

\bibitem{Lozano:Coop}
A.~Lozano, R.~Heath, and J.~Andrews, ``Fundamental limits of cooperation,''
  \emph{IEEE Trans. Inf. Theory}, vol.~59, no.~9, pp. 5213--5226, 2013.

\bibitem{Jafar:IA}
S.~A. Jafar, \emph{Interference alignment: A new look at signal dimensions in a
  communication network}.\hskip 1em plus 0.5em minus 0.4em\relax Now Publishers
  Inc, 2011.

\bibitem{Gesbert:2010JSAC}
D.~Gesbert, S.~Hanly, H.~Huang, S.~Shamai~Shitz, O.~Simeone, and W.~Yu,
  ``Multi-cell {MIMO} cooperative networks: A new look at interference,''
  \emph{IEEE J. Sel. Areas Commun.}, vol.~28, no.~9, pp. 1380--1408, Dec. 2010.

\bibitem{Jindal:2006}
N.~Jindal, ``{MIMO} broadcast channels with finite-rate feedback,'' \emph{IEEE
  Trans.~Inf.~Theory}, vol.~52, no.~11, pp. 5045 --5060, Nov. 2006.

\bibitem{Love:2008}
D.~Love, R.~Heath, V.~Lau, D.~Gesbert, B.~Rao, and M.~Andrews, ``An overview of
  limited feedback in wireless communication systems,'' \emph{IEEE
  J.~Sel.~Areas in Commun.}, vol.~26, no.~8, pp. 1341 --1365, Oct. 2008.

\bibitem{Caire:2010}
G.~Caire, N.~Jindal, M.~Kobayashi, and N.~Ravindran, ``Multiuser {MIMO}
  achievable rates with downlink training and channel state feedback,''
  \emph{IEEE Trans. Inf. Theory}, vol.~56, no.~6, pp. 2845--2866, Jun. 2010.

\bibitem{Lapidoth:2006}
A.~Lapidoth, S.~Shamai, and M.~Wigger, ``On the capacity of fading {MIMO}
  broadcast channels with imperfect transmitter side-information,'' \emph{arXiv
  preprint cs/0605079}, 2006.

\bibitem{MAT}
M.~Maddah-Ali and D.~Tse, ``Completely stale transmitter channel state
  information is still very useful,'' \emph{IEEE Trans. Inf. Theory}, vol.~58,
  no.~7, pp. 4418--4431, Jul. 2012.

\bibitem{Vaze:2012IC}
C.~Vaze and M.~Varanasi, ``The degrees of freedom region and interference
  alignment for the {MIMO} interference channel with delayed {CSIT},''
  \emph{IEEE Trans. Inf. Theory}, vol.~58, no.~7, pp. 4396--4417, Jul. 2012.

\bibitem{SalmanDelayed}
S.~Lashgari, A.~Avestimehr, and C.~Suh, ``Linear degrees of freedom of the x
  channel with delayed csit,'' \emph{Information Theory, IEEE Transactions on},
  vol.~60, no.~4, pp. 2180--2189, April 2014.

\bibitem{Yang:2013MISO}
S.~Yang, M.~Kobayashi, D.~Gesbert, and X.~Yi, ``Degrees of freedom of time
  correlated {MISO} broadcast channel with delayed {CSIT},'' \emph{IEEE
  Trans.~Inf.~Theory,}, vol.~59, no.~1, pp. 315--328, 2013.

\bibitem{Gou:2012Mixed}
T.~Gou and S.~Jafar, ``Optimal use of current and outdated channel state
  information: Degrees of freedom of the {MISO BC} with mixed {CSIT},''
  \emph{IEEE Communications Letters}, vol.~16, no.~7, pp. 1084--1087, Jul.
  2012.

\bibitem{Yi:2014MIMO}
X.~Yi, S.~Yang, D.~Gesbert, and M.~Kobayashi, ``The degrees of freedom region
  of temporally correlated {MIMO} networks with delayed {CSIT},'' \emph{IEEE
  Trans. Inf. Theory}, vol.~60, no.~1, pp. 494--514, Jan. 2014.

\bibitem{Jafar:BIA}
S.~Jafar, ``Blind interference alignment,'' \emph{IEEE J. Sel. Topics in Signal
  Processing}, vol.~6, no.~3, pp. 216--227, 2012.

\bibitem{Alternating}
R.~Tandon, S.~Jafar, S.~Shamai~Shitz, and H.~Poor, ``On the synergistic
  benefits of alternating {CSIT} for the {MISO} broadcast channel,'' \emph{IEEE
  Trans. Inf. Theory}, vol.~59, no.~7, pp. 4106--4128, 2013.

\bibitem{Lee:2012}
N.~Lee and R.~Heath~Jr, ``Not too delayed {CSIT} achieves the optimal degrees
  of freedom,'' in \emph{Proc.~Allerton 2012, arXiv:1207.2211}, 2012.

\bibitem{Slock:2012}
Y.~Lejosne, D.~Slock, and Y.~Yuan-Wu, ``Degrees of freedom in the {MISO BC}
  with delayed-{CSIT} and finite coherence time: A simple optimal scheme,'' in
  \emph{2012 IEEE International Conference on Signal Processing, Communication
  and Computing (ICSPCC)}.\hskip 1em plus 0.5em minus 0.4em\relax IEEE, 2012,
  pp. 180--185.

\bibitem{Caire:2010Rethinking}
G.~Caire, S.~A. Ramprashad, and H.~C. Papadopoulos, ``Rethinking network
  {MIMO}: Cost of {CSIT}, performance analysis, and architecture comparisons,''
  in \emph{ITA2010}, 2010.

\bibitem{Jafar:noCSIT}
S.~A. Jafar and A.~J. Goldsmith, ``Isotropic fading vector broadcast channels:
  The scalar upper bound and loss in degrees of freedom,'' \emph{IEEE Trans.
  Inf. Theory}, vol.~51, no.~3, pp. 848--857, 2005.

\bibitem{Guo:2012noCSIT}
Y.~Zhu and D.~Guo, ``The degrees of freedom of isotropic {MIMO} interference
  channels without state information at the transmitters,'' \emph{IEEE Trans.
  Inf. Theory}, vol.~58, no.~1, pp. 341--352, Jan. 2012.

\bibitem{Huang:2012noCSIT}
C.~Huang, S.~Jafar, S.~Shamai, and S.~Vishwanath, ``On degrees of freedom
  region of {MIMO} networks without channel state information at
  transmitters,'' \emph{IEEE Trans.~Inf.~Theory}, vol.~58, no.~2, pp. 849--857,
  Feb.~ 2012.

\bibitem{Vaze:noCSIT}
C.~Vaze and M.~Varanasi, ``The degree-of-freedom regions of {MIMO} broadcast,
  interference, and cognitive radio channels with no {CSIT},'' \emph{IEEE
  Trans. Inf. Theory}, vol.~58, no.~8, pp. 5354--5374, Aug. 2012.

\bibitem{graph_scheduling}
V.~Aggarwal, A.~Avestimehr, and A.~Sabharwal, ``On achieving local view
  capacity via maximal independent graph scheduling,'' \emph{IEEE Trans. Inf.
  Theory}, vol.~57, no.~5, pp. 2711--2729, 2011.

\bibitem{Jafar:CBIA}
S.~A. Jafar, ``Elements of cellular blind interference alignment---aligned
  frequency reuse, wireless index coding and interference diversity,''
  \emph{arXiv preprint arXiv:1203.2384}, 2012.

\bibitem{Jafar:2013TIM}
------, ``Topological interference management through index coding,''
  \emph{IEEE Trans. Inf. Theory}, vol.~60, no.~1, pp. 529--568, Jan. 2014.

\bibitem{Avestimehr:2013TIM}
N.~Naderializadeh and A.~S. Avestimehr, ``Interference networks with no {CSIT}:
  Impact of topology,'' \emph{submitted to IEEE Trans.~Inf.~Theory,
  arXiv:1302.0296}, Feb. 2013.

\bibitem{Jafar:1D}
H.~Maleki and S.~Jafar, ``Optimality of orthogonal access for one-dimensional
  convex cellular networks,'' \emph{IEEE Communications Letters}, vol.~17,
  no.~9, pp. 1770--1773, Sept. 2013.

\bibitem{Gou:2013TIM}
T.~Gou, C.~da~Silva, J.~Lee, and I.~Kang, ``Partially connected interference
  networks with no {CSIT}: Symmetric degrees of freedom and multicast across
  alignment blocks,'' \emph{IEEE Communications Letters}, vol.~17, no.~10, pp.
  1893--1896, 2013.

\bibitem{Sun:2013TIM}
H.~Sun, C.~Geng, and S.~A. Jafar, ``Topological interference management with
  alternating connectivity,'' in \emph{IEEE International Symposium on
  Information Theory Proceedings (ISIT)}, 2013.

\bibitem{Sezgin:TIM}
S.~Gherekhloo, A.~Chaaban, and A.~Sezgin, ``Topological interference management
  with alternating connectivity: The {Wyner}-type three user interference
  channel,'' \emph{arXiv preprint arXiv:1310.2385}, 2013.

\bibitem{Geng:TIM}
C.~Geng, H.~Sun, and S.~A. Jafar, ``Multilevel topological interference
  management,'' in \emph{Information Theory Workshop (ITW), 2013 IEEE}, 2013,
  pp. 1--5.

\bibitem{ISCOD}
Y.~Birk and T.~Kol, ``Informed-source coding-on-demand {(ISCOD)} over broadcast
  channels,'' in \emph{Proc.~INFOCOM'98, Seventeenth Annual Joint Conference of
  the IEEE Computer and Communications Societies}, vol.~3, 1998, pp.
  1257--1264.

\bibitem{Index2011}
Z.~Bar-Yossef, Y.~Birk, T.~S. Jayram, and T.~Kol, ``Index coding with side
  information,'' \emph{IEEE Trans. Inf. Theory}, vol.~57, no.~3, pp.
  1479--1494, March 2011.

\bibitem{Jafar:2012Index}
H.~Maleki, V.~Cadambe, and S.~Jafar, ``Index coding: An interference alignment
  perspective,'' in \emph{2012 IEEE International Symposium on Information
  Theory Proceedings (ISIT)}, 2012, pp. 2236--2240.

\bibitem{Sun:2013Index}
H.~Sun and S.~A. Jafar, ``Index coding capacity: How far can one go with only
  {Shannon} inequalities?'' \emph{arXiv preprint arXiv:1303.7000}, 2013.

\bibitem{Kim:2013Index}
F.~Arbabjolfaei, B.~Bandemer, Y.-H. Kim, E.~Sasoglu, and L.~Wang, ``On the
  capacity region for index coding,'' in \emph{IEEE International Symposium on
  Information Theory Proceedings (ISIT)}, July 2013, pp. 962--966.

\bibitem{IndexLP}
A.~Blasiak, R.~Kleinberg, and E.~Lubetzky, ``Index coding via linear
  programming,'' \emph{arXiv preprint arXiv:1004.1379}, 2010.

\bibitem{CoMP2012}
V.~S. Annapureddy, A.~El~Gamal, and V.~V. Veeravalli, ``Degrees of freedom of
  interference channels with {CoMP} transmission and reception,'' \emph{IEEE
  Trans. Inf. Theory}, vol.~58, no.~9, pp. 5740--5760, 2012.

\bibitem{ICCoMP2012}
A.~E. Gamal, V.~S. Annapureddy, and V.~V. Veeravalli, ``Interference channels
  with coordinated multi-point transmission: Degrees of freedom, message
  assignment, and fractional reuse,'' \emph{submitted to IEEE
  Trans.~Inf.~Theory, arXiv preprint arXiv:1211.2897}, 2012.

\bibitem{partcoloring}
G.~Li and R.~Simha, ``The partition coloring problem and its application to
  wavelength routing and assignment,'' in \emph{Proceedings of the First
  Workshop on Optical Networks}.\hskip 1em plus 0.5em minus 0.4em\relax
  Citeseer, 2000, p.~1.

\bibitem{selcoloring}
M.~Demange, T.~Ekim, B.~Ries, and C.~Tanasescu, ``On some applications of the
  selective graph coloring problem,'' \emph{European Journal of Operational
  Research}, 2014.

\bibitem{Gou:Compound}
T.~Gou, S.~Jafar, and C.~Wang, ``On the degrees of freedom of finite state
  compound wireless networks,'' \emph{IEEE Trans.~Inf.~Theory}, vol.~57, no.~6,
  pp. 3286--3308, Jun. 2011.

\bibitem{DynamicIndex}
M.~Neely, A.~Tehrani, and Z.~Zhang, ``Dynamic index coding for wireless
  broadcast networks,'' \emph{IEEE Trans. Inf. Theory}, vol.~59, no.~11, pp.
  7525--7540, Nov. 2013.

\bibitem{GT}
J.~A. Bondy and U.~S.~R. Murty, \emph{Graph theory with applications}.\hskip
  1em plus 0.5em minus 0.4em\relax Macmillan London, 1976, vol. 290.

\bibitem{Graph}
D.~B. West \emph{et~al.}, \emph{Introduction to graph theory}.\hskip 1em plus
  0.5em minus 0.4em\relax Prentice hall Englewood Cliffs, 2001, vol.~2.

\bibitem{Fractional2011}
E.~R. Scheinerman and D.~H. Ullman, \emph{Fractional graph theory}.\hskip 1em
  plus 0.5em minus 0.4em\relax DoverPublications. com, 2011.

\bibitem{Koetter}
R.~Koetter, M.~Effros, and M.~Medard, ``A theory of network equivalence -- part
  {I}: Point-to-point channels,'' \emph{IEEE Trans.~Inf.~Theory}, vol.~57,
  no.~2, pp. 972--995, Feb 2011.

\end{thebibliography}
%\bibliographystyle{IEEEtran}


\end{document}